%% file: main-arxiv.tex
\documentclass[11pt]{article}
\usepackage[margin=1in,letterpaper]{geometry}
\usepackage[T1]{fontenc}
\usepackage[sc]{mathpazo}
\usepackage{array}
\usepackage{microtype}
\usepackage{amsmath}
\usepackage[dvipsnames]{xcolor}
\usepackage{amssymb}
\usepackage{amsthm}
\usepackage{bm,float}
\usepackage{dsfont}
\usepackage{authblk}
\usepackage{fullpage,caption,wrapfig}
\usepackage{comment}
\usepackage{mathtools}
\usepackage[shortlabels,inline]{enumitem}
\usepackage{bbm}
\usepackage{booktabs}
\usepackage{makecell}
\definecolor{pku-red}{RGB}{139,0,18}
\usepackage{nag,tikz}
\usetikzlibrary{calc}

\usepackage{tablefootnote}
\usetikzlibrary{decorations.pathreplacing}
\usetikzlibrary{graphs, positioning, quotes, shapes.geometric}
\usepackage{pgfplots}
\pgfplotsset{compat=1.18}
\usetikzlibrary{arrows.meta}%
\usetikzlibrary{arrows}
\usepackage{tikz-3dplot}
\usepgfplotslibrary{colormaps,fillbetween}

\usepackage{multicol}
\usepackage[scaled]{helvet} 

\usepackage{float}

\usepackage{multirow}

\usepackage[normalem]{ulem}

\usepackage{etoolbox}
\makeatletter
\patchcmd{\@bibitem}{\ignorespaces}{\label{bib-#1}\ignorespaces}{}{}
\makeatother

\newcommand{\mycite}[1]{[\ref{bib-#1}]}
\newcommand{\titlecite}[1]{\texorpdfstring{\cite{#1}}{\mycite{#1}}}

\usepackage[ruled]{algorithm2e} 

\SetAlFnt{\small}
\SetAlCapFnt{\small}
\SetAlCapNameFnt{\small}
\SetAlCapHSkip{0pt}
\IncMargin{-\parindent}

\usepackage{hyperref}
\usepackage{crossreftools}
\hypersetup{
    colorlinks=true,
    linkcolor=pku-red,     
    urlcolor=violet,
    citecolor=BlueViolet,
    pdffitwindow=true,
}\usepackage[capitalize, nameinlink]{cleveref}

\allowdisplaybreaks[3]

\theoremstyle{plain}
\newtheorem{theorem}{Theorem}[section]
\newtheorem{lemma}[theorem]{Lemma}
\newtheorem{corollary}[theorem]{Corollary}
\newtheorem{proposition}[theorem]{Proposition}

\newtheorem{claim}[theorem]{Claim}

\theoremstyle{definition}
\newtheorem{definition}[theorem]{Definition}
\newtheorem{remark}[theorem]{Remark}
\newtheorem{example}[theorem]{Example}

\DeclareMathOperator*{\suppmax}{suppmax}
\DeclareMathOperator*{\suppmin}{suppmin}
\DeclareMathOperator*{\argmin}{argmin}

\DeclareMathOperator*{\supp}{supp}
\DeclareMathOperator*{\br}{br}
\DeclareMathOperator*{\aff}{aff}
\DeclareMathOperator*{\Span}{Span}
\DeclareMathOperator*{\poly}{poly}

\newcommand{\R}{\mathbb{R}}

\newcommand{\norm}[1]{\left\lVert #1\right\rVert}

\newcommand{\RHS}{\mathrm{RHS}}
\newcommand{\LHS}{\mathrm{LHS}}
\newcommand{\Term}{\mathrm{Term}\ }
\newcommand{\T}{\mathsf{T}}

\newcommand \za {\textbf{a}}
\newcommand \zb {\textbf{b}}
\newcommand \zc {\textbf{c}}
\newcommand \zd {\textbf{d}}
\newcommand \ze {\textbf{e}}
\newcommand \zf {\textbf{f}}
\newcommand \zg {\textbf{g}}
\newcommand \zh {\textbf{h}}
\newcommand \zi {\textbf{i}}
\newcommand \zj {\textbf{j}}
\newcommand \zk {\textbf{k}}
\newcommand \zl {\textbf{l}}
\newcommand \zm {\textbf{m}}
\newcommand \zn {\textbf{n}}
\newcommand \zu {\textbf{u}}
\newcommand \zv {\textbf{v}}
\newcommand \zw {\textbf{w}}
\newcommand \zt {\textbf{t}}



\title{The Search-and-Mix Paradigm in Approximate Nash Equilibrium Algorithms}

\author[1]{Xiaotie Deng\thanks{\href{mailsto:xiaotie@pku.edu.cn}{xiaotie@pku.edu.cn}.}}

\author[2]{Dongchen Li\thanks{\href{mailsto:dongchen.li@connect.hku.hk}{dongchen.li@connect.hku.hk}.}}

\author[1]{Hanyu Li\thanks{\href{mailsto:lhydave@pku.edu.cn}{lhydave@pku.edu.cn}.}}
\affil[1]{Peking University}
\affil[2]{The Univerity of Hong Kong}

\date{}

\begin{document}
\hypersetup{pageanchor=false}
\maketitle 
\begin{abstract}
\input{section/abstract}
\end{abstract}
\thispagestyle{empty} 

\newpage
\tableofcontents
\newpage
\setcounter{page}{1}
\hypersetup{pageanchor=true}
\pagenumbering{arabic}

\input{section/introduction}
\input{section/preliminaries}
\input{section/search_and_mix}
\input{section/find_opt_comb}
\input{section/approx_analysis}
\input{section/mixing_problem}
\input{section/conclusion-and-discussion}

\bibliographystyle{plain}
\bibliography{ref}

\newpage
\appendix

\input{section/append/details}

\newpage
\input{section/append/detail_proof}
\newpage
\input{section/append/detail_example}
\newpage
\input{section/append/def_geo}

\end{document}

%% file: section/abstract.tex
AI in Math deals with mathematics in a constructive manner so that reasoning becomes automated, less laborious, and less error-prone. For algorithms, the question becomes how to automate analyses for specific problems. For the first time, this work provides an automatic method for approximation analysis on a well-studied problem in theoretical computer science: computing approximate Nash equilibria in two-player games. We observe that such algorithms can be reformulated into a search-and-mix paradigm, which involves a search phase followed by a mixing phase. By doing so, we are able to fully automate the procedure of designing and analyzing the mixing phase. For example, we illustrate how to perform our method with a program to analyze the approximation bounds of all the algorithms in the literature. Same approximation bounds are computed without any hand-written proof. Our automatic method heavily relies on the LP-relaxation structure in approximate Nash equilibria. Since many approximation algorithms and online algorithms adopt the LP relaxation, our approach may be extended to automate the analysis of other algorithms.

%% file: section/introduction.tex
\section{Introduction}

Automated theorem discovering and proving have been a long-standing goal of computer science and artificial intelligence, with the earliest work dating from the 1950s \cite{davisComputingProcedureQuantification1960,gilmoreProofMethodQuantification1960}. Some of the most successful examples include geometry \cite{recioAutomaticDiscoveryTheorems1999,wuMechanicalTheoremProving1994} and inequality proofs \cite{akbarpourMetiTarskiAutomaticTheorem2010,cavinessQuantifierEliminationCylindrical1998}. With the advancement of deep learning, people are also leveraging deep neural networks to discover theorems and proofs \cite{daviesAdvancingMathematicsGuiding2021,gukovSearchingRibbonsMachine2023,lampleHyperTreeProofSearch2022}.

In computer science, the parallel tasks involve the automatic discovery and analysis of algorithms. Using deep reinforcement learning, machines can discover algorithms for tasks like matrix multiplication \cite{fawziDiscoveringFasterMatrix2022} and sorting \cite{mankowitzFasterSortingAlgorithms2023} that outperform human-designed algorithms on fixed-size inputs.

On the other hand, over the course of history, the complexity of algorithm analyses has increased dramatically. Many of them are error-prone and could only be verified by a handful of specialists. Thus, there is an increasing demand for the automatic method for algorithm analysis, leaving a certain amount of complicated proof to programs. However, it remains largely unexplored even for a specific problem, such as analyzing the worst-case approximation bounds of a particular algorithm.

As a typical example, we note that all these drawbacks arose in a well-studied problem of approximation algorithms in theoretical computer science: computing \emph{approximate Nash equilibria} in two-player games \cite{BBM07_0.36NE,CDF+16_0.382NE&0.653WSNE,DMP07_0.38NE,DMP06_0.5NE,DFM22_0.3333NE,KPS06_0.75NE,TS07_0.3393NE}. As is shown in \Cref{tab:approx-result}, the progress of approximation improvement is becoming more and more incremental, while the length of the paper's proof is increasing. More seriously, the grow-up of proof complexity led to errors in the preliminary version of the SOTA result \cite{wrong_DFM_0.3333}. 

In this work, we provide an automatic tool to derive the approximation bound of a given algorithm without any hand-written proof, which is far less error-prone than humans. Our method also dramatically speeds up the analysis process. To obtain the same approximation bounds for \emph{all} algorithms in the literature for this problem, our automatic method takes only a few seconds\footnote{The only exception is \cite{BBM07_0.36NE}-2. Mathematica has problems dealing with the function \texttt{Sqrt} in the \texttt{NMaximize} function. Therefore, we use the \texttt{RandomSearch} method to calculate the bound of \cite{BBM07_0.36NE}-2, which takes several minutes.} on a common laptop, while the original proofs could span more than 5 pages (see \Cref{tab:approx-result}). Our work can thus be viewed as the first step towards automatic algorithm analysis.

\begin{table}[ht]
\centering
\begin{tabular}{c|cccc}
\toprule
    Algorithm & \makecell{Bound proved by\\ original paper} & \makecell{Length of proofs\\(page) / codes (line)} &\makecell{Bound computed\\
    by our program} & \makecell{Running time\\(second)}\\
        \midrule
\cite{KPS06_0.75NE}&$0.75$& $1/32$ &$0.75$&$1.39$\\
\cite{DMP06_0.5NE}&$0.5$& $0.5/20$ &$0.5$&$0.06$\\
\cite{DMP07_0.38NE}&$0.381966+\delta$& $1/33$&$0.381966+\delta$&$1.99$\\ 
\cite{BBM07_0.36NE}-1&$0.381966$& $1/25$&$0.381966$&$4.63$\\
\cite{CDF+16_0.382NE&0.653WSNE}&$0.381966$& $2/34$&$0.381966$&$3.38$\\                   
\cite{BBM07_0.36NE}-2& $0.363917$& $6/32$ & $0.363917$&$484.23$\\
\cite{TS07_0.3393NE}  &$0.339331+\delta$ & $7/25$&$0.339331+\delta$&$6.28$ \\
\cite{DFM22_0.3333NE}&$1/3+\delta$& $13/42$&$0.333333+\delta$&$15.67$\\
\bottomrule 
\end{tabular}
    \caption{The approximation bounds of \emph{all} algorithms in the literature: original vs. our results.}
    \label{tab:approx-result}
\end{table}

\emph{Nash equilibrium} (NE) is a fundamental concept in economics. In economic research, human interactions are often modeled as strategic games, where rational players interact with each other and attempt to achieve the best possible outcome. Nash equilibrium captures a state of balance among the players, where no player can benefit from unilaterally changing their strategy.

In his seminar works \cite{N51_NashConcept,NashKakutani}, Nash established the existence of such equilibria. However, from a theoretical view, it is also important to find polynomial-time algorithms for Nash equilibria in practice. In fact, from the perspective of computational complexity, the established hardness results \cite{Chen20053NASHIP, CDT09_2PPADcomplete,Daskalakis2005ThreePlayerGA} give strong evidence that even in two-player games, there could probably not exist polynomial-time algorithm for Nash equilibria.

This naturally leads to the consideration of $\epsilon$-\emph{approximate Nash equilibria} ($\epsilon$-NE), where the players are permitted to gain at most $\epsilon$ by deviation. In view of computational complexity, a similar negative result is given by Rubinstein \cite{R17_consthard}. It shows that, under a certain moderate assumption, computing $\epsilon$-NE \emph{cannot} be polynomial time when $\epsilon < \epsilon^*$, where $\epsilon^*$ is some constant.

Therefore, the main open problem remains to identify the smallest $\epsilon$ for which a polynomial-time algorithm can compute an $\epsilon$-NE. There have been numerous works in this direction, achieving increasingly better approximations, such as \cite{KPS06_0.75NE, DMP06_0.5NE, DMP07_0.38NE, BBM07_0.36NE,TS07_0.3393NE, DFM22_0.3333NE}. A process of lower bound and upper bound results is presented in \Cref{fig:process-fig}.

\begin{figure}[ht]
    \centering
    \input{figures/approx-process-line}
    \caption{The process on approximation NE in bimatrix games. The blue points are the upper bound results while the orange point is the lower bound result.}
    \label{fig:process-fig}
\end{figure}
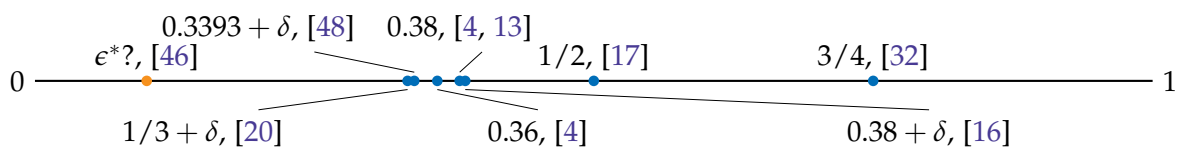

Our work focuses on the automatic method to derive an upper bound proven for an approximation algorithm. The main idea of our method is to restate the approximation algorithms into a \emph{search-and-mix} paradigm. Then, we provide a series of theories and algorithms that collectively assist in automating the design and analysis of any algorithm under this paradigm.

\subsection{Related work}
Nash equilibrium is a solution concept in which no player can gain more by deviating from her current (possibly randomized) strategy against her opponents. A typical scenario for this solution concept is finite normal-form games, where a finite number of players independently and simultaneously choose from a finite number of actions. By applying fixed-point theorems, Nash \cite{N51_NashConcept,NashKakutani} proved that such a solution always exists in any finite normal-form game, marking the start of a new paradigm in economics known as non-cooperative games. Nash equilibria subsequently became a standard solution concept in game theory, further enriched by contributions from computer science and artificial intelligence, such as algorithmic game theory~\cite{nisan_roughgarden_tardos_vazirani_2007} and adversarial learning~\cite{GAN}. As a result, game theory has become a thriving interdisciplinary field.

From the perspective of computer science, Nash's non-constructive proof was just the beginning. The further key question is how to find Nash equilibria in polynomial time. This question adopts the view from computational complexity. Papadimitriou \cite{P94_PPAD} introduced a complexity class called PPAD and demonstrated that the problem of computing Nash equilibria (\textsc{Nash}) lies in PPAD, so do $k$-player games ($k$-\textsc{Nash}).\footnote{Strictly speaking, in most cases where $k \geq 3$, Nash equilibria cannot be represented by Turing machines, a standard computational model. Therefore, ($k$-)\textsc{Nash} actually aims to find \emph{approximate Nash equilibria} with an approximation of $O(1/2^m)$ or even $O(1/\poly(m))$. See the introduction below for more details.} However, the establishment of PPAD-completeness results required more efforts. Daskalakis, Goldberg, and Papadimitriou~\cite{DGP06_nPPADcomplete} proved the completeness result for $k$-\textsc{Nash} with $k \geq 4$. Subsequently, two independent works by \cite{Chen20053NASHIP} and \cite{Daskalakis2005ThreePlayerGA} showed that $3$-\textsc{Nash} is PPAD-complete. Finally, Chen and Deng \cite{chenSettlingComplexityTwoPlayer2006} established the PPAD-completeness of $2$-\textsc{Nash}. These findings provide compelling evidence that unless $\text{PPAD} \subseteq \text{P}$, there is no polynomial-time algorithm\footnote{More accurately, no fully polynomial-time approximation scheme (FPTAS)} for Nash equilibria, even in the simplest two-player game scenarios. 

This naturally leads us to consider \emph{approximate Nash equilibria} with a constant approximation. In an $\epsilon$-Nash equilibrium ($\epsilon$-NE), any player can gain at most $\epsilon$ additional payoff by deviating from the original strategy. To have a unified discussion on $\epsilon$-NE, it is a convention in the literature to normalize the payoffs to the interval $[0,1]$. It is important to note that $\epsilon$-NE is not solely a computational consideration: It also characterizes the outcomes of players with limited rationality~\cite{dengComplexityCooperativeSolution1994,samuelson1997evolutionary}. 

Until now, most studies still have been focusing on two-player games, also known as bimatrix games, where the two players independently choose a row and a column from the payoff matrix. However, even this seemingly simple concept involves significant intricacy. 

Initially, using a straightforward argument in probabilistic methods, Lipton, Markarkis, and Mehta \cite{LMM03_quasipoly} proposed an algorithm to find an $\epsilon$-NE in time $n^{O(\log n/\epsilon^2)}$ (i.e., quasi-polynomial time). For a long time, it was conjectured that such a simple algorithm was not optimal. However, Rubinstein \cite{R17_consthard} proved that, under a certain moderate assumption\footnote{That is, the exponential-time hypothesis (ETH) for PPAD.}, there exists a constant $\epsilon^*$ such that computing $\epsilon$-NE is at least quasi-polynomial for any $\epsilon < \epsilon^*$.

Therefore, the main open problem remains to identify the smallest $\epsilon$ for which a polynomial algorithm can compute an $\epsilon$-NE. There have been numerous works in this direction. Kontogiannis, Panangopoulou, and Spirakis \cite{KPS06_0.75NE} introduced an approximation algorithm with $\epsilon=3/4$, which initiated the improvement process. Soon after, Daskalakis, Mehta, and Papadimitriou \cite{DMP06_0.5NE,DMP07_0.38NE} improved the approximation to $1/2$ and $0.38+\delta$, respectively. Bosse, Byrka, and Markakis \cite{BBM07_0.36NE} achieved $\epsilon=0.364$. Concurrently, Tasknakis and Spirakis \cite{TS07_0.3393NE} presented the famous TS algorithm that achieved $\epsilon=0.3393+\delta$, which remained the then state-of-the-art for 15 years. Research and discoveries continued... Czumaj et al. \cite{CDF+16_0.382NE&0.653WSNE} introduced the first distributed method with a $0.38$ approximation bound. Chen et al. \cite{CDH+21_0.3393tight} proved that $0.3393$ is the tight bound for the TS algorithm. Subsequently, in 2022, Deligkas, Fasoulakis, and Markakis \cite{DFM22_0.3333NE} showed that by carefully revising the TS algorithm, $\epsilon=1/3+\delta$ can be reached, marking it as the next milestone. Very recently, the tightness of their algorithmic $1/3$ bound was shown by Chen et al. \cite{chen_tightness_2023}.

\subsection{The traditional paradigm}

It is worth noting that all the polynomial-time approximation algorithms in the literature can be fitted into a general framework which we call the \emph{search-and-mix paradigm}. In the search phase, these algorithms \emph{search} in polynomial time for strategies of both players that satisfy specific properties. In the subsequent mixing phase, the algorithms give a few specific convex combinations to \emph{mix} the strategies obtained in the search phase and select the one with the minimum approximation. The mixing phase usually takes a negligible amount of time. Through the exploration of diverse strategies and convex combinations, researchers have been able to develop algorithms with various approximation bounds as mentioned above.

As a concrete example\footnote{If the readers are not familiar with concepts discussed below, see \Cref{sec:pre} for definitions.}, here we restate the algorithm given by Daskalakis, Mehta, and Papadimitriou~\cite{DMP06_0.5NE} in this paradigm. 

\begin{example}\label{ex:1/2NE}
In the search phase, select a pure strategy for one of two players, such as the $i$th row for the row player. Then, calculate the best-response column $j$ for the column player against row $i$ and then the best-response row $k$ for the row player against column $j$. In words, the result of the search is row player's strategies $i,k$ and column player's strategy $j$.

In the mixing phase, among all the possible convex combinations of the strategies, given by $\{(\alpha i+(1-\alpha)k, j):\alpha\in [0,1]\}$, select the strategy profile $((i+k)/2,j)$ as the output of the algorithm. By using the properties of best responses, it was shown that this strategy profile yields a $1/2$-NE.    
\end{example}

This paradigm follows a natural idea. In the search phase, it generates strategies with certain properties. However, using only a single strategy profile could lead to a poor approximation bound. To address this problem, it introduces a mixing phase that combines the generated strategies along with their properties. This combination produces new strategy profiles that one of them yields a better approximation bound.

However, this traditional approach also has some drawbacks. 
\begin{itemize}
    \item \emph{Practice}. The selection of the convex combinations emphasizes too much on the approximation bound. It weakens the performance of the algorithms in non-worst-case scenarios, which is considerably important in practice as these algorithms are widely applied in finding the approximate NE in polynomial time.
    
    \item \emph{Theory}. As algorithms with smaller approximation bounds are proposed, the improvements they achieved are becoming increasingly incremental. As there are more and more strategies generated in the search phase\footnote{One can refer to \Cref{tab:instances}.}, the design of convex combinations in the mixing phase is suffering from increasing intricacy and complexity. 
    
    Moreover, the approximation analysis is going to be complicated, lengthy, and error-prone. \Cref{tab:proof-length} shows the length of the approximation analysis of the original algorithms and the modified algorithms. These works made a very long proof but obtained a very incremental improvement in the approximation. Even worse, in the preliminary versions of the SOTA algorithm \cite{wrong_DFM_0.3333}, the proof is wrong.
    
\begin{table}[ht]
        \centering
        \begin{tabular}{c|cc}
        \toprule
             & Original Algorithm& Modified Algorithm \\\midrule
             Work & \cite{BBM07_0.36NE} &\cite{BBM07_0.36NE}\\
             Approximation Bound & $0.38$ & $0.36$\\
             Analysis Length (pages) & $1$ &$6$\\\hline
             Work & \cite{TS07_0.3393NE} &\cite{DFM22_0.3333NE}\\
             Approximation Bound& $0.3393+\delta$ & $1/3+\delta$\\
             Analysis Length (pages) & $7$ & $13$\\
             \bottomrule
        \end{tabular}
        \caption{Length of the approximation analysis of the original algorithms and the modified algorithms. The analysis of the DFM algorithm \cite{DFM22_0.3333NE} is based on that of the TS algorithm. Thus, it needs an additional $6$ pages to derive the approximation bound.}
        \label{tab:proof-length}
    \end{table}
\end{itemize}
Thus, it is crucial to develop a systematic understanding of how to design approximation algorithms and analyze their approximation bounds. Our approach is a novel attempt to overcome these long-standing limitations.

\subsection{The novel paradigm}

To begin with, our work presents a new formalization of the search-and-mix paradigm, establishing a framework for further analysis and understanding of such algorithms. The main observation, initially formalized by \cite{TS07_0.3393NE}, is that the approximation on a given strategy profile $(x,y)$ can be expressed by a function $f(x,y)$. From the definition of $\epsilon$-NE, if we set $f_R(x,y)$ and $f_C(x,y)$ to be the maximum additional payoff of the row and the column player by deviating from their current strategies, then $f(x,y)=\max\{f_R(x,y),f_C(x,y)\}$ is exactly the approximation of the strategy profile $(x,y)$. Thus, the mission of finding the minimum approximation can be transferred to solving an optimization problem with the objective function $f$.

The main difference of our formalization is in the mixing phase. As mentioned above, in the traditional approach, people select certain convex combinations and output the one with the minimum approximation. However, our mixing phase now directly finds the \emph{instance optimal} convex combination with the minimum approximation. This problem can be formalized as a global optimization problem, as discussed above.

To solve this optimization problem, we present a polynomial-time subroutine for our mixing phase. The main advantage of our subroutine is that it is \emph{search-phase independent}. In other words, it does not rely on the particular form of the search phase. Thus, it is suitable for \emph{all} approximation algorithms in the literature. It is both theoretically and practically better than the traditional approach.

For instance, in continuation of \Cref{ex:1/2NE}, instead of outputting the fixed strategy profile $((i+k)/2,j)$ as in \cite{DMP06_0.5NE}, our approach directly computes the optimal $\alpha^*$ with the minimum approximation among all strategy profiles in $\{(\alpha i+(1-\alpha)k, j):\alpha\in [0,1]\}$. Namely, we solve the following optimization problem
\[
\begin{array}{rl}
    \text{minimize}&\quad f(\alpha i+(1-\alpha)k, j)\\
    \text{ s.t. }&\quad \alpha\in [0,1].
\end{array}
\]

The main difficulty is that this optimization problem, in general cases, is non-convex and even non-smooth\footnote{That is, the objective function is non-differentiable.}. Thus, we cannot rely on traditional global optimization techniques based on iterative methods and gradients. Our solution is as follows. First, using classic methods in computational geometry, we introduce a divide-and-conquer method to partition the problem into several optimization subproblems to overcome the non-smoothness. Then, we use techniques from continuous optimization and discrete geometry to derive necessary conditions for the local minimum of the subproblems. Finally, for each subproblem, we find its feasible solutions satisfying the necessary conditions. The global minimum of the original problem is thus the optimal one among all found subproblem solutions.

Having obtained the subroutine for the new mixing phase, the immediate next problem is to analyze the approximation bound of the algorithm equipped with this subroutine. We develop a corresponding method for this purpose. Our solution is to convert the approximation analysis into a fixed-size constraint optimization problem. Moreover, this conversion can be accomplished in a fully automated manner. For this purpose, we divide the approximation analysis into two steps. The first step is search-phase independent (just like our mixing-phase subroutine), while the second only depends on the search phase.

In the first step, we show how to construct an upper bound for our mixing subroutine in a unified and search-phase-independent manner. Our main tool is an extension of the technique proposed in \cite{CDH+21_0.3393tight}. We call our extension the \emph{linearization} method. Using this technique, we linearize functions $f_R$ and $f_C$. By doing so, we can express the upper bound $h^*$ of the global minimum of $f$ solely by $f_R$ and $f_C$ values of strategies given by the search phase.

For instance, consider \Cref{ex:1/2NE} again. The result of linearization is presented in \Cref{fig:linearization}. The blue point is the global minimum of $f$, and the pink point is the upper bound. We can show that the upper bound of $\min_{\alpha\in [0,1]} f(\alpha i+(1-\alpha)k,j)$ is fully determined by $f_R(i,j)$, $f_C(i,j)$, $f_R(k,j)$, and $f_C(k,j)$, given by $h^*=\min_{\alpha\in [0,1]}\max\{\alpha f_R(i,j)+(1-\alpha) f_R(k,j), \alpha f_C(i,j)+(1-\alpha) f_C(k,j)\}$.

\begin{figure}[htb]
    \centering
    \hspace{4em}\input{figures/linearization}
    \caption{Linearization procedure.}
    \label{fig:linearization}
\end{figure}
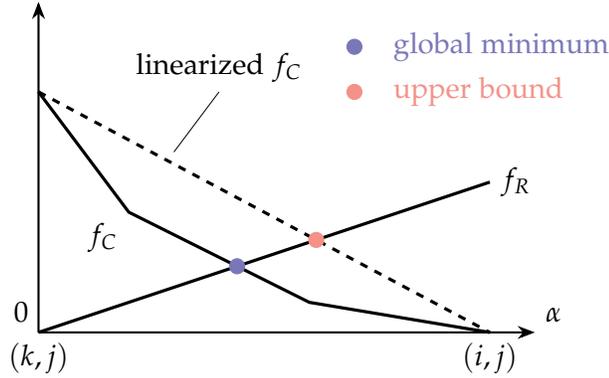

In the second step, we require the search phase to provide relations of $f_R$ and $f_C$ values of strategies given by it. Note that the upper bound $h^*$ derived in the first step is fully determined by these $f_R$'s and $f_C$'s. We view these $f_R$'s and $f_C$'s as variables and their relations as constraints, and maximize $h^*$, which then produces a fixed-size constraint optimization problem. No matter what values these $f_R$'s and $f_C$'s have, we can guarantee that the maximized $h^*$ is the desired constant upper bound. We can solve this optimization problem by a computer program.

For instance, in \Cref{ex:1/2NE}, denote $a=f_R(i,j)$, $b=f_R(k,j)$, $c=f_C(i,j)$, and $d=f_C(k,j)$. Since $j$ is the best response for the column player against $i$, we have $c=0$. Similarly, $b=0$. Besides, by the $[0,1]$ normalization assumption of the game payoffs, we have $0\leq a,d\leq 1$. Bearing with these relations, the expression of the approximation upper bound is given by:
\begin{equation}
\begin{aligned}
\max_{a,b,c,d} \min_{\alpha} & \quad h\\
    \text{s.t.} &\quad h =\max\{\alpha a+(1-\alpha)b,\alpha c+(1-\alpha) d\},\\
    &\quad 0\leq \alpha\leq 1,\\
    &\quad b=0, c=0,\\
    &\quad 0\leq a,d \leq 1.
\end{aligned}\label{eq:example-opt-prob}    
\end{equation}
The intuition behind this optimization is clear: By the definition of $h^*$, no matter what values these $f_R$ and $f_C$ have, the approximation should be bounded by $h^*$. Then, to compute the approximation bound, we need to consider the worst, namely the largest $h^*$. Thus, we have to maximize all these $f_R$'s and $f_C$'s.  It is not hard to show the optimal value to \eqref{eq:example-opt-prob} is $1/2$, which is the desired approximation bound.

To demonstrate the power of our new paradigm, we apply our mixing-phase subroutine to revise \emph{all} approximation algorithms in the literature and provide the corresponding approximation analysis. We present the details for the BBM algorithm \cite{BBM07_0.36NE}, the TS algorithm \cite{TS07_0.3393NE}, and the state-of-the-art DFM algorithm \cite{DFM22_0.3333NE}. To illustrate how to conduct the approximation analysis in an automated way, we also write down the optimization problem for the approximation bound similar to \eqref{eq:example-opt-prob} for \emph{all} algorithms in the literature. Subsequently, we implement Mathematica codes to calculate the corresponding approximation bounds, as shown in \Cref{tab:approx-result}. It is worth noting that most computations are completed within a few seconds and achieve very high accuracy!

It is worth noting that these algorithms have very different search phases. However, our mixing-phase subroutine can be applied directly without any modifications. Moreover, the approximation analysis for the revised algorithms now takes a unified approach and can be conducted completely by a computer program. The authors highly encourage the readers to look at \Cref{details:literature-examples} as well as the Mathematica codes\footnote{We have put them as well as the evaluation outputs (also shown in \Cref{tab:approx-result}) in an anonymous repository. Please click the following link \url{https://anonymous.4open.science/r/approxNE-auto-analyzer-D71E/}.} for this purpose.

\subsection{Main Contributions}

Our new paradigm enjoys the following advantages:
\begin{itemize}
    \item Our subroutine of the mixing phase is search-phase independent. Moreover, it finds the instance optimal strategy profile in the mixing phase. Thus, researchers no longer need to design ad hoc convex combinations with great efforts.
    
    \item Approximation analysis can be fully automated. The constructions of upper bounds are given by formal rules and do not depend on the search phase. Then, given suitable relations from the search phase, we can automatically write down the optimization problem of the approximation bound and solve it using a computer program. As a result, when designing new approximation algorithms, it is sufficient to only focus on the design of the search phase and find out the relations used in the analysis. This reduces a lot of workload.
    
    \item It leads to the first systematic approach to prove the tightness of the approximation bound. By examining the optimal solution to the optimization problem of the approximation bound, we can obtain conditions under which the approximation bound becomes tight. This helps generate tight instances and prove the tightness of the approximation bound.
    
    \item Using our paradigm, it is convenient to refine the search phase. When introducing parameters in the search phase, we can obtain the corresponding parameterized approximation bound. This helps in selecting the optimal parameters with the minimum approximation bound.
\end{itemize}

We are also aware of the value of the mixing problem itself. The mixing problem is closely related to quadratic constraint quadratic programming. It brings new perspectives on approximate Nash equilibria and a wide range of open problems concerning the relationships between PPAD, P, and NP. It turns out that the mixing problem serves a role in approximate Nash equilibria similar to the \textsc{Sat} problem in computational complexity theory.

Moreover, our linearization method is essentially a special kind of LP relaxations. In fact, many approximation algorithms and online algorithms heavily adopt the LP-relaxation, both for design (e.g., vertex cover and max-cut \cite{vaziraniApproximationAlgorithms2003}) and for analysis (e.g., factor revealing LP \cite{jainGreedyFacilityLocation2003,mehtaAdWordsGeneralizedOnline2007}).  our approach may be extended to automate the analysis of these kinds of algorithms.

\subsection{Organization}

This paper is organized as follows.
\begin{itemize}
    \item In \Cref{sec:pre}, we define the basic notions about Nash Equilibria.
    \item In \Cref{sec:search-and-mix}, we provide a formal definition of the mixing phase, together with our new mixing phase.
    \item In \Cref{sec:find-opt-comb}, we propose three polynomial-time algorithms to calculate the \emph{optimal} convex combinations of a certain number of strategies. They can serve as the desired subroutine for the new mixing phase.
    \item In \Cref{sec:approx-anal}, we propose upper bound constructors corresponding to algorithms in \Cref{sec:find-opt-comb}. This constructor is able to write the approximation analysis into a fixed-size constraint optimization problem via a systematic and automated approach.
\end{itemize}

As a concrete example, we demonstrate how our framework works on a concise but nontrivial $0.38$-approximation algorithm \cite{BBM07_0.36NE} (called the BBM algorithm) in \Cref{sec:search-and-mix}, \Cref{sec:find-opt-comb}, and \Cref{sec:approx-anal}. In \Cref{sec:mixing-problem}, we further discuss the mixing problem. Finally, in \Cref{sec:conclusion}, we summarize our work and propose further directions.

%% file: figures/approx-process-line.tex
\begin{tikzpicture}[x=0.09\textwidth,y=2em]
\draw [thick] (0,0) -- (10,0); 

\node[left] at (0,0) {$0$}; 

\node[right] at (10,0) {$1$}; 

\fill[BurntOrange] (1,0) circle (2pt); 
\node[above] at (1,0) {$\epsilon^*$?, \cite{R17_consthard} }; 

\fill[NavyBlue] (3.334,0) circle (2pt); 
\node[below] at (1.5,-0.5) {$1/3+\delta$, \cite{DFM22_0.3333NE} }; 
\draw [ultra thin] (3.334,-0.15) -- (2,-0.5);

\fill[NavyBlue] (3.394,0) circle (2pt); 
\node[above] at (2,0.5) {$0.3393+\delta$, \cite{TS07_0.3393NE} }; 
\draw [ultra thin] (3.394,0.15) -- (2.5,0.5);

\fill[NavyBlue] (3.6,0) circle (2pt); 
\node[below] at (4.5,-0.5) {$0.36$, \cite{BBM07_0.36NE} }; 
\draw [ultra thin] (3.6,-0.15) -- (4.5,-0.5);

\fill[NavyBlue] (3.8,0) circle (2pt); 
\node[above] at (3.8,0.5) {$0.38$, \cite{BBM07_0.36NE, CDF+16_0.382NE&0.653WSNE} }; 
\draw [ultra thin] (3.8,0.15) -- (4,0.5);

\fill[NavyBlue] (3.85,0) circle (2pt); 
\node[below] at (8,-0.5) {$0.38+\delta$, \cite{DMP07_0.38NE} }; 
\draw [ultra thin] (3.85,-0.15) -- (8,-0.5);

\fill[NavyBlue] (5,0) circle (2pt); 
\node[above] at (5,0) {$1/2$, \cite{DMP06_0.5NE} }; 

\fill[NavyBlue] (7.5,0) circle (2pt); 
\node[above] at (7.5,0) {$3/4$, \cite{KPS06_0.75NE} }; 

\end{tikzpicture}

%% file: figures/linearization.tex
\begin{tikzpicture}[scale=4]
    \draw[-Stealth, thick] (0,0) -- (1.1*1.5,0) node[above right] {$\alpha$};
    \draw[-Stealth, thick] (0,0) -- (0,1.1);
    \draw[very thick] (0,0) -- (1*1.5,0.5) node[right] {$f_R$};
    \draw[very thick] (0,0.8) -- (0.2*1.5,0.4) node [below left] {$f_C$} -- (0.6*1.5,0.1) -- (1*1.5,0);
    \draw[very thick, dashed] (0,0.8) -- (1*1.5,0);
    \draw[ultra thin] (0.3*1.5,0.6) -- (0.4*1.5,0.8) node[above] {linearized $f_C$};
    
    \draw (0,0) node [above left] {$0$};
    \draw (0,0) node[below] {$(k,j)$};
    \draw (1*1.5,0) node[below] {$(i,j)$};
    
    \draw[fill,Periwinkle] (0.44*1.5,0.22) circle [radius=0.025];
    \draw[fill,Periwinkle] (0.7*1.5,0.95) circle [radius=0.025] node[right] {\quad global minimum};
    
    \draw[fill,Salmon] (0.8/1.3*1.5,0.4/1.3) circle [radius=0.025];
    \draw[fill,Salmon] (0.7*1.5,0.8) circle [radius=0.025] node[right] {\quad upper bound};
\end{tikzpicture}

%% file: section/preliminaries.tex
\section{Preliminaries}\label{sec:pre}
\subsubsection*{Vectors and matrices}
The $n$-dimensional Euclidean space is denoted by $\R^n$. The standard orthonormal basis of $\R^n$ is $e_1,\dots,e_n$. Denote by $\Span(v_1,\dots,v_k)$ the linear space spanned by vectors $v_1,\dots,v_k\in\R^n$. Notation $[n]:=\{1,\dots,n\}$ represents an index set. For vector $v\in\R^n$, denote its $i$th item by $v_i$. For vector $u\in\R^n$, define the following operators: $\max\{u\}:=\max\{u_1,\dots,u_n\}$, $\min\{u\}:=\min\{u_1,\dots,u_n\}$, $\supp(u):=\{i\in[n]:u_i\neq 0\}$, $\suppmax(u):=\{i\in[n]: \forall j\in[n]\ u_{i} \geq u_{j}\}$, and $\suppmin(u):=\{i\in[n]: \forall j\in[n]\ u_{i} \leq u_{j}\}$. For vector $u\in\R^n$, denote its Euclidean norm by $\norm{u}$. For two vectors $v,w\in\R^n$, notation $v\geq w$ ($v>w$) represents that $v_i\geq w_i$ ($v_i> w_i$) holds for every $i\in[n]$.

For an $m\times n$ matrix $A$, denote its $i$th row by $A_i$, its $j$th column by $A^j$, and its item at $i$th row $j$th column by $A_{ij}$. Its transpose is denoted by $A^\T$.

\subsubsection*{Simplex and mixing operations}

A standard \emph{$(n-1)$-simplex} is the set $\Delta_n:=\{\alpha \in\R^n:\alpha \geq 0\text{ and }\sum_{i=1}^n \alpha_i=1\}$. A simplex can be viewed as a probability space and its elements are probability vectors. 

For given elements $z_1,\dots ,z_w$ from $\R^n$, the set of their \emph{convex combinations} is defined to be $\{\alpha_1 z_1+\dots +\alpha_w z_w: \alpha\in \Delta_w\}$. For any element in the set, the corresponding $\alpha$ is called its \emph{barycentric coordinate}. Any vector $\alpha\in \Delta_w$ uniquely determines a convex combination, which can be viewed as an operation:
\[
\alpha(z_1,\dots, z_w):=\alpha_1 z_1+\dots +\alpha_w z_w.
\]
We call it a \emph{mixing operation} over $(z_1,\dots, z_w)$ with $\alpha$. Under this situation, $\alpha$ is called the \emph{mixing coefficient}.

\subsubsection*{Games, mixed strategies, and best responses}
We focus on \emph{bimatrix games}, in which there are only two players. We refer to them as the row player and the column player. We assume that the row player has $m$ pure strategies and the column player has $n$ pure strategies. A game can be defined by a pair of payoff matrices $R$ and $C$ of size $m\times n$. When the row player chooses the $i$th row and the column player chooses the $j$th column, their payoffs are denoted by $R_{ij}$ and $C_{ij}$, respectively. By adjusting the payoffs through shifting and scaling, we assume that all entries in $R$ and $C$ belong to interval $[0,1]$.

For each player, a \emph{mixed strategy} represents a probability distribution over all pure strategies. Therefore, the space of mixed strategies over $n$ pure strategies is precisely captured by $\Delta_n$. We also use vectors $x\in \Delta_m$ and $y\in \Delta_n$ to represent mixed strategies for the row and the column player, respectively. In particular, pure strategies can be regarded as mixed strategies where a specific pure strategy is chosen with a probability of $1$. Unless otherwise specified, whenever we refer to "strategy" below, we are referring to mixed strategies.

A \emph{strategy profile} $(x,y)$ refers to a pair of mixed strategies $x$ and $y$ from the row and column players, respectively. We also call it a \emph{strategy pair}.

The \emph{payoff} of a strategy profile is defined as the expectation of the payoff over all pure strategy profiles. Consider the strategy profile $(x,y)$, where the row player selects $x$ and the column player selects $y$ independently. Then their payoffs are $x^\T Ry$ and $x^\T Cy$ respectively.

A \emph{best response} against a strategy $x$ from the row player is a mixed strategy of the column player that maximizes the expected payoff against $x$. We denote the set of best response strategies by $\br_{C}(x)$. Similarly, we can define $\br_{R}(y)$ as the set of best response strategies against the strategy $y$ of the column player. By definition, for a strategy pair $(x,y)$, $\br_{C}(x)$ is the set of all probability distributions on $\suppmax(C^\T x)$, and $\br_{R}(y)$ is the set of all probability distributions on $\suppmax(Ry)$.

\subsubsection*{Approximate Nash equilibria and the optimization viewpoint}
An \emph{$\epsilon$-approximate Nash equilibrium} ($\epsilon$-NE) is a strategy profile $(x,y) \in \Delta_m \times \Delta_n$ such that for any $x'\in \Delta_m$ and $y'\in \Delta_n$, it holds that
\begin{align*}
(x')^\T Ry \leq x^\T Ry + \epsilon, \text{\quad and\quad} x^\T Cy' \leq x^\T Cy + \epsilon.
\end{align*}
We call the minimum $\epsilon$ the (additive) \emph{approximation}\footnote{In the literature, this quantity has various names. Here we use the one we think the most appropriate.} of the profile. In particular, a strategy profile is a Nash equilibrium if it is a $0$-NE (or has an approximation of zero). The famous Nash's theorem \cite{N51_NashConcept} states that there exists a Nash equilibrium for any bimatrix game.\footnote{Actually, this holds for any finite normal-form game.}

We follow \cite{TS07_0.3393NE} to express the problem of computing $\epsilon$-NE in an optimization form. Define the regrets of the row and the column player as follows:
\begin{align*}
f_R(x,y):=\max\{Ry\}-x^\T  R y \text{\quad and\quad}
f_C(x,y):=\max\{C^\T  x\}-x^\T  C y.
\end{align*}
Note that $\max \{R y\}=\max_{x'\in \Delta_n} (x')^\T R y$. Similarly, we have $\max\{C^\T  x\}=\max_{y'\in \Delta_m} x^\T C y'$. Therefore, both $f_R$ and $f_C$ are non-negative, and a strategy profile $(x,y)$ is an $\epsilon$-NE if and only if 
\[\max \{f_R(x,y),f_C(x,y)\} \leq \epsilon.\] 

Define $f(x,y):=\max\{f_R(x,y),f_C(x,y)\}$. Then finding an $\epsilon$-NE is equivalent to finding a point $(x,y)$ in $\Delta_m\times\Delta_n$ with $f(x,y)\leq \epsilon$. Thus, it suffices to minimize $f$ in order to find solutions with a certain approximation.

%% file: section/search_and_mix.tex
\section{The search-and-mix paradigm}\label{sec:search-and-mix}
In this section, we give the formalization of the search-and-mix paradigm, including both traditional and novel ones. 

The traditional paradigm can be divided into two phases. In the search phase, an algorithm computes several (at most three in the literature) strategies of each player in polynomial time. In the mixing phase, the algorithm then mixes the selected strategies into several strategy profiles and outputs the profile with the minimum $f$ value. The illustration is presented in \Cref{fig:framework-search-mix}.
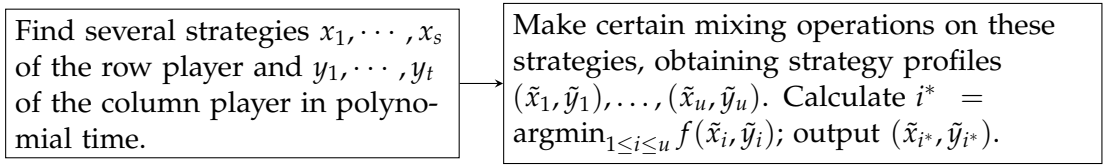
\begin{figure}[ht]
    \centering
\input{figures/framework-search-mix}
\caption{Procedure of the search-and-mix paradigm in the literature.}
\label{fig:framework-search-mix}
\end{figure}

A typical example is as follows.

\begin{example}[BBM algorithm~\cite{BBM07_0.36NE}]\label{example:start}
~
\begin{itemize}[fullwidth]
    \item \emph{Search phase}: Compute a NE $(x^*, y^*)$ of the zero-sum game $(R-C, C-R)$\footnote{Computing NE in zero-sum games can be modeled by a linear program and thus can be solved in polynomial time. See, e.g., \cite{nisan_roughgarden_tardos_vazirani_2007}.}. Let $g_1=f_R(x^*,y^*)$ and $g_2=f_C(x^*,y^*)$. By symmetry, assume without loss of generality that $g_1\geq g_2$. Then compute $r_1\in \br_R(y^*)$ and $b_2\in \br_C(r_1)$.
    \item \emph{Mixing phase}: Produce strategy profiles $(x^*,y^*)$ and $(r_1,(1-\delta_2) y^*+\delta_2 b_2)$, where $\delta_2=(1-g_1)/(2-g_1)$. Output the one with the smaller $f$ value.
\end{itemize}
\end{example}

Throughout the paper, we use this example repeatedly to demonstrate how to use our framework in general.

\Cref{tab:instances} shows how approximation algorithms in the literature fit into this paradigm. Our purpose is merely to demonstrate how these algorithms can be presented as search-and-mix instances. Thus, we omit the detailed meanings of the notations but keep them the same as those in the original papers.

\begin{table}[ht]
\renewcommand{\arraystretch}{1.25}
\centering
\begin{tabular}{cccc}
\toprule
\multirow{2}*{Approximation}                                                        & \multicolumn{2}{c}{Search Phase}      & \multirow{2}*{Mixing Phase}   \\ \cmidrule{2-3}
& Row         & Column    &   \\
\midrule

$0.75$ \cite{KPS06_0.75NE}                             & $e_{i_1}, e_{i_2}$  & $e_{j_1},e_{j_2}$   & $(\frac{1}{2}(e_{i_1}+e_{i_2}),\frac{1}{2}(e_{j_1}+e_{j_2}))$                                                 \\
$0.5$ \cite{DMP06_0.5NE}                             & $e_i, e_k$        & $e_j$             & $(\frac{1}{2}(e_i+e_k), e_j)$                                                                             \\
$0.38+\delta$ \cite{DMP07_0.38NE}                            & $x,\alpha$        & $y,\beta$         & $(\delta \alpha+(1-\delta)x, \delta \beta+(1-\delta) y), (x,y)$                                           \\
$0.38$ \cite{BBM07_0.36NE}       & $x^*, r_1$        & $y^*,b_2$         & $(x^*,y^*), (r_1, (1-\delta_2)y^*+\delta_2 b_2)$                                                          \\
$0.38$ \cite{CDF+16_0.382NE&0.653WSNE}                         & $x^*,
\hat{x}, r$        & $y^*, j$         & $(\hat{x},y^*), (\frac{1}{2-v_r} x^*+\frac{1-v_r}{2-v_r}r, j)$                                           \\
$0.36$ \cite{BBM07_0.36NE}       & $\hat{x}$       & $y^*,b_2$         & $(\hat{x},(1-\delta_2)y^*+\delta_2 b_2)$     
\\
$0.3393+\delta$ \cite{TS07_0.3393NE}  & $x^*, w^*$        & $y^*, z^*$        & $(x^*, y^*), (\frac{1}{1+\lambda-\mu} w^{*}+\frac{\lambda-\mu}{1+\lambda-\mu} x^{*}, z^{*})$ \\
\specialrule{0pt}{2pt}{2pt}
$1/3+\delta$ \cite{DFM22_0.3333NE}& $x_s, w$ & $y_s, z, \hat{z}$ & \makecell{$(x_s,y_s),(w,z)$, \\
$(w,pz+(1-p)\hat{z}),(\frac{1-q}{2}x_s+\frac{1+q}{2}w,z)$}     \\
\bottomrule 
\end{tabular}
\bigskip
\caption{The search-and-mix instances in the literature. The notations are the same as those in the original papers. If there are symmetric cases, we only provide one for conciseness.}
\label{tab:instances}
\end{table}

As is shown by \Cref{tab:instances}, different algorithms produce different strategies in the search phase. Then, in the mixing phase, they choose different strategy profiles obtained from mixing operations in order to ensure a certain approximation in worst cases.

However, we notice that the traditional approach has the following two defects:
\begin{enumerate}    
    \item \emph{Ad hoc designs}. The intentional design of the mixing phase can undermine the practical applicability of an algorithm. In order to minimize the approximation bound, strategies are often chosen based on worst-case scenarios. However, it is important to note that worst-case scenarios take only a small fraction of all cases. In Appendix C.1 in \cite{CDH+21_0.3393tight}, two million game instances were randomly generated, but only two of them approached even one third of the approximation bound. Consequently, from a practical standpoint, it cannot be justified to deliberately design an algorithm solely for worst-case scenarios.

    \item \emph{Theoretical hardness}. 
   In the literature, there is increasing sophistication in the design of mixing strategies, involving a larger number of strategies and cases, longer proofs (see \Cref{tab:proof-length}) but incremental improvements.
\end{enumerate}

In this paper, we develop a new approach that reshapes the procedure of the mixing phase. Instead of designing certain strategy profiles for the worst case, we develop polynomial-time algorithms to directly compute the \emph{global minimum} of $f$. Now, the procedure of algorithm design is simplified to \Cref{fig:framework-new}.

\begin{figure}[ht]
    \centering
\input{figures/framework-new}
\caption{The novel procedure for the search-and-mix paradigm.}
\label{fig:framework-new}
\end{figure}
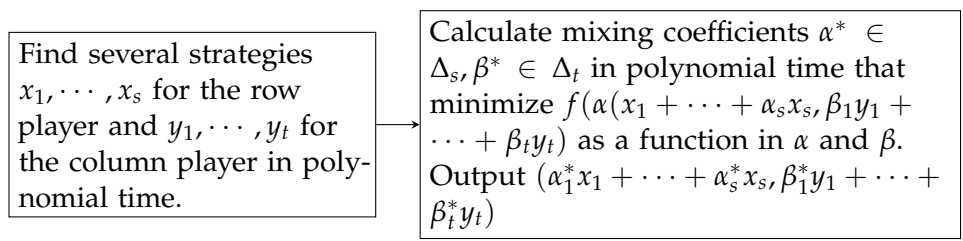

\begin{example}[The modified BBM algorithm]
\label{example:modify}
Applying the novel procedure to the mixing phase in \Cref{example:start}, we can obtain the modified mixing phase as follows.
\begin{itemize}[fullwidth]
    \item \emph{Modified mixing phase}: Calculate coefficients $\alpha^* \in \Delta_2$, $\beta^* \in \Delta_2$ in polynomial time that minimize $f(\alpha_1 x^*+\alpha_2 r_1, \beta_1 y^*+\beta_2 b_2)$ as a function of $\alpha_1$, $\alpha_2$, $\beta_1$, and $\beta_2$.
\end{itemize}
\end{example}

Following this novel procedure, we surprisingly resolve the two defects together of the above traditional approach:

\begin{enumerate}
    \item \emph{One unified design serving for all search phases.} In \Cref{sec:find-opt-comb}, we propose a polynomial-time subroutine for the new mixing phase. The mixing phase is not related to the particular form of the search phase. Moreover, it directly computes the global minimum. Thus, in our new approach, ad hoc designs for the subroutine of the mixing phase are no longer needed. In this way, people now only need to focus on designing the search phase. 
    
    \item \emph{Automated approximation analysis}. In \Cref{sec:approx-anal}, we provide a corresponding approximation analysis method for the novel paradigm. This method is suitable for all algorithms following the novel paradigm. Moreover, the analysis can be implemented in a fully automated manner. In this way, when designing a new approximation algorithm, people can now avoid sophisticated analysis and compute the approximation bound completely using a computer program.
\end{enumerate}

%% file: figures/framework-search-mix.tex
\begin{tikzpicture}[node distance=1.5em,>=stealth]
  \node[draw, text width = 15em]                         (step 1)  {Find several strategies $x_1,\cdots, x_s$ of the row player and $y_1,\cdots ,y_t$ of the column player in polynomial time.};
  \node[draw, text width = 20em, right=of step 1]        (step 2)  {Make certain mixing operations on these strategies, obtaining strategy profiles $(\tilde{x}_1, \tilde{y}_1), \ldots, (\tilde{x}_u, \tilde{y}_u)$. Calculate $i^* = \argmin_{1 \leq i \leq u} f(\tilde{x}_i, \tilde{y}_i)$; output $(\tilde{x}_{i^*}, \tilde{y}_{i^*})$.};
  \graph{
    (step 1) -> (step 2)
  };
\end{tikzpicture}

%% file: figures/framework-new.tex
\begin{tikzpicture}[node distance=1.5em,>=stealth]
  \node[draw, text width = 12em]                         (step 1)  {Find several strategies $x_1,\cdots, x_s$ for the row player and $y_1,\cdots ,y_t$ for the column player in polynomial time.};
  \node[draw, text width = 18em, right=of step 1]        (step 2)  {Calculate mixing coefficients $\alpha^*\in \Delta_s, \beta^* \in \Delta_t$ in polynomial time that minimize $f(\alpha(x_1+\cdots +\alpha_s x_s, \beta_1 y_1+\cdots +\beta_t y_t)$ as a function in $\alpha$ and $\beta$. Output $(\alpha_1^* x_1+\cdots +\alpha_s^* x_s, \beta_1^* y_1+\cdots +\beta_t^* y_t)$};
  \graph{
    (step 1) -> (step 2)
  };
\end{tikzpicture}

%% file: section/find_opt_comb.tex
\section{Algorithms for the mixing phase}\label{sec:find-opt-comb}
In this section, we state our problem in the mixing phase and present a sketch idea on how our algorithms are devised. To begin with, we formally define the problem for the new mixing phase illustrated in \Cref{fig:framework-new}.

\begin{definition}[$(s,t)$-mixing problems and algorithms]\label{def:st-mixing}
An $(s,t)$-mixing problem has the following input and output.
\begin{itemize}
    \item \emph{Input}: a game $(R,C)$, mixed strategies $x_1,\dots ,x_s$ for the row player and $y_1,\dots ,y_t$ for the column player.
    \item \emph{Output}: coefficients $\alpha\in \Delta_s$, $\beta \in \Delta_t$ that minimize $f(\alpha_1 x_1+\dots +\alpha_s x_s, \beta_1 y_1+\dots +\beta_t y_t)$.\footnote{Due to the definition, function $f$ is uniquely determined by $R,C$.}
\end{itemize}

An $(s,t)$-mixing algorithm solves an $(s,t)$-mixing problem.
\end{definition}

Due to symmetry, we can suppose without loss of generality that $s\leq t$. Otherwise, we simply exchange the positions of the players.

To propose an $(s,t)$-mixing algorithm running in polynomial time, we first scrutinize the form of the problem. Expanding the objective function $f$ of the mixing problem given in \Cref{def:st-mixing}, we have the expression
\begin{equation}
\label{equ:st-mixing-expansion}
\begin{aligned}
     \max \{&\max \{R(\beta_1 y_1+\dots +\beta_t y_t)\}-(\alpha_1 x_1+\dots +\alpha_s x_s)^\T R (\beta_1 y_1+\dots +\beta_t y_t),\\
    &\max\{C^\T (\alpha_1 x_1+\dots +\alpha_s x_s)\}-(\alpha_1 x_1+\dots +\alpha_s x_s)^\T C (\beta_1 y_1+\dots +\beta_t y_t)\}.
\end{aligned}
\end{equation}

There are three major components in \eqref{equ:st-mixing-expansion}: inner maximum operators (e.g., $\max\{Ry\}$), vector-matrix-vector products (e.g., $x^\T R y$), and the outermost maximum operator (i.e., $\max\{f_R,f_C\}$). All three kinds of terms present difficulties from different aspects.

\begin{enumerate}
    \item Terms in the form of $\max\{Ry\}$ and $\max\{C^\T x\}$ are piecewise-linear in $\beta$ and $\alpha$, respectively. These two linear-piecewise terms are convex but non-differentiable.
    \item Terms in the form of $x^\T R y$ and $x^\T C y$ are bilinear in $\alpha,\beta$. Thus, these terms are differentiable but nonconvex.
    \item The form of the objective function $f=\max\{f_R,f_C\}$ is non-differentiable.
\end{enumerate}

The sketch of our effort to overcome these difficulties is listed below:
\begin{enumerate}
    \item To solve the first difficulty, we adopt the divide-and-conquer method. We divide the problem and solve it on each linear piece of the $\max$ term. We reduce the problem into the famous \emph{half-plane intersection problem} in computational geometry. The detailed approach is explained in \Cref{details:subsec:linearpiece}.

    \item To solve the second and the third difficulties, we derive necessary optimal conditions for the sub-problems resulting from the previous divide-and-conquer. The main technique we use is the combination of discrete geometry and optimization (more specifically, properties of polytopes and KKT conditions). The detailed approach is presented in \Cref{details:subsec:opt-over-polytope}.

    \item Finally, by discussing these linear pieces and conditions case by case, we are able to design polynomial algorithms to find the \emph{global} minimum of $f$ over $(\alpha,\beta)$. To do so, we formulate the problem into various optimization problems (univariate quadratic programming, linear programming, and fractional programming). The algorithms are presented in \Cref{details:subsec:mixing-algo}.
\end{enumerate}

With these efforts, we are able to present $(1,w)$, $(2,2)$, and $(2,3)$-mixing algorithms.

\begin{theorem}\label{thm:all-mixing}
The following statements hold.

\begin{enumerate}
    \item \label{thm:1-w-mixing}
    There exists a $(1,w)$-mixing algorithm in time $O(mnw+L(w,m))$, where $L(w,m)$ is the time complexity of solving a linear program with $w$ variables and $m$ constraints. 
    \item \label{thm:2-2-mixing}
    There exists a $(2,2)$-mixing algorithm in time $O(mn)$.
    \item \label{thm:2-3-mixing}
    There exists a $(2,3)$-mixing algorithm in time $O\left(m^2(n+\log m)+n\log n\right)$.
\end{enumerate}
\end{theorem}

The algorithms are presented in \Cref{details:sec:mixing-algos} as \Cref{algo:concise-1dbpoint}, \Cref{algo:concise-2dbpoint}, and \Cref{algo:concise-3dbpoint}. They are all polynomial-time and cover the need for algorithms in \Cref{tab:instances}. Still, it is interesting to consider whether there exist polynomial-time mixing algorithms with larger parameters. We discuss the difficulties and the value of this problem in \Cref{sec:mixing-problem}.

Now, in continuation of \Cref{example:modify}, we again consider the BBM algorithm. Below we illustrate the detailed process of the mixing algorithm over it.

\begin{example}[The mixing process of the modified BBM algorithm]
\label{example:mix-modify}
The modified BBM algorithm implements the $(2,2)$-mixing algorithm in its mixing phase. Observe the form of the function $f_R$ over the mixing region $\mathcal{A}=\{(\alpha x^*+(1-\alpha)r_1, \beta y^*+(1-\beta)b_2)\}$:
\[
\begin{aligned}
&f_R(\alpha x^*+(1-\alpha)r_1, \beta y^*+(1-\beta)b_2)\\
=&\max\{R(\beta y^*+(1-\beta)b_2)\}-(\alpha x^*+(1-\alpha)r_1)^\T R (\beta y^*+(1-\beta)b_2).
\end{aligned}
\]

Note that the $\max$ term can be written as $\max\{\beta R (y^*- b_2)+ R b_2\}$. Thus, it has the form of the maximum of $m$ linear functions about $\beta$, which is piecewise linear in $\beta$.

Now, we want to compute the exact form of the piecewise linear function, given by a sequence of breakpoints $0=\beta_1\leq \dots \leq \beta_t=1$ ($t\leq m+1$) so that $f$ is linear in $\beta$ on each $[\beta_i,\beta_{i+1}]$. This is a famous problem in computational geometry called the \emph{envelope problem}, which can be solved in time $O(m \log m)$ (See \Cref{app:prec-2-m-separation} for full details). However, note that we still need to compute the exact form of this problem, that is to compute the value of $R(y^*-b_2)$ and $Rb_2$ with time $O(mn)$.

Similarly, we can compute the linear pieces given by breakpoints $0=\alpha_1\leq \dots \leq \alpha_s=1$ ($s\leq n+1$) in time $O(n \log n)$. Therefore, on each grid $[\alpha_i,\alpha_{i+1}]\times [\beta_j,\beta_{j+1}], i\in [s], j\in [t]$, both $f_R$ and $f_C$ are linear in $\alpha$ and $\beta$, respectively.

Then, we minimize the objective function over each grid and compare the results to take the one with minimal $f$ value. By doing so, we obtain the global minimum of $f$ on region $\mathcal{A}$.

On each grid, the objective function is in the form of the maximum of two bilinear functions $g_1$ and $g_2$. However, it is still non-differentiable. We apply the KKT condition from continuous optimization to obtain the necessary optimal conditions for this problem. We can show (see Appendix \ref{app:prec-2-2-mixing}) that the minimum must be attained at the following three kinds of points:

\begin{enumerate}
\item Points where the partial derivative of $g_1$ or $g_2$ with respect to $\alpha$ or $\beta$ is zero.
\item The four vertices of the grid.
\item Points where $g_1$ and $g_2$ are equal.
\end{enumerate}

Since $g_1$ and $g_2$ are bilinear, the partial derivative is linear. We can easily check points of the first kind in constant time. Then, the second kind contains only four points. Finally, we can show that points of the third kind form a quadratic curve on the plane, on which the minimum can be easily computed by checking at most six points.

In words, on each grid, we only need constant time to compute minimum $f$. Thus, by scanning over all the grids in $O(mn)$ time, we can compute the global minimum of $f$ on $\mathcal{A}$.

The total complexity is given by $O(mn + m\log m + n\log n) = O(mn)$.
\end{example}

Besides, in \Cref{app:subsec:mixing-example}, we also illustrate how to modify the TS algorithm \cite{TS07_0.3393NE} and the DFM algorithm \cite{DFM22_0.3333NE}. Note that these algorithms adopt a gradient descent method in the search phase. Their search phases are totally different from those of the BBM algorithm, which solves a zero-sum game. However, due to the generality of our approach, \Cref{algo:concise-2dbpoint} and \Cref{algo:concise-3dbpoint} can be directly applied to them without any modification.

%% file: section/approx_analysis.tex
\section{Approximation analysis}\label{sec:approx-anal}

After developing the novel procedure of the mixing phase, the immediate next question is how the approximation bound is influenced. This prompts us to develop a novel method to analyze the approximation bound. Below, we refer to the algorithms equipped with the new mixing phases as the revised algorithms.

A clear fact is that the approximation bound of the algorithms in \Cref{tab:instances} will, at the very least, not get worse after the revision. Indeed, our mixing algorithms find the global optimum of the mixing problem, while the original algorithms only choose specific mixing coefficients.

However, we should not expect a direct analysis of the approximation bound of the revised algorithms. There are two main barriers. First, the mixing problem in general is both nonconvex and nonsmooth. Thus, it is difficult to control the global behavior of the objective function $f$. Second, since the search phases have very different characteristics, there is no one-size-fits-all solution for all algorithms.

Instead of a direct analysis, we try to establish an \emph{upper bound} of $f$ for all cases, and in the meantime maintain the ability to attain equal conditions of the upper bound. More specifically, we divide the approximation analysis into two steps.

\begin{itemize}
    \item First, we provide \emph{upper bound constructors} for the corresponding mixing algorithms, which have the following properties:
\begin{enumerate}[label=(\arabic*)]
    \item The form of the upper bound is \emph{search-phase independent}. That is, it does not rely on the particular form of the search phase.
    \item The upper bound can be written explicitly as an expression $h^*$ of $f_R$ and $f_C$ values of strategy profiles given in the search phase.
\end{enumerate} 

\item Second, from the search phase, we dig out relations about $f_R$'s and $f_C$'s occurring in the first step. Viewing these relations as constraints, we maximize $h^*$\footnote{That is, we find the worst possible approximation.}, which produces a fixed-size constraint optimization problem on the approximation upper bound. 

By solving this optimization problem, we can obtain a constant approximation bound. Moreover, conditions to reach the optimal solution then can serve as tightness conditions of the approximation.
\end{itemize}

Now we sketch our approach to both steps. Our approach for the first step is to linearize functions $f_R$ and $f_C$. This approach was first proposed in \cite{CDH+21_0.3393tight}, applied to provide tightness analysis for the TS algorithm \cite{TS07_0.3393NE}. We generalize this method to derive an upper bound of $f$ for a wider range of domains of mixing operations. This approach is presented in \Cref{details:subsec:linearization}.

Through the generalized linearization method, we provide formal rules to generate the corresponding $(v,w)$-upper bound constructor as well as the expression $h^*$ of the upper bound of the $(v,w)$-mixing algorithm for \emph{any} given $v,w$. This is given in \Cref{details:subsec:auxiliary-design}.

The upper bound constructor is valuable from the following perspectives:
\begin{enumerate}
\item Benefiting from the linearization approach, the upper bound derived by the constructor is only related to the value of $f_R$ and $f_C$ on the strategies given by the search phase. Therefore, this shows that the approximation bound of an algorithm is solely determined by its search phase, rather than the mixing phase.

\item The upper bound constructor is fully given by formal rules and is thus totally automatable. Thus it is possible to provide computer programs that make the full approximation analysis and compute the approximation bound.
\end{enumerate}

Now, as an illustration, we continue \Cref{example:mix-modify} to give a detailed analysis of the approximation bound of the modified BBM algorithm.

\begin{example}[The approximation analysis of the modified BBM algorithm]
\label{example:approx-ana}
The modified BBM algorithm in \Cref{example:mix-modify} applies the $(2,2)$-mixing algorithm to compute the global minimal $f^*_\mathcal{A}$ of $f$ on square $\mathcal{A}:= \{(\alpha_1 x^*+\alpha_2 r_1, \beta_1 y^*+\beta_2 b_2): \alpha, \beta\in \Delta_2\}$. Now, we derive an upper bound of $f^*_\mathcal{A}$ following the four steps below.

\vspace{5pt}
\textbf{Step 1: Construct the upper bound.}

We illustrate the construction of the corresponding $(2,2)$-upper bound constructor step by step. We denote $g_1=f_R(x^*, y^*)$, $g_2=f_C(x^*, y^*)$, $h_1=f_R(x^*, b_2)$, $h_2=f_C(x^*, b_2)$, $v_1=f_R(r_1, b_2)$, $v_2=f_C(r_1, b_2)$, $u_1=f_R(r_1, y^*)$, and $u_2=f_C(r_1, y^*)$. The aliases are illustrated in \Cref{fig:mod-BBM}.

\begin{figure}[ht]
    \centering \includegraphics[width=0.5\textwidth]{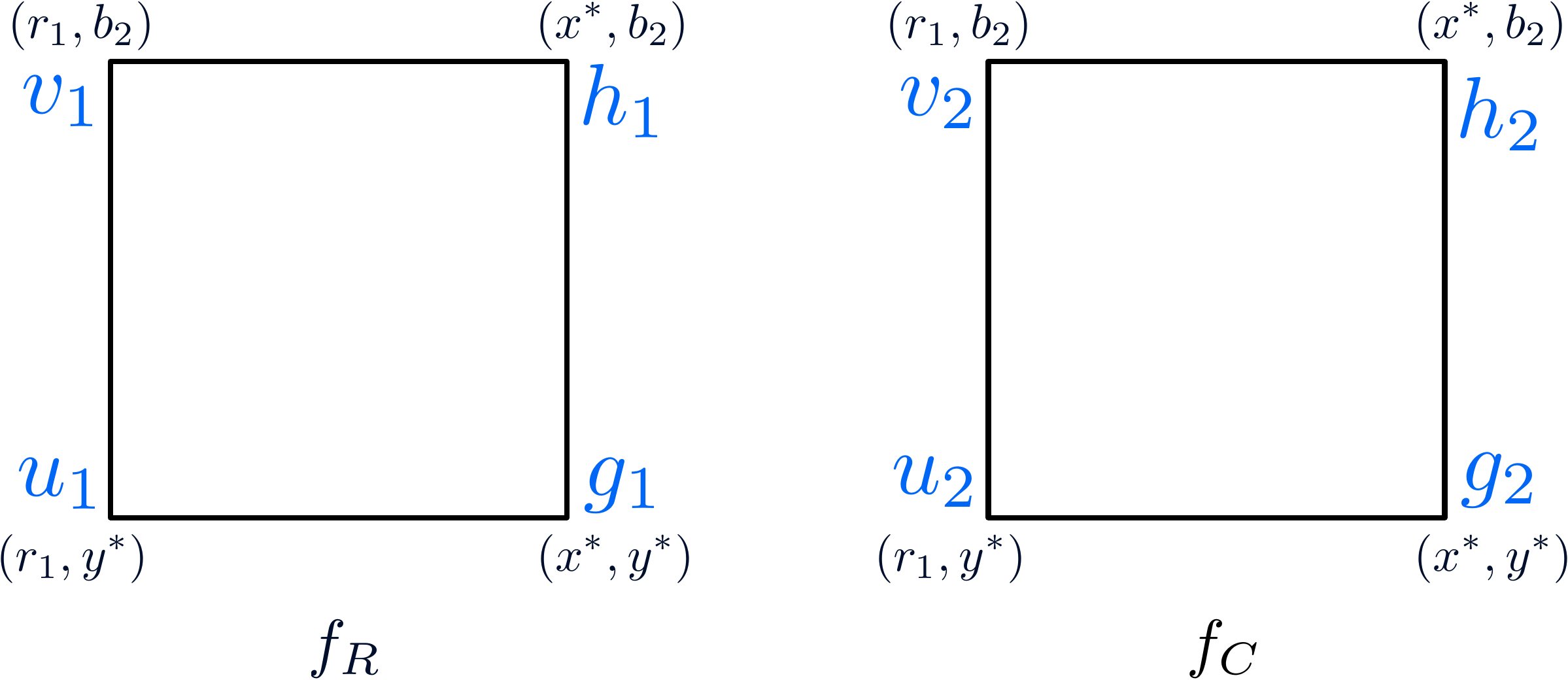}
    \caption{Positions of little letter aliases (case $g_1\geq g_2$).}
    \label{fig:mod-BBM}
\end{figure}

Now, we derive the upper bound of $f^*_{\mathcal{A}}$. Observe that it is upper bounded by the global minimum of $f$ over the four edges of square $\mathcal{A}$. On each edge, by the convexity of $f_R$ and $f_C$, the value of $f_R$ or $f_C$ is bounded by the linear combination of the values at the vertices. This is exactly the linearization process stated in \Cref{details:subsec:linearization}. For example, on the edge $\{(\alpha x^*+(1-\alpha) r_1, y^*): \alpha\in [0,1]\}$, we have
\[
\begin{aligned}
&f_C(\alpha x^*+(1-\alpha) r_1, y^*)\\
=&\max \{C^\T (\alpha x^*+(1-\alpha) r_1) \}- (\alpha x^*+(1-\alpha) r_1)^\T C y^*\\
\leq& \alpha (\max \{C^\T x^*\}-(x^*)^\T C y^*)+(1-\alpha)(\max \{C^\T r_1\}-r_1^\T C y^*)\\
=& \alpha f_C(x^*,y^*)+(1-\alpha)f_C(r_1,y^*).
\end{aligned}
\]
Similarly, we have $f_R(\alpha x^*+(1-\alpha) r_1, y^*)\leq \alpha f_R(x^*,y^*)+(1-\alpha)f_R(r_1,y^*)$. Therefore, we obtain an upper bound of $f$ on this edge: $s_1=\min_{\alpha\in [0,1]}\max\{\alpha g_1+(1-\alpha)u_1, \alpha g_2+(1-\alpha)u_2\}$.

We can do similar discussions on other edges. On the edge $\{(\alpha x^*+(1-\alpha) r_1, b_2): \alpha\in [0,1]\}$, the upper bound is $s_2=\min_{\alpha\in [0,1]}\max\{\alpha v_1+(1-\alpha)h_1, \alpha v_2+(1-\alpha)h_2\}$. On the edge $\{(x^*, \beta y^*+(1-\beta)b_2): \beta\in [0,1]\}$, the upper bound is $s_3=\min_{\beta\in [0,1]}\max\{\beta g_1+(1-\beta)h_1, \beta g_2+(1-\beta)h_2\}$. On the edge $\{(r_1, \beta y^*+(1-\beta)b_2): \beta\in [0,1]\}$, the upper bound is $s_4=\min_{\beta\in [0,1]}\max\{\beta u_1+(1-\beta)v_1, \beta u_2+(1-\beta)v_2\}$. Therefore, $f^*_\mathcal{A}$ is upper bounded by the minimum of the four upper bounds of $f$ on the edges of $\mathcal{A}$, i.e., 
\[f^*_\mathcal{A}\leq\min\{s_1,s_2,s_3,s_4\}=h^*.\]

It should be noted that the analysis until now has nothing to do with the property of the strategies in the search phase. Besides, the presented upper bound is only related to the value of $f_R$ and $f_C$ on the vertices of $\mathcal{A}$. Therefore, the approximation bound of \Cref{example:mix-modify} is totally determined by the search phase and has nothing to do with the mixing phase!

\vspace{5pt}
\textbf{Step 2: Extract relations from the search phase.}

Then, we extract the relations of $f_R$ and $f_C$ values on the vertices.

From \Cref{example:start}, we obtain the following properties: $r_1\in \br_R(y^*)$, $b_2\in \br_C(r_1)$, $(x^*,y^*)$ is the NE of game $(R-C, C-R)$, and $g_1\geq g_2$.

According to the definition of $\br$, $u_1=v_2=0$. Besides, since we have supposed that all the elements of $R,C$ are in $[0,1]$, we have that for any $x\in\Delta_m,y\in\Delta_n$, $0\leq x^\T Ry \leq 1$, $0\leq x^\T Cy \leq 1$, $0\leq \max\{ Ry\} \leq 1$, $0\leq \max\{ C^\T x\} \leq 1$, $0\leq f_R(x,y) \leq 1$, and $0\leq f_C(x,y) \leq 1$.

Then, since $(x^*,y^*)$ is the NE of game $(R-C, C-R)$, for any $i, j\in [n]$,
\begin{equation}
\label{equ:R-C-NE}
    (x^*)^\T R e_j-(x^*)^\T C e_j \geq(x^*)^\T R y^*-(x^*)^\T C y^* \geq e_i^\T R y^*-e_i^\T C y^*.
\end{equation}
Specifically, take $e_i=r_1$, we get $r_1^\T C y^* \geq r_1^\T R y^*-(x^*)^\T R y^*+(x^*)^\T C y^*$. But we know that $r_1^\T R y^*-(x^*)^\T R y^*=g_1$, thus:
\[
r_1^\T C y^* \geq g_1+(x^*)^\T C y^* \geq g_1.
\]
Therefore, $u_2= \max\{C^\T r_1\}-r_1^\T C y^*\leq 1-g_1$. 

A summary of relations is as follows.
\begin{gather*}
    u_1=v_2=0,\\
    0\leq g_1,g_2,h_1,h_2,u_2,v_1\leq 1,\\
    g_1\geq g_2,\\
    u_2\leq 1-g_1.
\end{gather*}

\vspace{5pt}
\textbf{Step 3: Write down the optimization problem of the upper bound.}

To produce a constant approximation bound $s$, we require that no matter what values these $f_R$ and $f_C$ have, they cannot have a greater approximation than $s$. Thus, we have to maximize all these $f_R$'s and $f_C$'s. Therefore, we can write down the following optimization problem of $h^*$. The optimal value is the desired $s$.

\begin{equation}
    \begin{aligned}
\max_{g_1,g_2,h_1,h_2,u_1,u_2,v_1,v_2} & \quad \min_{\alpha_1,\alpha_2,\alpha_3,\alpha_4} \min\{s_1,s_2,s_3,s_4\}\\
    \text{s.t.} &\quad s_1 =\max\{\alpha_1 g_1 + (1-\alpha_1)u_1, \alpha_1 g_2 + (1-\alpha_1)u_2\},\\
    &\quad s_2=\max\{\alpha_2 g_1 + (1-\alpha_2)h_1, \alpha_2 g_2 + (1-\alpha_2)h_2\},\\
    &\quad s_3=\max\{\alpha_3 v_1 + (1-\alpha_3)h_1,\alpha_3 v_2 + (1-\alpha_3)h_2\},\\
    &\quad s_4=\max\{\alpha_4 v_1 + (1-\alpha_4)u_1,\alpha_4 v_2 + (1-\alpha_4)u_2\},\\
    &\quad 0\leq \alpha_1, \alpha_2, \alpha_3, \alpha_4 \leq 1,\\
    &\quad u_1=v_2=0,
    0\leq g_1,g_2,h_1,h_2,u_2,v_1\leq 1,
    g_1\geq g_2,
    u_2\leq 1-g_1.
\end{aligned}\label{eq:origin-bound}
\end{equation}

Now, if we can solve the optimization problem \eqref{eq:origin-bound}, then we can finish the approximation analysis. However, on the one hand, directly solving this expression in Mathematica is extremely slow (it takes hours). On the other hand, it is inconvenient for further proving the correctness of the calculated bound. Thus, below we divide the calculation into two steps.

\vspace{5pt}
\textbf{Step 4-1: Eliminate the minimum operator over $\alpha_i$'s in $h^*$.}
\label{step:eliminate-operator}

Now, we use the relations in Step 2 to eliminate the minimum operator over $\alpha_i$'s. Then $h^*$ can be explicitly expressed using $f_R$'s and $f_C$'s on the vertices. Note that this step can be \emph{automated} as well: In \Cref{details:subsec:auxiliary-minmax}, we provide a search-phase-independent expression for such $s_i$'s while eliminating the minimum operator. This is sufficient for our program to solve the optimization problem \eqref{eq:origin-bound} and obtain the approximation bound of $0.38$. However, to prove that the bound is correct, we show the exact expression concerning the relations of the $f_R$'s and $f_C$'s below.

Since $u_1=0$, $g_1\geq g_2$, $s_1=\min_{\alpha\in [0,1]}\max\{\alpha g_1+(1-\alpha)u_1, \alpha g_2+(1-\alpha)u_2\}=\frac{u_2 g_1}{u_2+g_1-g_2}$. Similarly, on the edge $\{(\alpha x^*+(1-\alpha) r_1, b_2): \alpha\in [0,1]\}$, $s_2$ is $\min \{v_1,h_1\}$ (if $h_1\geq h_2$) or $\frac{v_1h_2}{v_1+h_2-h_1}$ (if $h_1<h_2$). On the edge $\{(x^*, \beta y^*+(1-\beta)b_2): \beta\in [0,1]\}$, $s_3$ is $\min \{g_1,h_1\}$ (if $h_1\geq h_2$) or $\frac{g_1 h_2- h_1 g_2}{g_1+h_2-h_1-g_2}$ (if $h_1<h_2$). On the edge $\{(r_1, \beta y^*+(1-\beta)b_2): \beta\in [0,1]\}$, $s_4= \frac{u_2 v_1}{u_2+v_1}$.

Now we collect the bounds above. The optimization \eqref{eq:origin-bound} is now simplified to maximize \(h^*\) given by the following expression:
\begin{equation}
\label{equ:example:upperbound}
    \begin{aligned}
&\min\left\{\frac{u_2 g_1}{u_2+g_1-g_2},v_1,h_1, g_1, \frac{u_2v_1}{u_2+v_1}\right\},\quad\text{ if } h_1\geq h_2, \\
&\min\left\{\frac{u_2 g_1}{u_2+g_1-g_2}, \frac{v_1 h_2}{v_1+h_2-h_1}, \frac{g_1 h_2- h_1 g_2}{g_1+h_2-h_1-g_2}, \frac{u_2 v_1}{u_2+v_1}\right \},\quad \text{ otherwise}. 
\end{aligned}
\end{equation}

\vspace{5pt}
\textbf{Step 4-2: Obtain the constant upper bound.}

Finally, we maximize $h^*$ given in \eqref{equ:example:upperbound} under constraints presented in \eqref{eq:origin-bound}. The step can also be accomplished by a computer program. Here we demonstrate a hand-written procedure.

First, note that either $h_1 \geq h_2$ or $h_1 < h_2$. The bound can be relaxed to $\min\left\{\frac{u_2 g_1}{u_2+g_1-g_2}, \frac{u_2 v_1}{u_2+v_1}\right\}$.

Then, since $g_1 \geq g_2$, $\frac{u_2 g_1}{u_2+g_1-g_2} \leq g_1$. By $0 \leq v_1 \leq 1, 0 \leq u_2 \leq 1-g_1$, we have $\frac{u_2 v_1}{u_2+v_1} \leq \frac{1-g_1}{2-g_1}$.

Therefore, the bound is finally relaxed to $\min\left\{g_1, \frac{1-g_1}{2-g_1}\right\}$. Note that the solution of $g_1 = \frac{1-g_1}{2-g_1}$ is $g_1 = \frac{3-\sqrt{5}}{2}$. Moreover, $\frac{1-g_1}{2-g_1}$ is decreasing in $g_1$ when $0 \leq g_1 \leq 1$. Thus, for any possible $g_1$, we have $\min\left\{g_1, \frac{1-g_1}{2-g_1}\right\} \leq \frac{3-\sqrt{5}}{2}$. This implies that $f^*_\mathcal{A} \leq \frac{3-\sqrt{5}}{2} \approx 0.38$, as desired.
\end{example}

Although the process of our analysis seems complex, we can see that the upper bound can be completely calculated by a computer program. For instance, such a program can be written in Mathematica. More importantly, this analysis method works for general algorithms following the novel search-and-mix paradigm! See the last paragraph of this section.

Another advantage of our procedure is that it shows a clear structure of the tight cases. Since each step of relaxation is very clear, we can easily deduce a necessary condition for tight cases as illustrated below.

\vspace{5pt}
\textbf{One more step: Derive the tightness conditions.}

Denote $\alpha= (3-\sqrt{5})/2$. In a tight case, we must have $g_1=\alpha$. Besides, each step of relaxation should be tight. Therefore, we have $v_1=1$, $u_2=1-\alpha$, $g_1=g_2=\alpha$, $\max\{C^\T r_1\}=1$, and $(x^*)^\T C y^*=0$. These properties together give a computable necessary condition for the tight case, which helps to further derive ad hoc algorithms to improve the approximation bound. Benefiting from a similar discussion, \cite{CDH+21_0.3393tight} is even able to give a tight instance generator for the TS algorithm \cite{TS07_0.3393NE}.

After the tightness analysis, we might ask if we can improve the approximation by modifying the search phase. As a motivating example, we show a very simple approach to modify the search phase of the BBM algorithm. Our goal is not to improve the approximation, but to show that our framework allows giving a result very quickly for the newly modified search phase (which does not occur in the literature to the best of our knowledge).

\vspace{5pt}
\textbf{Beyond the BBM algorithm: Modify the search phase.}

We make a very naive change: Let the returned $(x^*,y^*)$ of the BBM algorithm to be a NE of the zero-sum game $(R-C/t,C/t-R)$, where $t>0$. Note that now we cannot suppose without harm that $g_1\geq g_2$. So, we make discussions in cases $g_1\geq g_2$ and $g_1<g_2$.

We note that the only change it brings is that now \eqref{equ:R-C-NE} is transformed to
\[
(x^*)^\T R e_j-(x^*)^\T C/t e_j \geq(x^*)^\T R y^*-(x^*)^\T C/t y^* \geq e_i^\T R y^*-e_i^\T C/t y^*.
\]
In this way, the bound of $u_2$ is changed to $u_2\leq 1-t g_1$. Thus, the final bound is changed to $\min\{g_1, (1-t g_1)/(2-t g_1)\}$. By solving a similar equation, the upper bound is given by $\left((t+2)-\sqrt{(t+2)^2-4t}\right)/(2t)$.

The case of $g_1<g_2$ is symmetric, where the inequality is changed to $u_2\leq 1-g_1/t$. The upper bound is given by $\left((1+2t)-\sqrt{(1+2t)^2-4t}\right)/2$.

In the worst case, the upper bound $h^*$ is 
\[
\max\left\{\frac{(t+2)-\sqrt{(t+2)^2-4t}}{2t},\frac{(1+2t)-\sqrt{(1+2t)^2-4t}}{2}\right\}.
\]
We can easily show that the former term is increasing in $t$, while the latter one is decreasing in $t$. In this way, the bound is minimized when $t=1$, which happens to be the case of the BBM algorithm. So, we have actually seen that the choice of form $R-C$ leads to an optimal approximation bound for the BBM algorithm!

We have the following remarks on the extensions of our analysis approach.

\begin{remark}\label{remark:0.36-modify-BBM}
In the above modification, we introduce a parameter $t$ in the search phase. Then using our framework, we are able to obtain an expression of the approximation bound with a parameter $t$. This is a usual way to modify the search phase. Indeed, there are other modifications for the BBM algorithm, where a successful one is given by \cite{BBM07_0.36NE}. It changes the choice of $b_2$ to a pure strategy in $\br_C((1-\delta_1)x^*+\delta_1r_1)$ with $\delta_1$ specifically chosen. In this way, it managed to obtain an improved approximation of $0.36$. In our framework, $\delta_1$ is just like $t$, and thus our framework can also be used to provide quick results for such modifications.
\end{remark}

\begin{remark}\label{remark:stronger-upper-bound}
It is worth mentioning that instead of only scanning over the boundary as shown in \Cref{example:approx-ana} above, there is an alternative way to derive a stronger upper bound by scanning over the whole region. This method is stated in \Cref{prop:v-w-strong-aux-mixing}. However, due to the difficulty in eliminating the minimum operators in Step 4-1, we practically only apply it to several examples in a restricted form, as described in \Cref{details:literature-examples}.
\end{remark}

Finally, we also provide more examples on other algorithms. In \Cref{details:subsec:auxiliary-minmax} and \Cref{details:subsec:aux-example}, we respectively show how to prove that the approximation bounds of the modified TS algorithm and the modified DFM algorithm are still $0.3393+\delta$ and $1/3+\delta$. We also provide a tightness analysis for these examples. To show our method can indeed be automated, for \emph{all} algorithms in \Cref{tab:instances}, we give details about the formalization of approximation analysis in \Cref{details:literature-examples}. We then implement Mathematica code to derive the corresponding approximation bounds. The results are presented in \Cref{tab:approx-result}.

%% file: section/mixing_problem.tex
\section{The mixing problems}\label{sec:mixing-problem}
In this section, we further discuss the mixing problems defined in \Cref{def:st-mixing}. It turns out that this class of problems brings new algorithmic and computational complexity issues related to approximate Nash equilibria.

First, note that \Cref{def:st-mixing} states the problem in a \emph{functional optimization style}, that is, we try to find mixing coefficients (i.e., a witness) that minimizes function $f$. Due to our divide-and-conquer strategy stated in \Cref{sec:find-opt-comb} and \Cref{details:subsec:opt-over-polytope}, from an optimization perspective, the mixing problem can be reduced to a special kind of quadratic constraint quadratic programming (QCQP). More precisely, the mixing problem can be reduced to a bilinear program with one bilinear equality constraint.

The main result in \Cref{sec:find-opt-comb} (i.e., \Cref{thm:all-mixing}) shows that the $(s,t)$-mixing problem has a polynomial-time algorithm when $(s,t)$ is one of $(1,w)$, $(v,1)$, $(2,2)$, $(2,3)$, and $(3,2)$. It is natural to ask whether we could have a better result.

Actually, it highly depends on the relation between $s,t$ and the number of pure strategies, i.e., $n,m$. If $s=\Theta(m)$ and $t=\Theta(n)$, we cannot expect the existence of polynomial-time algorithms. This claim relies on the following complexity result on Nash equilibria: It is NP-hard to decide whether there is a Nash equilibrium without a certain support \cite{CS08}. Suppose the $(\Theta(m),\Theta(n))$-mixing problem has a polynomial-time algorithm $A$. Then we can input $A$ with pure strategies that we want to become the support of a Nash equilibrium but without the others. Then $A$ can compute the minimum approximation $f^*$ of these input pure strategies. From the definition of $f$, these strategies can form a Nash equilibrium if and only if $f^*=0$. The above argument shows that such $A$ can hardly exist.

On the other hand, if $s,t$ are constant with respect to $m,n$, numerical difficulty arises even for $(s,t)=(3,3)$. In this case, the QCQP that we reduce to has the form $g(\alpha,\beta)=\sum_{1\leq i,j\leq 2}a_{ij}\alpha_i\beta_j+\sum_{1\leq i\leq 2}b_i\alpha_i+\sum_{1\leq i\leq 2}c_i\beta_i+d$ for both the objective function and the quadratic equality constraint. This forms a quadratic surface in $\mathbb{R}^4$ and the minimum might be obtained at points that are the solutions of a quintic equation. Thus, this solution could not be written in a radical form \cite{dummit2004abstract}, and therefore is not representable by a Turing machine.

Thus, it is more proper to ask for an \emph{approximate solution} for the mixing problem. To fit the tradition in computational complexity, we write the decision version of the mixing problem \textsc{MixStgy} as follows: Decide whether there exist rational mixing coefficients $\alpha$ and $\beta$ such that the corresponding $f$ value is less than a given $\epsilon > 0$. Note that the problem of computing approximate Nash equilibria then becomes a special case of \textsc{MixStgy}.

To understand the difficulty of \textsc{MixStgy} over input parameters $m$, $n$, $s$, and $t$, we make an analogy to the \textsc{Sat} problem. \textsc{Sat} has two input parameters: the number of clauses, $p$, and the number of proposition variables, $q$. The hardness of \textsc{Sat} actually comes from $q$, but not $p$. This is because when $q$ is fixed, the simple truth enumeration algorithm takes $\poly(p)$ time to decide a \textsc{Sat} instance.

Similarly, $m,n$ in \textsc{MixStgy} can be viewed as $p$, and $s,t$ in \textsc{MixStgy} can be viewed as $q$. When $s$ and $t$ become larger, difficulty arises. At last, when $s=\Theta(m)$ and $t=\Theta(n)$, the problem becomes NP-hard as discussed above. Thus, even for small and fixed $s$ and $t$, \textsc{MixStgy} is harder than one may expect.

People try to consider various restricted versions of \textsc{Sat}, which thrive in computational complexity theory. Just like the situation of \textsc{Sat}, we hope our considerations on \textsc{MixStgy} can bring new perspectives on equilibrium computation as well as computational complexity.

bal

%% file: section/conclusion-and-discussion.tex
\section{Conclusion and Discussion}\label{sec:conclusion}
In this paper, we formalize the search-and-mix paradigm of polynomial-time algorithms for approximate Nash equilibria in bimatrix games. We call it \emph{paradigm} because \emph{all} algorithms in the literature follow it. We reformulate the mixing phase and propose a generally applicable polynomial-time subroutine for the reformed mixing phase. We also present an automated method to derive the approximation bound of any search-and-mix algorithm when suitable information from the search phase is provided. After this paper, researchers are freed from the cumbersome inequality calculations in the mixing phase and can solely focus on the search phase.

Our procedure of the search-and-mix paradigm shown in \Cref{fig:framework-new} can be extended to multiplayer scenarios by simply adding players. Moreover, the method of approximation analysis presented in \Cref{sec:approx-anal} can be extended as well. This is because it only relies on the edges of the mixing region, each of which is solely determined by one player. There are very few works in the literature that calculate $\epsilon$-NE for games with more than two players. Our work opens up a possible approach to designing and analyzing multiplayer approximation algorithms.

We also realize and promote the optimization viewpoint of approximate Nash equilibria. We provide a detailed discussion on a basic problem from this viewpoint: the mixing problem. We relate the mixing problem to the \textsc{Sat} problem to capture its role in approximate Nash equilibria. Other problems in approximate Nash equilibria should benefit from this perspective.

%% file: section/append/details.tex
\section{Mixing algorithms}
\label{details:sec:mixing-algos}

\subsection{Linear piece partitioning}
\label{details:subsec:linearpiece}

In this part, we provide the solution for the first difficulty concerning \eqref{equ:st-mixing-expansion}. The idea is to partition the domain into regions where both $f_R$ and $f_C$ are linear in $\alpha$ and $\beta$, respectively. To make a concise description, for a function $F:X\to\R$, we say a \emph{linear 
piece} of $F$ is the maximal region $\Omega\subseteq X$ such that $F$ is linear in each variable on $\Omega$.

Consider function $\max \{R(\beta_1 y_1+\dots +\beta_t y_t)\}$. We have:
\begin{equation}
    \max \{R(\beta_1 y_1+\dots +\beta_t y_t)\} \\
    =\max \{\beta_1 R y_1+\dots +\beta_t R y_t\} \\
    =\max_{1\leq i\leq m} \{\beta_1 (R y_1)_i+\dots +\beta_t (R y_t)_i \}
\end{equation}
It is the maximum of $m$ linear functions in $\beta$. Therefore, we can partition the domain into several linear pieces, as illustrated in \Cref{fig:linear-pieces}.

\begin{figure}[ht]
    \centering  
    \input{figures/linear-pieces}
    \caption{Linear pieces.}
    \label{fig:linear-pieces}
\end{figure}
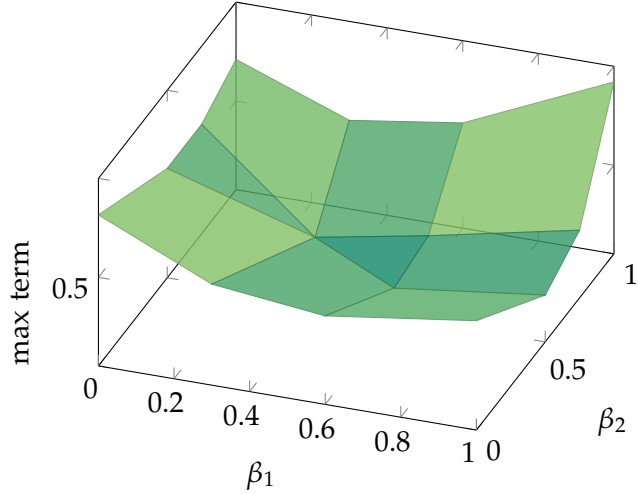

An important observation is that the linear pieces can be expressed by linear inequalities. When the $i$th linear function attains the maximum, we have
\[
\forall j\in [m], j \neq i, \beta_1 (R y_1)_i+\dots +\beta_t (R y_t)_i \geq \beta_1 (R y_1)_j+\dots +\beta_t (R y_t)_j.
\]
Namely,
\[
\forall j\in [m], j \neq i, \beta_1 [(R y_1)_i-(R y_1)_j]+\dots +\beta_t [(R y_t)_i-(R y_t)_j] \geq 0.
\]
Thus, each linear piece of the $\max\{Ry\}$ term can be determined by $m-1$ inequalities, and forming a (possibly empty) polytope (see \Cref{app:def-geo} for the formal definition). There are $m$ such polytopes, denoted by $P^C_1,\dots , P^C_m$, where notation $C$ means that they are collections of the column player's strategies.

Similarly, the linear pieces of $\max\{C^\T x\}$ term can be determined by $n-1$ inequalities. We denote the polytopes by $P^R_1,\dots, P^R_n$.

Our solution is to divide the problem into each polytope $P^R_i\times P^C_j$ ($i,j\in [n]\times [m]$), which we call the \emph{separated polytope}. In this way, we can eliminate the inner $\max$ term of \eqref{equ:st-mixing-expansion} and obtain the following problem:
\begin{equation}
\label{equ:smoothed}
\begin{aligned}
    \min_{\alpha\in (\Delta_s \cap P^R_i), \beta \in (\Delta_t\cap P^C_j) } \max \{\beta_1 (R y_1)_i+\dots +\beta_t (R y_t)_i-(\alpha_1 x_1+\dots +\alpha_s x_s)^\T R (\beta_1 y_1+\dots +\beta_t y_t),\\
    \alpha_1 (C^\T x_1)_j+\dots +\alpha_s (C^\T x_s)_j-(\alpha_1 x_1+\dots +\alpha_s x_s)^\T C (\beta_1 y_1+\dots +\beta_t y_t)\}.
\end{aligned}
\end{equation}

To meet further needs, we also care about solving and expressing the separated polytopes efficiently. Formally, we want to solve the following problem:

\begin{definition}[$(t,m)$-separation algorithm]
\label{def:tm-sep-algo}
~
\begin{itemize}
    \item \emph{Input}: dimension $m$, $t$ vectors $x_1,\dots ,x_t$ in $\R^m$.
    \item \emph{Output}: an appropriate representation\footnote{For example, when $t=2$, this can be given by a clockwise enumeration of vertices of the polytope. For other cases, see \Cref{app:prec-3-m-separation}.} of $P_1,\dots, P_m$ such that $P_i$ is the polytope $\{\beta \in \Delta_t: \beta_1 [(x_1)_i-(x_1)_j]+\dots +\beta_t [(x_t)_i-(x_t)_j] \geq 0,\forall j\in [m]\}$.
\end{itemize}
\end{definition}

When $t\leq 3$, we represent the polytopes with a clockwise enumeration of its vertices. In this case, the separation problem can be restated as a famous problem in computational geometry called the \emph{half-plane intersection problem} \cite{MOMM08}. Benefiting from geometric intuitions, we obtain polynomial-time algorithms for $t\leq 3$, as stated below.

\begin{theorem}
\label{thm:2-m-separation}
There exists a $(2,m)$-separation algorithm in time $O(m\log m)$.
\end{theorem}

\begin{theorem}
\label{thm:3-m-separation}
There exists a $(3,m)$-separation algorithm in time $O(m^2 \log m)$.
\end{theorem}

The algorithm and complexity analysis are presented in \Cref{app:prec-2-m-separation} and \Cref{app:prec-3-m-separation}.

We also note that if we only require the "appropriate" expression in \Cref{def:tm-sep-algo} to be a vertex enumeration of the polytope, then in general cases, it can be stated in the famous \emph{vertex enumeration problem}.

Suppose we are given a polytope in $\R^t$ determined by $m$ inequalities, then McMullen's upper bound theorem \cite{M70} gives a close upper bound $\mathcal{A} (m^{t/2})$ on the number of its vertices $|V|$.

Several algorithms are proposed for the vertex enumeration problem. Using the pivoting method, Dyer \cite{D83} proposed an $O(mt^2 |V|)$-time algorithm. Then, Avis and Fukuda \cite{AF91} proposed an $O(mt|V|)$-time algorithm, which has remained state-of-the-art since then. For a brief summary of this subject, see \cite{AMA22} as a reference.

We note that our algorithms (\Cref{algo:concise-2dbpoint} and \Cref{algo:concise-3dbpoint}) are faster than these algorithms in the corresponding cases. Indeed, when $t=2$, the time complexity of vertex enumeration is $O(m^3)$. For $t=3$, the time complexity is $O(m^{3.5})$.

It is also worth mentioning the complexity results regarding this problem. For the unbounded case (polyhedra), vertex enumeration has been proven to be NP-hard \cite{KBBGE09}. However, for the bounded case (the case of our problem, which is bounded in $\Delta_t$), it is still an open problem. There is a strong indication of NP-hardness, though, as \cite{K01} proved that uniformly sampling the vertices is NP-hard.

\subsection{Optimization over polytopes}
\label{details:subsec:opt-over-polytope}

Using linear piece separation, we have transferred the form of the problem into solving the subproblem \eqref{equ:smoothed}. To solve this subproblem, we first derive optimal conditions for a slightly generalized problem. For differentiable functions $g_1, g_2$ and polytope $S$, consider the following optimization problem:
\begin{equation}
\label{equ:opt-over-polytope}
\begin{aligned}
    \text{minimize}\quad &\max\{g_1(x), g_2(x)\}\\
    \text{s.t.}\quad &x\in S.
\end{aligned}
\end{equation}

We begin our preparation by a direct application of the KKT condition (see theorem 12.1 in \cite{N99_NumOpt} for details).
\begin{lemma}\label{details:lemma:argmax}
Consider $U\in \R^{k\times n}, V\in \R^k, R\in \R^{j\times n}, T\in \R^j$, where every row of $U,R$ is not zero. Define a convex polytope $S=\{x\in \R^n: Ux\leq V, Rx=T\}$ . Suppose $g_1$ and $g_2$ are two real-valued differentiable functions defined on $S$. Set $g=\max\{g_1,g_2\}$. If the ranges of $g_1$,$g_2$ on $S$ are $[m_1,M_1]$ and $[m_2,M_2]$, respectively, then we have:

\begin{enumerate}[fullwidth]
    \item If $m_1\geq M_2$, $\min_{S} g(x)=m_1$. The minimum is attained precisely on set $g_1^{-1}(m_1)$.
    \item If $m_2\geq M_1$, $\min_{S} g(x)=m_2$. The minimum is attained precisely on set $g_2^{-1}(m_2)$.
     \item Otherwise, $\min_Sg(x)=\min_{S^*} g(x)$ , where $S^*$ is the union of following sets:
     \begin{gather*}
         \{x\in S: g_1(x)=g_2(x)\},\\
         \left\{x\in S:g_1(x)>g_2(x),\exists \lambda\geq 0\left(\nabla g_1(x)+ \lambda^\T \left( 
         \begin{matrix}
         U\\R
         \end{matrix}
         \right)=0\text{ and }\forall i\in [k],  \lambda_i (U_i x-V_i)=0\right) \right\},\\
         \left\{x\in S:g_2(x)>g_1(x),\exists \lambda\geq 0\left(\nabla g_2(x)+ \lambda^\T \left( 
         \begin{matrix}
         U\\R
         \end{matrix}
         \right)=0\text{ and }\forall i\in [k],  \lambda_i (U_i x-V_i)=0\right) \right\}.
     \end{gather*}
     And the minimum must be attained on $S^*$.
\end{enumerate}
\end{lemma}

The proof is presented in \Cref{app:argmax}.

Now we turn to polytopes. For concepts in polytopes, see \Cref{app:def-geo} and textbook \cite{Z12_GTMpolytope}. The following proposition captures the relationship of geometric properties and constraint expressions, which helps in the further analysis of the minimization problem on a certain polytope.

\begin{proposition}
\label{prop:stand}
Consider the polytope $S=\left\{x\in\R^n:a_i^\T x\leq b_i,\forall i\in [k]\right\}$, where $a_i\in \R^n\setminus\{0\}$, $b_i\in\R$. Suppose the dimension of $S$, denoted by $\dim(S)$, is $m\leq n$. Then we have
\begin{enumerate}
\item \label{prop:stand:state1}
The affine hull $\aff(S)$ of $S$ can be written in the form $\left\{x\in\R^n:u_i^\T x=v_i,\forall i\in [n-m]\right\}$. 
\item \label{prop:stand:state2}
Vector $d$ is parallel to $S$ (denoted by $d\parallel S$) if and only if for every $i\in[n-m]$, $u_i^\T d=0$.
\end{enumerate}
The representations of geometric concepts about $S$ can be presented in the following order.

\begin{enumerate}[resume]
\item \label{prop:stand:S}
    (Representation of $S$) There exists a set $W\subseteq [k]$ of indices such that:
    \[\displaystyle S=\left\{x\in\R^n:u_i^\T x=v_i,\forall i\in [n-m]\right\}\cap \left\{x\in\R^n:a_i^\T x\leq b_i, \forall i\in W\right\}.\]
\item \label{prop:stand:partialS}
    (Representation of boundary $\partial S$ and interior $S^\circ$) Moreover:
    \[\partial S=\{x\in S: \exists i\in W, a_i^\T x=b_i\}, S^\circ= \{x\in S: \forall i\in W, a_i^\T x<b_i\}.\]
\item \label{prop:stand:facetS}
    (Representation of facets of $S$) For every $j\in W$, $S'_j:=\left\{x\in S: a_j^\T x=b_j\right\}$ is a distinct facet of $S$, and every facet of $S$ coincides with exactly one $S_j'$.
\item \label{prop:stand:faceS}
    (Representation of faces of $S$) For any face $T$ of $S$, $\dim(T)\leq m-1$, and $T$ can be expressed as the intersection of facets of $S$.
\end{enumerate}
\end{proposition}

The proof is presented in \Cref{details:stand}.

Now we combine discrete geometry and optimization. We derive three corollaries from \Cref{details:lemma:argmax} to deal with simpler cases. 

\begin{corollary}\label{details:cor:mincol}
For any convex polytope $S\in \R^n$ such that $\dim(S)=n$, suppose without loss of generality that it has a form that $S=\{x\in\R^n:U x\leq V \}$, where $U\in \R^{m\times n}$, $V\in \R^m$ and no rows of $U$ are zero. We have the following statements.

\begin{enumerate}[fullwidth]
\item \label{cor:mincol:state1}
The minimum of $g$ on $S$ must be obtained on
\begin{align*}
    S^+=\partial {S}\cup \{x\in S:\nabla g_1(x)=0\} \cup\{x\in S:\nabla g_2(x)=0\}\cup\{x\in S:g_1(x)=g_2(x)\}.
\end{align*}

\item \label{cor:mincol:state2}
Let $e_1,\dots,e_n$ be the standard orthonormal basis. For any $e_i$, we can divide the facets of $S$ into two collections: $P_i$ and $N_i$ according to whether they are parallel to $e_i$. Define $\partial S_P=\bigcup_{T\in P_i} T$ and $\partial S_T=\bigcup_{T\in N_i} T$.  $\partial S_P\cup \partial  S_T=\partial S$. For any index $i$, statement 1 still holds if we substitute $\partial S$ with
\[\left(\partial S_P\bigcap \bigcup_{k=1,2}\left\{x\in S:\frac{\partial g_k}{\partial x_i}(x)=0\right\}\right)\bigcup\partial  S_T.\]

\item \label{cor:mincol:state3}
If the polytope $S$ has the form $[m_1,M_1]\times [m_2,M_2]\times\dots\times [m_n,M_n]$ with $m_i<M_i$, then the minimum must be obtained on
\begin{align*}
    S^+=&\left\{x\in\R^n:\forall i, x_i\in \{m_i,M_i\}\right\}\bigcup\\
    &\bigcup_{i\in [n],k\in\{1,2\}}\left(\left\{x\in S:\frac{\partial g_k}{\partial x_i}(x)=0\right\}\right)\bigcup\\
    &\left\{x\in S:g_1(x)=g_2(x)\right\}.
\end{align*}
\end{enumerate}
\end{corollary}

The proof is presented in \Cref{app:mincol}.

Statement 1 can be used to compute the minimum of $g$ on any polytope $S$ with recursion. Since all components of $S^+$ have at most $(n-1)$ dimensions ($\partial S$ can be split into many facets), we can compress certain dimensions and recursively compute the $(n-1)$-dimensional case. Although we only present algorithms to solve cases where $t \leq 3$, we present statements 2 and 3 in a very general form. They are useful for further investigation of cases with $t>3$.

\subsection{Detailed algorithms}
\label{details:subsec:mixing-algo}

With all above preparations, we are now able to derive our algorithms for the mixing problem defined in \Cref{def:st-mixing}.

We first consider the $(1,w)$-mixing problem. Note that this problem can be directly transformed into a linear program given by \Cref{algo:concise-1dbpoint}. We denote the complexity of solving a standard-form linear program with $w$ variables and $m$ inequalities by $L(w,m)$, which is polynomial in $w$ and $m$. See, e.g. \cite{CLS21}. Then, the complexity of our $(1,w)$-mixing algorithm is $O(mnw+L(w,m))$.

\begin{algorithm}
\caption{$(1,w)$-mixing algorithm}
\label{algo:concise-1dbpoint}
\textbf{Input:} An $m\times n$ bimatrix game $(R,C)$, mixed strategies $x_1$ for the row player and $y_1, y_2,\dots , y_w$ for the column player.

\textbf{Output:} $\beta\in \Delta_w$ that minimizes $f(x_1,\beta_1 y_1+\dots +\beta_w y_w)$.

Calculate and store the $m$-dimensional vectors $Ry_1, \dots, Ry_w$ and the values $x_1^\T R y_1 \dots x_w^\T R y_w$.  \tcp{This can be done by direct matrix multiplication within  $O(mnw+m w)$ time.}

Solve the optimal $\alpha$ of the following linear program and output it.
\[
\begin{aligned}
    \min_\alpha \quad &t\\
    \text{s.t.}\quad & t\geq \max(C^\T x_1) -\alpha_1(x_1^\T Cy_1)-\dots -\alpha_w (x_w^\T C y_w),\\
    &\text{for every }i\in [m], \quad t\geq \alpha_1 (Ry_1)_i+\dots +\alpha_w (Ry_w)_i -\alpha_1(x_1^\T R y_1)-\dots -\alpha_w (x_w^\T R y_w),\\
    &\text{for every }j\in [w], \quad \alpha_j\geq 0,\\
    &\alpha_1+\dots+\alpha_w=1.
\end{aligned}
\]
\tcp{The complexity is equivalent to solving a non-negative linear programming problem with $m+1$ constraints and $w+1$ variables.}
\end{algorithm}

To avoid detailed case-by-case discussions, we only give sketches of the $(2,2)$ and $(2,3)$-mixing algorithms here in \Cref{algo:concise-2dbpoint} and \Cref{algo:concise-3dbpoint}. The full process, correctness, and time-complexity analysis are presented in \Cref{app:prec-2-2-mixing} and \Cref{app:prec-2-3-mixing}.

\begin{algorithm}
\caption{$(2,2)$-mixing algorithm}
\label{algo:concise-2dbpoint}
\textbf{Input}: A size $m\times n$ bimatrix game $(R,C)$, mixed strategies $x_1,x_2$ for the row player and $y_1, y_2$ for the column player.

\textbf{Output}: $\alpha,\beta\in \Delta_2$ that minimizes $f(\alpha_1 x_1+\alpha_2 x_2,\beta_1 y_1+\beta_2 y_2)$.

Apply the $(2,n)$-separation algorithm (see \Cref{app:prec-2-m-separation}) for $\alpha$ that outputs separated polytopes $P_i^R$, where $i\in [n]$ (actually intervals of $\alpha_1$). \tcp{Time complexity $O(n\log n)$}

Apply the $(2,m)$-separation algorithm (see \Cref{app:prec-2-m-separation}) for $\beta$ that outputs separated polytopes $P_j^C$, where $j\in [m]$ (actually intervals of $\beta_1$). \tcp{Time complexity $O(m\log m)$}

Compute the exact form of $F_i(\alpha,\beta)=f_i(\alpha x_1+(1-\alpha) x_2, \beta y_1+(1-\beta)y_2)$, where $i\in \{R,C\}$. \tcp{Time complexity $O(mn)$}

\For{$i=1:n$, $j=1:m$}
{
    Minimize $f$ in each grid $P_i^R \times P_j^C$. Apply statement 3 in \Cref{details:cor:mincol}. It suffices to scan the following regions:\\
    
    (1) Points with $\partial F_k(\alpha,\beta)/\partial\alpha=0$ or $\partial F_k(\alpha,\beta)/\partial\beta=0$, where $k=R,C$.\\
    
    (2) The four vertices of its domain.\\
    
    (3) Points with $F_R(\alpha,\beta)=F_C(\alpha,\beta)$.\\
    
    \tcp{For details, see \Cref{state:append}.}
}
\tcp{We can show that each case can be done in constant time over $m,n$. Thus, the time complexity is $O(mn)$.}
Finally, compare the $f$-values of the minimum on the $mn$ grids and obtain the global minimum of $f$ on $\Delta_2\times \Delta_2$. \tcp{Time complexity $O(mn)$}
\end{algorithm}

\begin{algorithm}
\caption{$(2,3)$-mixing algorithm}
\label{algo:concise-3dbpoint}
\textbf{Input:} A size $m\times n$ bimatrix game $(R,C)$, mixed strategies $x_1,x_2$ for the row player and $y_1, y_2, y_3$ for the column player.

\textbf{Output:} $\alpha \in \Delta_2$, $\beta \in \Delta_3$ that minimizes $f(\alpha_1 x_1+\alpha_2 x_2,\beta_1 y_1+\beta_2 y_2+\beta_3 y_3)$.

Apply the $(2,n)$-separation algorithm (\Cref{app:prec-2-m-separation}) for $\alpha$ that outputs separated polytopes $P_i^R, i\in [n]$ (Actually are intervals for $\alpha_1$). \tcp{Time complexity $O(n\log n)$}

Apply the $(3,m)$-separation algorithm (\Cref{app:prec-3-m-separation}) for $\beta$ that outputs separated polytopes $P_j^C, j\in [m]$. \tcp{Time complexity $O(m^2\log m)$}

Compute the exact form of $F_i(\alpha,\beta,\gamma)=f_i(\alpha x_1+(1-\alpha) x_2, \beta y_1+\gamma y_2+(1-\beta-\gamma)y_3), i\in {R,C}$. \tcp{Time complexity $O(mn)$}

\For{$i=1:n$, $j=1:m$}
{
Minimize $f$ in each grid $P_i^R \times P_j^C$. Apply statement 2 in \Cref{details:cor:mincol}, it suffices to scan the following regions:
\begin{enumerate}
\item $(\alpha,\beta)$ belongs to side surfaces of $S$ and
\begin{enumerate}
\item either there exists $k\in{R,C}$ such that $\partial F_k/\partial \gamma=0$, or
\item $(\alpha,\beta)$ is in the intersection of side surfaces and top/bottom surfaces.
\end{enumerate}
\item $(\alpha,\beta)$ belongs to top/bottom surfaces of $S$ and
\begin{enumerate}
\item there exists $k\in{R,C}$ such that either $\partial F_k/\partial \alpha=0$ or $\partial F_k/\partial \beta=0$, or
\item $(\alpha,\beta)$ is in the intersection of side surfaces and top/bottom surfaces.
\end{enumerate}
\item $F_R(\alpha,\beta)=F_C(\alpha,\beta)$.
\item $\nabla F_R(\alpha,\beta)=0$ or $\nabla F_C(\alpha,\beta)=0$.\\
\tcp{For details, see \Cref{state:append2}}
\end{enumerate}
}
\tcp{We can show that each case can be done in $O(m)$ time. Thus, the time complexity is $O(m^2 n)$.}

Finally, compare the $f$-values of the minimum on the $mn$ grids, and obtain the global minimum of $f$ on $\Delta_2\times \Delta_3$. \tcp{Time complexity $O(mn)$}
\end{algorithm}

\subsection{Examples}\label{app:subsec:mixing-example}

From the procedure given in \Cref{fig:framework-search-mix}, we can modify the TS algorithm~\cite{TS07_0.3393NE} and the DFM algorithm~\cite{DFM22_0.3333NE} with the mixing algorithms.

Both algorithms try to minimize the objective function $f(x,y)$ by a descent procedure. The difficulty of such an attempt is that $f(x,y)$ is a nonconvex nonsmooth function. Gradients cannot be defined at non-differentiable points. Even worse, local minima are usually non-differentiable. Fortunately, we can calculate \emph{directional derivatives} of the function $f$ at each point and in each direction. Furthermore, we can compute the steepest descent direction using a linear program. Here, the "steepest descent direction" refers to a direction vector $(x'-x,y'-y)$ such that
\[(x',y')\in\argmin_{(x',y')} Df(x,y,x',y'),\]
where $Df(x,y,x',y')$ is the Dini directional derivative defined as the following limit:
\[\lim_{\alpha\downarrow 0}\frac{f(x+\alpha(x'-x),y+\beta(y'-y))-f(x,y)}{\alpha}.\]

Such a descent process will stop if the minimum directional derivative at $(x,y)$ is zero. If that is the case, we call $(x,y)$ a \emph{stationary point}. Usually, stationary points are hard to reach. Hence, the descent procedure actually tries to find a $\delta$-stationary point, at which $\min_{(x',y')} Df(x,y,x',y')\geq -\delta$. \cite{TS07_0.3393NE} showed that the descent procedure finds a $\delta$-stationary point in time polynomial of $m,n$ and $1/\delta$. Since we calculate the steepest descent direction by a linear program, we can also calculate the optimal solution of its dual program. The dual solution can be expressed as $(\rho, w,z)$, where $\rho\in[0,1]$, $w\in\br_{R}(y)$ and $z\in\br_{C}(x)$.

\begin{example}[The modified TS algorithm~\cite{TS07_0.3393NE}]\label{example:modified-TS}
~
\begin{itemize}
    \item \emph{Search phase}: Use the descent procedure to compute a $\delta$-approximate stationary point $(x_s, y_s)$ of the function $f$. Then, by linear programming, compute the dual strategy profile $(w, z)$ as well as the dual parameter $\rho$. 
    \item \emph{Mixing phase}:  Directly apply the $(2, 2)$-mixing algorithm on the strategies $x_s$, $w$ of the row player and $y_s$, $z$ of the column player. The algorithm outputs the optimal mixing parameters $\alpha^*$, $\beta^*$ that minimize the $f$ value on all possible mixing.
\end{itemize}

\end{example}

\begin{example}[The modified DFM algorithm~\cite{DFM22_0.3333NE}]\label{example:modified-DFM}
~
\begin{itemize}
    \item \emph{Search phase}\footnote{To simplify the further discussions, we state it in a slightly different form from the original literature.}: The DFM algorithm is an adjustment of the TS algorithm. Besides the found $(x_s,y_s), (w,z)$ pairs, if $f_R(w,z)\geq f_C(w,z)$, it also considers the strategy $\hat{y}=\frac{1}{2}(y_s+z)$ and selects $\hat{w}\in \br_C(\hat{y})$. The operations in the case $f_R(w,z)< f_C(w,z)$ are symmetric.
    \item \emph{Mixing phase}: Directly apply the $(2,3)$-mixing algorithm on the strategies $x_s,w,\hat{w}$ of the row player and $y_s, z$ of the column player. The algorithm outputs the optimal mixing parameters $\alpha, \beta$ that minimize the $f$ value over all the possible mixes.
\end{itemize}

\end{example}

\section{Approximation analysis}
\subsection{Linearization}
\label{details:subsec:linearization}
To begin with, we propose some basic notions for linearization. Recall the concept of \emph{barycentric coordination} in mixing operation defined in \Cref{sec:pre}. In the discussion below, if we refer to "barycentric coordination", we mean the operation defined by this coordination.

Linearization aims at giving an upper bound of a convex function which meets the function values on the vertices, illustrated in \Cref{fig:linearization-comp} below.

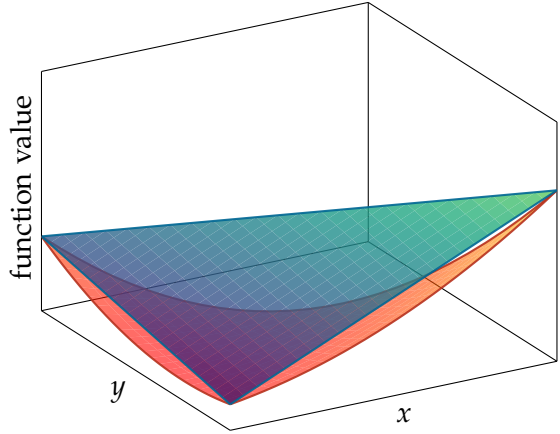
\begin{figure}[ht]
    \centering  
    \input{figures/linearization_comp}
    \caption{The value of a convex function controlled by the linear function over the vertices.}
    \label{fig:linearization-comp}
\end{figure}

\begin{proposition}
\label{details:prop:linear-tight}
    For any convex function $h(\alpha)$ defined on $\Delta_w$ and any $\alpha\in\Delta_w$, $h(\alpha)\leq \alpha_1 h(e_1)+\dots+\alpha_w h(e_w)$. The equality holds if $h$ is linear in $\alpha$.
\end{proposition}
\begin{proof}
    By the definition of convexity, we note that for any $\beta_1,\beta_2 \in \Delta_w$ and $\alpha\in \Delta_2$, we have $h(\alpha_1 \beta_1 +\alpha_2 \beta_2)\leq \alpha_1 h(\beta_1)+\alpha_2 h(\beta_2)$. Thus, we can gradually decompose $h(\alpha)$ onto the components of $e_1,\dots, e_w$, and obtain the conclusion. 
    
    Besides, if $h$ is linear, then by linearity, $h(\alpha_1 \beta_1 +\alpha_2 \beta_2)= \alpha_1 h(\beta_1)+\alpha_2 h(\beta_2)$ holds for any $\beta_1,\beta_2 \in \Delta_w$ and $\alpha\in \Delta_2$. Similarly, by the decomposition, we obtain the equality condition.
\end{proof}

\subsection{Designing an upper bound constructor}
\label{details:subsec:auxiliary-design}

As the counterpart to $(1,w), (2,2),(2,3)$-mixing algorithms, in this part, we propose upper bound constructors for cases $(1,w)$, $(2,2)$, $(2,3)$ to analyze the approximation bound. In fact, we are able to propose the upper bound constructor for arbitrary general case $(v,w)$. The idea is to use the linearization approach proposed in \Cref{details:subsec:linearization} to give an upper bound to $f$ with respect to its values over the vertices.

\subsubsection{A brief overview of the techniques}\label{details:subsubsec:brief-overview-auxiliary-design}

For case $(1,w)$, the upper bound can be derived easily. Our task is to give an upper bound for the minimum of $f$ on the region $\mathcal{A}=\{(x_1,\alpha_1 y_1+ \dots +\alpha_w y_w), \beta\in \Delta_w\}$, denoted by $f^*_{\mathcal{A}}$. We can show that both $f_R$ and $f_C$ are convex in $\alpha$. Thus a direct use of the linearization gives a linear upper bound function for both $f_R$ and $f_C$, and further an upper bound for $f_*^{\mathcal{A}}$ by linear programming.

For general case $(v,w)$, the region $\mathcal{A}$ is modified to $\{(\alpha_1 x_1+ \dots +\alpha_v x_v,\beta_1 y_1+ \dots +\beta_w y_w), \alpha\in \Delta_v, \beta\in \Delta_w\}$. However, in functions $f_R, f_C$, only max terms are convex in $(\alpha,\beta)$, while the bi-linear term $x^T R y$ and $x^T C y$ contains the cross term of $\alpha,\beta$, and thus is no longer convex.

To overcome this difficulty, we only apply linearization to the max term of $f_R$ and $f_C$, but we leave the bilinear term unaltered. Then, we propose two solutions:
\begin{enumerate}
\item \emph{Simple upper bound constructor}: We restrict $f$ to the boundary\footnote{For a formal definition of these geometric concepts, see \Cref{app:def-geo}.} of $\mathcal{A}$ and consider the faces of this boundary one by one. On each face, either $\alpha$ or $\beta$ is fixed. Thus, the bilinear function deteriorates to a linear function, which allows us to apply the result in case $(1,w)$ (or $(v,1)$). The upper bound given in this way is the minimum of several linear programs (LP).
\item \emph{Strong upper bound constructor}: We do not restrict $f$. The upper bound given in this way becomes the minimum of several quadratically constrained linear programs (QCLP).
\end{enumerate}

Clearly, simple upper bound constructors are easier to analyze and computationally tractable than the stronger ones. In fact, most works in the literature follow a similar but ad hoc analysis approach to our simple upper bound constructors.

The detailed algorithms are given in the follows sections.

\subsubsection{Case \texorpdfstring{$(1,w)$}{(1,w)}}

\begin{proposition}[$(1,w)$-upper bound constructor]\label{prop:1-w-aux-mixing}

Given strategies $x_1 \in \Delta_m$ and $y_1,\dots ,y_w \in \Delta_n$, the region defined by their convex combinations is $\mathcal{A}=\{(x_1, \alpha_1 y_1+\dots +\alpha_w y_w) \mid \alpha \in \Delta_w\}$. Let $f_*^{\mathcal{A}}$ be the minimum of $f$ on $\mathcal{A}$. Solve the following linear program:
\[
\begin{aligned}
    \min & \quad h\\
    \text{s.t.} &\quad h \geq \alpha_1 f_R(x_1,y_1)+\dots +\alpha_w f_R(x_1,y_w),\\
    &\quad h \geq \alpha_1 f_C(x_1,y_1)+\dots +\alpha_w f_C(x_1,y_w),\\
     &\quad \alpha \in \Delta_w.
\end{aligned}
\]

Denote by $\alpha^*, h^*$ the optimal solution and objective value to the linear program above. Then we have: $f^{\mathcal{A}}_{*} \leq f(x_1, \alpha^*_1 y_1+\dots +\alpha^*_w y_w) \leq h^*$. The inequalities hold simultaneously when $f_R$ is linear in $\alpha$. 
\end{proposition}

\begin{proof}
Over the region $\mathcal{A}$, we have $f_C(x_1, \alpha_1 y_1+\dots +\alpha_w y_w)=\max\{C^T x_1\}-x_1^T C(\alpha_1 y_1+\dots +\alpha_w y_w)=\alpha_1 f_C(x_1,y_1)+\dots +\alpha_w f_C(x_1,y_w)$, which is linear on $\alpha$. Besides, $f_R(x_1, \alpha_1 y_1+\dots +\alpha_w y_w)=\max\{R(\alpha_1 y_1+\dots +\alpha_w y_w)\}-x_1^T R(\alpha_1 y_1+\dots +\alpha_w y_w)$, which is convex on $\alpha$. Thus, by linearization in \Cref{details:prop:linear-tight}, on region $\mathcal{A}$, $f_R(x_1, \alpha_1 y_1+\dots +\alpha_w y_w)\leq \alpha_1 f_R(x_1,y_1)+\dots +\alpha_w  f_R(x_1,y_w)$. 

From the deduction above, the optimal solution $h^*$ of the linear program in the proposition gives an upper bound of $f$. Specifically, taking $\alpha=\alpha^*$ in the inequalities above, we have $f(x_1,\alpha^*_1 y_1+...+\alpha^*_w y_w) \leq h^*$. Besides, from \Cref{details:prop:linear-tight}, the equalities hold when $f_R$ is linear.
\end{proof}

The readers can refer to \Cref{fig:12-upper-bound} as an illustration for the case $w=2$.

\begin{figure}[htb]
    \centering
    \input{figures/1-2-upper-bound}
    \caption{Illustration of $(1,2)$-upper bound constructor.}
    \label{fig:12-upper-bound}
\end{figure}
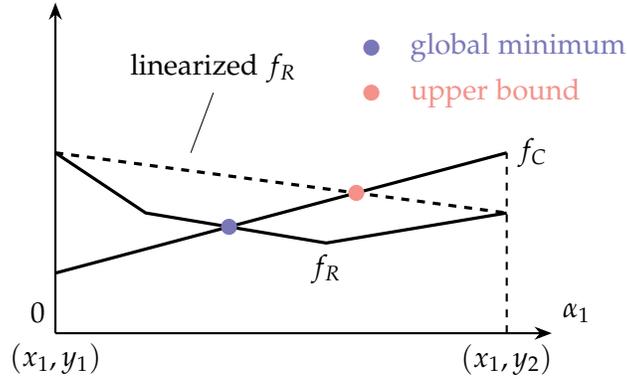


\begin{remark}
\label{rk:linear-decomposition}
One of the key properties of linear programming is that there exists an optimal solution on the boundary of the domain. Applying this property reversely, we can actually show that \Cref{prop:1-w-aux-mixing} is equivalent to using the $(1,2)$-upper bound constructor on each edge of the domain.
\end{remark}


\subsubsection{Case \texorpdfstring{$(v,w)$}{(v,w)}}

As discussed in \Cref{details:subsubsec:brief-overview-auxiliary-design}, there are two approaches for the general case $(v,w)$. One gives a simple upper bound constructor and another gives a strong upper bound constructor. In \Cref{fig:ubc}, we illustrate the relationship between the minimum $f$, simple upper bound constructor, and strong upper bound constructor on a section.

\begin{figure}[ht]
    \centering
    \input{figures/upper-bound-constructor}
    \caption{Relationships of function $f$, simple upper bound, and strong upper bound in case $(2,2)$. We only show a section of figure on line $\alpha_1=\beta_1$. The linearization approach only applies on max terms. The simple upper bound only considers the boundary (i.e., two endpoints) while the strong upper bound considers the whole domain (i.e., the whole interval).}
    \label{fig:ubc}
\end{figure}
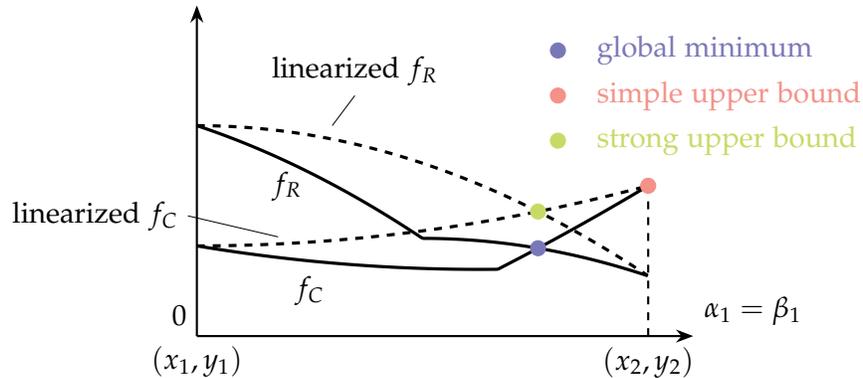

\begin{proposition}[simple $(v,w)$-upper bound constructor]\label{prop:v-w-aux-mixing}
Consider strategies $x_1, \dots, x_v\in \Delta_m$, $y_1,\dots, y_w\in \Delta_n$. On the region $\mathcal{A}=\{(\alpha_1 x_1+ \dots +\alpha_v x_v,\beta_1 y_1+ \dots +\beta_w y_w), \alpha\in \Delta_v, \beta\in \Delta_w\}$, denote the minimum of $f$ by $f_*^{\mathcal{A}}$. 

An upper bound of $f_*^{\mathcal{A}}$ is given by the minimum of $f$ over all the faces of $\mathcal{A}$, given by 
following $v+w$ regions:
\[
\begin{aligned}
    \{(x_i,\beta_1 y_1+ \dots +\beta_w y_w), \beta\in \Delta_w\},\quad i\in [v],\\
    \{(\alpha_1 x_1+ \dots +\alpha_v x_v, y_j), \alpha\in \Delta_v\},\quad j\in [w].
\end{aligned}
\]
Moreover, on each region, the minimum of $f$ is bounded by the output of the $(1,w)$ or $(v,1)$-upper bound constructor.
\end{proposition}

\begin{proof}
   Denote the boundary of $\mathcal{A}$ by $\partial \mathcal{A}$. By the property of simplex, $\partial \mathcal{A}$ is composed of the facets of $\mathcal{A}$, namely $\{(x_i,\beta_1 y_1+ \dots +\beta_w y_w):\beta\in \Delta_w\}$, $i\in [v]$ and $\{(\alpha_1 x_1+ \dots +\alpha_v x_v, y_j):\alpha\in \Delta_v\}$, $j\in [w]$. Since $\partial \mathcal{A}\subseteq \mathcal{A}$, $f_*^{\mathcal{A}}\leq f_*^{\partial \mathcal{A}}$.   
    
    Now, on each facet, directly apply the results in \Cref{prop:1-w-aux-mixing} and obtain an upper bound given by a linear program. Combine all these upper bounds, we obtain an upper bound of $f_*^{\partial \mathcal{A}}$, which is also an upper bound of $f_*^{\mathcal{A}}$.
\end{proof}

\begin{remark}
\label{rk:linear-decomposition+}
Due to \Cref{rk:linear-decomposition}, by further decomposing the $(1,w)$ and $(v,1)$-upper bound constructors, the $(v,w)$-upper bound constructor can be decomposed into $(1,2)$ or $(2,1)$-upper bound constructors over all edges of the domain.
\end{remark}


\begin{proposition}[strong $(v,w)$-upper bound constructor]\label{prop:v-w-strong-aux-mixing}
Consider strategies $x_1, \dots, x_v\in \Delta_m$, $y_1,\dots, y_w\in \Delta_n$. On the region $\mathcal{A}=\{(\alpha_1 x_1+ \dots +\alpha_v x_v,\beta_1 y_1+ \dots +\beta_w y_w), \alpha\in \Delta_v, \beta\in \Delta_w\}$, denote the minimum of $f$ by $f_*^{\mathcal{A}}$. An upper bound of $f_*^{\mathcal{A}}$ is given by the following program:
\[
\begin{aligned}
    \min & \quad h\\
    \text{s.t.} & \quad h \geq \sum_{i=1}^v \sum_{j=1}^w \alpha_i \beta_j f_R(x_i,y_j),\\
    & \quad h \geq \sum_{i=1}^v \sum_{j=1}^w \alpha_i \beta_j f_C(x_i,y_j),\\
    &\quad \alpha \in \Delta_w,\\
    &\quad \beta \in \Delta_v.
\end{aligned}
\]
\end{proposition}

\begin{proof}
The idea is to only apply the linearization on the max terms. For any point in $\mathcal{A}$, we have:
\[
\begin{aligned}
    &f_R(\alpha_1 x_1+ \dots +\alpha_v x_v,\beta_1 y_1+ \dots +\beta_w y_w) \\
    = &\max\{ R(\beta_1 y_1+ \dots +\beta_w y_w)\}-(\alpha_1 x_1+ \dots +\alpha_v x_v)^\T R (\beta_1 y_1+ \dots +\beta_w y_w)\\
     \leq & \beta_1 \max\{R y_1\}+ \dots +\beta_w \max\{R y_w\}-(\alpha_1 x_1+ \dots +\alpha_v x_v)^\T R (\beta_1 y_1+ \dots +\beta_w y_w)\\
    = &\sum_{i=1}^v \sum_{j=1}^w \alpha_i \beta_j f_R(x_i,y_j).
\end{aligned}
\]

Similarly, we have $f_C(\alpha_1 x_1+ \dots +\alpha_v x_v,\beta_1 y_1+ \dots +\beta_w y_w)\leq \sum_{i=1}^v \sum_{j=1}^w \alpha_i \beta_j f_C(x_i,y_j)$. Thus, $f_*^{\mathcal{A}}=\min_{(x,y)\in \mathcal{A}} \max\{ f_R(x,y), f_C(x,y)\}$ is bounded by the optimal value of the program above.
\end{proof}

Note that the upper bound constructed above only depends on the $f_R, f_C$ values on the vertices of the domain, i.e., the strategy profiles given by the search phase. However, this QCLP is much harder to solve than LP in \Cref{prop:v-w-aux-mixing}.

\begin{remark}\label{remark:simple-and-strong-rel}
The relationship between simple upper bound constructors and the strong ones is as follows. Set $\beta_j=1$ and other $\beta_i$'s to $0$. Then constraints of the program in \Cref{prop:v-w-strong-aux-mixing} becomes $h\geq\sum_{i=1}^v\alpha_i f_R(x_i,y_j)$, $h\geq \sum_{i=1}^v\alpha_i f_R(x_i,y_j)$, and $\alpha\in\Delta_v$, which are exactly constraints in the program of simple $(w,1)$-upper bound constructor (\Cref{prop:v-w-aux-mixing}). Thus, the simple upper bound constructor is a strong upper bound constructor additionally constrained on the boundary of the domain.
\end{remark}

\begin{remark}\label{remark:2-2-contructor-with-eq}
A more practical consideration of the strong $(2,2)$-upper bound constructor is to add an additional constraint $\alpha=\beta$ in the program of \Cref{prop:v-w-strong-aux-mixing}. With more constraints, the optimal $h$ in \Cref{prop:v-w-strong-aux-mixing} can only be larger, which is thus again an upper bound. As a return, the program now becomes a univariate program, which is much easier to solve.
\end{remark}

\begin{remark}
In the spirit of \Cref{prop:v-w-aux-mixing}, we can also have a hybrid version of \Cref{prop:v-w-strong-aux-mixing}. For example, a hybrid strong $(w,v)$-upper bound constructor can be given by strong $(2,2)$-upper bound constructors as follows. Suppose the region to be minimized on is $\mathcal A=\{(\alpha_1 x_1+ \dots +\alpha_v x_v,\beta_1 y_1+ \dots +\beta_w y_w), \alpha\in \Delta_v, \beta\in \Delta_w\}$. Then an upper bound of $f_*^\mathcal A$ is given by the minimum of upper bounds on the following regions: $\mathcal A_{i_1,i_2,j_1,j_2}=\{(\alpha x_{i_1}+(1-\alpha) x_{i_2},\beta y_{j_1}+ (1-\beta) y_{j_2}), \alpha\in [0,1], \beta\in [0,1]\}$, where $i_1,i_2\in[v]$ and $j_1,j_2\in[w]$.
\end{remark}

\subsection{Analyzing an upper bound constructor}
\label{details:subsec:auxiliary-minmax}

After the intuitive linear programming description of the upper bound constructor, now we compute the explicit form of $(v,w)$-upper bound.

From \Cref{rk:linear-decomposition+}, we know that any upper bound constructor can be viewed as taking the minimum of the output of several $(1,2)$-upper bound constructors. Therefore, we only need to present the explicit solution of the $(1,2)$-upper bound constructor.

\begin{proposition}
\label{prop:minmax-mixing}
Given strategies $x_1\in \Delta_m$ and $y_1,y_2\in \Delta_n$, denote $a_R=f_R(x_1,y_1)$,  $a_C=f_C(x_1,y_1)$, $b_R=f_R(x_1,y_2)$, $b_C=f_C(x_1,y_2)$. The solution $h^*$ of $(1,2)$-upper bound constructor on strategy is:
\begin{enumerate}
\item If $a_C\leq a_R$ and $b_C\leq b_R$, then $h^*= \min\{a_R,b_R\}$.
\item If $a_R\leq a_C$ and $b_R\leq b_C$, then $h^*= \min\{a_C,b_C\}$.
\item If $a_R+b_C-a_C-b_R=0$, then $h^*=\min\{\max\{a_C,b_C\},\max\{a_R,b_R\}\}$.
\item Otherwise, 
$h^*=\frac{a_Rb_C-a_Cb_R}{a_R+b_C-a_C-b_R}$.
\end{enumerate}

Conversely, when $f_R$ and $f_C$ are both linear in parameter $t$, the equalities of $f_*^{AB}$ in all three cases hold.
\end{proposition}

The proof is given in \Cref{app:minmax-mixing}. 

For (hybrid) strong upper bound constructors, the explicit form also exists and is purely given by the $f_R, f_C$ values over the vertices of the domain. Of course, it can be solved by a symbolic calculator, such as Mathematica. However, there is no neat expression for general cases. For practical use, we almost always use hybrid upper bound constructors given by the strong $(2,2)$-upper bound constructor in with an additional constraint $\alpha=\beta$, as stated in \Cref{remark:2-2-contructor-with-eq}. Now the optimal $h$ is in the form of $\min_{\alpha\in[0,1]}\max\{a_1\alpha^2+a_2\alpha+a_3,b_1\alpha^2+b_2\alpha+b_3\}$, i.e., the minimum value of the maximum of two univariate quadratic functions. Using \Cref{details:cor:mincol}, the minimum value of this program can be easily solved even by hand.

\subsection{Examples}
\label{details:subsec:aux-example}
In this part, we analyze the approximation of the modified TS algorithm in \Cref{example:modified-TS} and the modified DFM algorithm in \Cref{example:modified-DFM}. 

\begin{example}[Approximation analysis of the modified TS algorithm]
\label{ex-TSalgo}
In this example, we analyze the approximation bound of the modified TS algorithm in \Cref{example:modified-TS}. Additionally, we provide a tight instance for the modified TS algorithm. The mixing region is $\mathcal{A}:= \{(\alpha x_s+(1-\alpha) w, \beta y^*+(1-\beta) b_2): \alpha, \beta\in[0,1]]\}$. Now, we derive an upper bound of $f^*_\mathcal{A}$ following the same procedure as in \Cref{sec:approx-anal}.

\vspace{10pt}
\textbf{Step 1: Construct the upper bound.}

First, give the name alias of the vertices. Denote $\za=f_R(x_s,y_s)$, $\zb=f_C(x_s,y_s)$, $\zc=f_R(x_s,z)$, $\zd=f_C(x_s,z)$, $\ze=f_R(w,y_s)$, $\zf=f_C(w,y_s)$, $\zg=f_R(w,z)$, $\zh=f_C(w,z)$. The positions of them are illustrated in \Cref{fig:square}.

\begin{figure}[ht]
    \centering
    \includegraphics[width=0.6\textwidth]{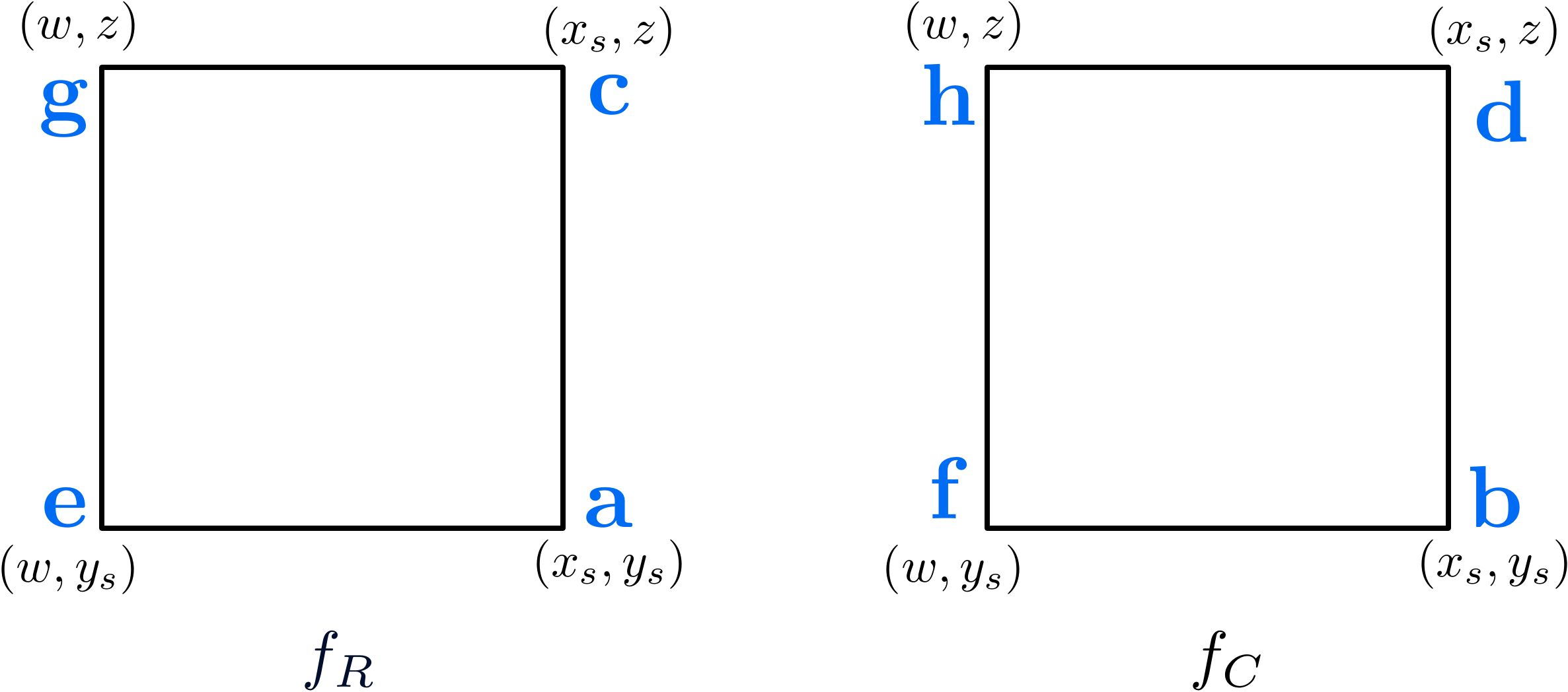}
    \caption{Positions of little letter aliases in the modified TS algorithm.}
    \label{fig:square}
\end{figure}

Then, by directly taking $v=w=2$ in \Cref{prop:v-w-aux-mixing} and \Cref{rk:linear-decomposition}, we can derive the $(2,2)$-upper bound constructor. Then $f_\mathcal A^*$ is bounded by $h^*$ given by the optimal value of the following optimization problem:
\[
\begin{aligned}
    \min_{\alpha_1,\alpha_2,\alpha_3,\alpha_4}&\quad \min\{b_1,b_2,b_3,b_4\}\\
    \text{s.t.} &\quad b_1 \geq \alpha_1 \za + (1-\alpha_1)\ze, b_1 \geq \alpha_1 \zb + (1-\alpha_1)\zf,\\
    &\quad b_2 \geq \alpha_2 \za + (1-\alpha_2)\zc, b_2 \geq \alpha_2 \zb + (1-\alpha_2)\zd,\\
    &\quad b_3 \geq \alpha_3 \zg + (1-\alpha_3)\zc, b_3 \geq \alpha_3 \zh + (1-\alpha_3)\zd,\\
    &\quad b_4 \geq \alpha_4 \zg + (1-\alpha_4)\ze, b_4 \geq \alpha_4 \zh + (1-\alpha_4)\zf,\\
    &\quad 0\leq \alpha_1, \alpha_2, \alpha_3, \alpha_4 \leq 1.
\end{aligned}
\]

\vspace{10pt}
\textbf{Step 2: Extract relations from the search phase.}

Although the descent algorithm in fact finds a $\delta$-stationary point, we actually consider stationary points when we make the bound analysis. The upper bound will increase by at most $\delta$. In this way, we can simplify notations.

For the strategies $x_s$, $y_s$, $w$, and $z$ given by the search phase of the TS algorithm, we can prove the following properties:

\begin{lemma}[Proposition 1 in \cite{CDH+21_0.3393tight}]
\label{lemma:strongineq}
\begin{align*}
    f_R(x_s,y_s)=f_C(x_s,y_s)=f(x_s,y_s)
    \leq &\rho(w^\T R y'-(x')^\T Ry_s-x_s^\T R y'+x_s^\T R y_s)+\\
    &\quad(1-\rho)((x')^\T C z-(x')^\T C y_s-x_s^\T C y'+x_s^\T C y_s).
\end{align*}
\end{lemma}

We can also rewrite it in the form of objective functions:
\begin{align*}
    f(x_s,y_s)\leq &\rho(f_R(x_s,y')-f_R(w,y')+f_R(x',y_s)-f_R(x_s,y_s))+\\
    &\quad(1-\rho)(f_C(x',y_s)-f_C(x',z)+f_C(x_s,y')-f_C(x_s,y_s)).
\end{align*}

Define 
\begin{align*}
    F_I(\alpha, \beta) &:= f_I\left(\alpha w+(1-\alpha) x_s, \beta z+(1-\beta) y_s\right), \quad I \in \{R, C\}, \\
    F(\alpha, \beta) &:= \max\left\{F_R(\alpha, \beta), F_C(\alpha, \beta)\right\},
\end{align*}
where $\alpha, \beta \in [0,1]$. Then we can directly use the coordinates $(\alpha, \beta)$ of each point to express the value of $f$ on it. The following lemmas can be derived from \Cref{lemma:strongineq}:

\begin{lemma}\label{lemma:frfc}
\begin{enumerate}
    \item $F(0,0) = F_R(0,0) = F_C(0,0)$.
    \item $F_R(1,0) = F_C(0,1) = 0$.
    \item $0 \leq F_I(\alpha, \beta) \leq 1$, for $I \in \{R, C\}$ and $\forall \alpha, \beta \in [0,1]$.
\end{enumerate}
\end{lemma}

\begin{lemma}[Lemma 2 in \cite{CDH+21_0.3393tight}]\label{lemma:increasing}
For the functions $F_R$ and $F_C$ defined above:

\begin{enumerate}
\item Given $\beta$, $F_{C}(\alpha, \beta)$ is increasing, convex, and linear-piecewise in $\alpha$; $F_{R}(\alpha, \beta)$ is decreasing and linear in $\alpha$.

\item Given $\alpha$, $F_{R}(\alpha, \beta)$ is increasing, convex, and linear-piecewise in $\beta$; $F_{C}(\alpha, \beta)$ is decreasing and linear in $\beta$.
\end{enumerate}
\end{lemma}

We mention that these inequalities fully characterize strategies $x_s$, $y_s$, $w$, and $z$, thus providing crucial tools for the bound analysis. 

Now we can find relations of these $f_R$'s and $f_C$'s. We notice that $\za=\zb$, $\zd=\ze=0$, $\za, \zb, \zc, \zd, \ze, \zf, \zg, \zh\in [0,1]$ (\Cref{lemma:frfc}), $\zc\geq \zg$, $\zc\geq \za$, $\zf\geq \zh$, $\zf\geq \zb$ (\Cref{lemma:increasing}). By symmetry, we assume without loss of generality that $\zg\geq \zh$. Finally, we have $\za\leq \rho(\zc-\zg), \za\leq (1-\rho)(\zf-\zh)$ (\Cref{lemma:strongineq}).

\vspace{10pt}
\textbf{Step 3: Write down the optimization problem of the upper bound.}

By the relations derived above, we can directly write down the optimization problem of the approximation bound. See \Cref{app:subsec:TS-formalization} for details. 

\vspace{10pt}
\textbf{Step 4-1: Eliminate the minimum operator over $\alpha_i$'s in $h^*$.}

On the edge $\{(\alpha x_s+(1-\alpha)w, y_s): \alpha \in [0,1]\}$, since $\za=\zb, \ze=0,\zf\geq \zb\geq 0$, by \Cref{prop:minmax-mixing}, the optimal value is $\za$, which is obtained at $\alpha^*=1$. Similarly, we can show that the optimal value on $\{(\alpha x_s+(1-\alpha)w, z): \alpha \in [0,1]\}$ is $\zg$, obtained at $\alpha^*=0$; the optimal value on $\{(x_s, \beta y_s+(1-\beta)z): \beta \in [0,1]\}$ is $\zb$, obtained at $\beta^*=1$; the optimal value on $\{(w, \beta y_s+(1-\beta)z): \beta \in [0,1]\}$ is $\frac{\zf\zg}{\zf+\zg-\zh}$, obtained at $\beta^*=\frac{\zg-\zh}{\zf-\zh+\zg}$. 

Thus, the upper bound $h^*$ is simplified to 
\[\min\left\{\za, \frac{\zf \zg}{\zf+\zg-\zh}\right\}.\]

\vspace{10pt}
\textbf{Step 4-2: Obtain the constant upper bound.}

Now the optimization problem of the approximation bound becomes
\begin{align*}
    \text{maximize}\quad &\min\left\{\za,\frac{\zf \zg}{\zf+\zg-\zh}\right\},\\
    \text{s.t.}\quad &0\leq \za,\zc,\zd,\zf,\zg,\zh,\rho\leq 1,\\
               & \za \leq \rho(\zc-\zg),\\
               & \za \leq (1-\rho)(\zf-\zh),\\
               & \zf\geq \za.
\end{align*}

This can be solved by a computer program. Below we show the solution manually. Let $\lambda=\zc-\zg$ and $\mu=\zf-\zh$. Then,
\begin{align}
    h^*=&\min \left\{\za, \frac{\zf \zg}{\zf+\zg-\zh}\right\}\notag\\
    \leq  & \min \left\{\rho\lambda, (1-\rho)\mu, \frac{\zf(\zc-\lambda)}{\zc+\mu-\lambda}\right\}\label{formula:TS:scale-1}\\
    \leq & \min \left\{\frac{\lambda\mu}{\lambda+\mu}, \frac{1-\lambda}{1+\mu-\lambda}\right\}. \label{formula:TS:scale-2}
\end{align}

By the same argument in \cite{TS07_0.3393NE}, the optimal value of \eqref{formula:TS:scale-2} is $b\approx 0.3393$, obtained when $\lambda=\lambda^*\approx 0.5825$ and $\mu=\mu^*\approx 0.8128$. That is the desired upper bound of the approximation.

Therefore, we have shown that:
\begin{theorem}
\label{details:thm:1/3-tight}
    The approximation bound of the modified TS algorithm is $0.3393+\delta$.
\end{theorem}

\vspace{10pt}
\textbf{One more step: Derive the tightness conditions.}

The analysis also provides us with the necessary condition for tightness. Tracing each step of the relaxation, we can easily obtain the tightness condition. From step \eqref{formula:TS:scale-2}, we have $\zf^*=\zc^*=1$, $\rho^*=\mu^*/(\lambda^*+\mu^*)\approx 0.5825$. Thus, $\zg^*=\zc^*-\lambda^*\approx 0.4175$, $\zh^*=\zf^*-\mu^*\approx 0.1872$. From step \eqref{formula:TS:scale-1}, we have $\za^*\approx 0.3393$. In this way, we uniquely determine the value at each point in the tight case. A more elaborate consideration is provided by \cite{CDH+21_0.3393tight}, which is further able to devise a generator that generates all tight instances.

We also mention that Theorem 1 in \cite{CDH+21_0.3393tight} gives a game instance $(R,C)$:
\[
R=\begin{pmatrix}
0.1 & 0 & 0 \\
0.1+b & 1 & 1 \\
0.1+b & \lambda^* & \lambda^*
\end{pmatrix}, \quad C=\begin{pmatrix}
0.1 & 0.1+b & 0.1+b \\
0 & 1 & \mu^* \\
0 & 1 & \mu^*
\end{pmatrix},
\]
such that on the square $\mathcal{A}$ defined by $x_s=y_s=(1,0,0)^\T, z=w=(0,0,1)^\T$, $f^*_{\mathcal{A}}$ is exactly $b\approx 0.3393$. Therefore, the approximation bound is tight.
\end{example}

\begin{example}[Approximation analysis of the modified DFM algorithm]\label{ex-modifiedDFM}

In this example, we analyze the approximation bound of the modified DFM algorithm in \Cref{example:modified-DFM}. Additionally, we provide a tight instance for the modified DFM algorithm. Due to symmetry, we only present the case where $\zg\geq \zh$, namely $f_R(w,z)\geq f_C(w,z)$.

Note that the mixing region is $\mathcal{A}:=\{(\alpha_1 x_s+\alpha_2 w +\alpha_3 \hat{w},\beta_1 y_s+ \beta_2 z+\beta_3 \hat{y}) \mid \alpha\in \Delta_3, \beta\in \Delta_3\}$. Now, we derive an upper bound of $f^*_\mathcal{A}$ following the same procedure as in \Cref{sec:approx-anal}.

\vspace{10pt}
\textbf{Step 1: Construct the upper bound.}

First, we present the following alias table.
\begin{center}
\centering
    \begin{minipage}{0.4\textwidth}
\centering
\[\begin{array}{c|c}
\text{original notation}&\text{new alias}\\
\hline
f_R(x_s,y_s) &\za\\
f_C(x_s,y_s) &\zb\\
f_R(x_s,z) &\zc\\
f_C(x_s,z) &\zd\\
f_R(w,y_s) &\ze\\
f_C(w,y_s) &\zf\\
f_R(w,z) &\zg\\
f_C(w,z) &\zh
\end{array}
\]
\end{minipage}
\begin{minipage}{0.5\textwidth}
\centering
\[\begin{array}{c|c|c}
\text{condition}&\text{original notation}&\text{new alias}\\
\hline
\multirow{6}{*}{$\zg\geq \zh$}
&f_R(\hat{w},z)&\zi \\
&f_R(\hat{w},y_s)&\zj \\
&f_C(\hat{w},z)&\zk \\
&f_C(\hat{w},y_s)&\zl\\
&f_R(w,\hat{y})&\zm\\
&f_R(x_s,\hat{y})&\zn\\
\hline
\multirow{6}{*}{$\zg <\zh$}
&f_C(w,\hat{z})&\zi \\
&f_C(x_s,\hat{z})&\zj \\
&f_R(w,\hat{z})&\zk \\
&f_R(x_s,\hat{z})&\zl\\
&f_C(\hat{x},y_s)&\zm\\
&f_C(\hat{x},z)&\zn
\end{array}\]
\end{minipage}
\end{center}

Here, we draw the prism $\mathcal{A}$ to show the positions of little letter aliases, as shown in \Cref{fig:prism}.

\begin{figure}[ht]
    \centering
    \includegraphics[width=0.9\textwidth]{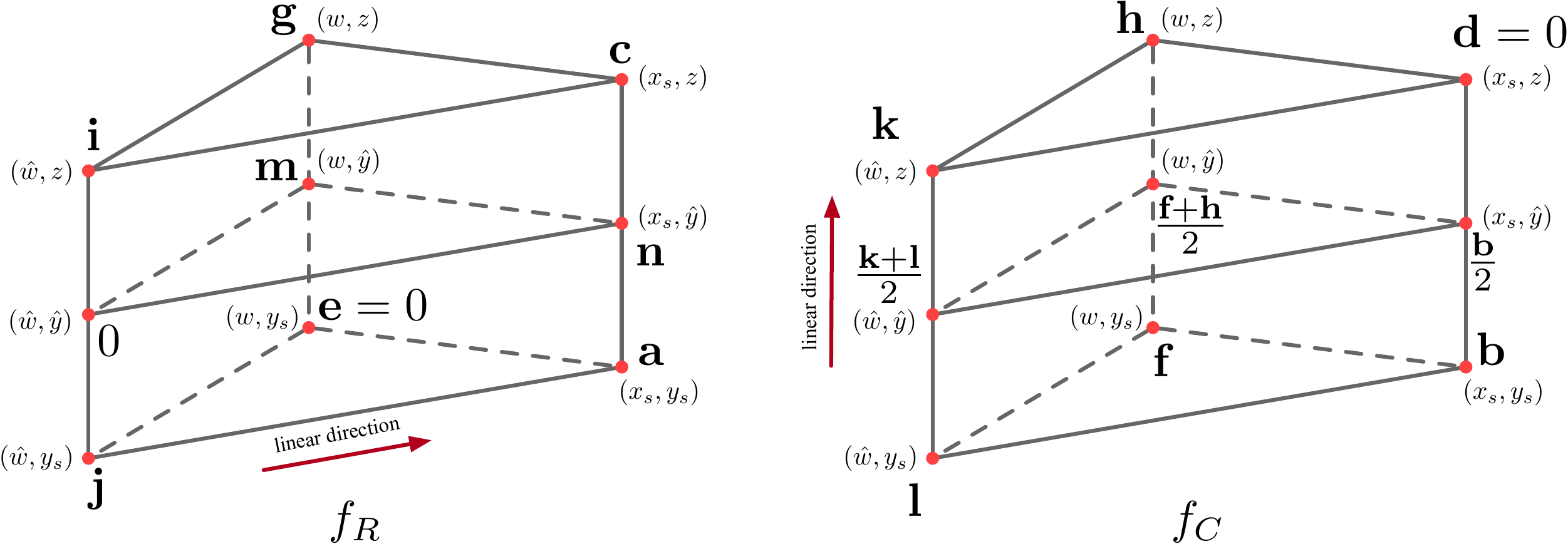}
    \caption{Positions of little letter aliases in the modified DFM algorithm (case $\zg\geq \zh$).}
    \label{fig:prism}
\end{figure}

Then, directly using \Cref{prop:v-w-aux-mixing} and \Cref{rk:linear-decomposition+}, we can derive the $(3,3)$-upper bound constructor, which is decomposed into $(1,2)$-upper bound constructors among all the edges of the triangular prism $\mathcal{A}$. The form is very complicated so we omit it here.

\vspace{10pt}
\textbf{Step 2: Extract relations from the search phase.}

Similar to \Cref{ex-TSalgo}, to simplify the notations, we here consider stationary points rather than $\delta$-stationary points. The relations is presented below. The proof is given in \Cref{app:DFM-relation}.

\begin{lemma}\label{lemma:DFM-relation}
Let $\zb=f_C(x_s,y_s)$, $\zd=f_C(x_s,z)$, $\ze=f_R(w,y_s)$. We have the following relations.
\begin{enumerate}
\item $\zd=\ze=0$, $\za=\zb$.
\item $\za,\zb,\zc,\zf,\zg,\zh,\zi,\zj,\zk,\zl, \zn, \zm \in [0,1]$.
\item $\zc\geq \zg$, $\zg\geq \zm$, $\zc\geq \zn\geq \za$, $\zf\geq \zh$, $\zf\geq \zb$.
\item $\za\leq \rho(\zc-\zg)$, $\zb\leq (1-\rho)(\zf-\zh)$, $\za\leq \rho \zj+(1-\rho)(\zl-\zk)$.
\item $f_R(\hat{w},\hat{y})=0$.
\item $\zg\geq \zi+\zj$.
\item $f_C(\hat{w},\hat{y})=\frac{1}{2}(\zl+\zk)$.
\end{enumerate}
\end{lemma}

\vspace{10pt}
\textbf{Step 3: Write down the optimization problem of the upper bound.}

With these relations, we can write down the optimization problem of the approximation bound. See \Cref{app:subsec:DFM-formalization} for details.

\vspace{10pt}
\textbf{Step 4-1: Eliminate the minimum operator over $\alpha_i$'s in $h^*$.}

Now we give the explicit solution of the $(3,3)$-upper bound constructor $h^*$. Each $(1,2)$-upper bound composing the $(3,3)$-upper bound constructor can be completely determined by \Cref{prop:minmax-mixing}. This can be calculated by a computer program. Again, since there are too many cases to present, we omit it here.

\vspace{10pt}
\textbf{Step 4-2: Obtain the constant upper bound.}

Now, the form of the optimization problem of the upper bound has no inner $\min_{\alpha_i}$ operator and thus can be solved by a numerical optimization solver, such as Mathematica.

The output of the computer program shows that the approximation bound of the modified DFM algorithm in \Cref{example:modified-DFM} is $1/3$. However, the correctness of the numerical optimization procedure is somewhat hard to check, so we also prove manually that the bound is exactly $1/3$ in \Cref{app:approx-DFM}.

\begin{theorem}
\label{details:thm:1/3-approx}
    The approximation bound of the modified DFM algorithm is $1/3+\delta$.
\end{theorem}

\begin{remark}\label{remark:optimal-1/2-DFM}
    One may wonder if it is possible to improve the approximation by choosing different $\hat{y}$ and $\hat{x}$. Indeed, we can generalize them to a parameterized form, e.g., $\hat{y}_\theta=\theta w+(1-\theta) y_s$. The letter $\theta$ can be regarded as a parameter. Involving $\theta$, \Cref{lemma:DFM-relation} and \Cref{lemma:1stpart} will change correspondingly. The optimization problem will also change, including a new variable $\theta$. However, numerical computation shows that the optimal value of the optimization problem is still $1/3$, attained at $\theta=1/2$. Thus, the modified DFM algorithm in \Cref{example:modified-DFM} makes the optimal choice. Once again, our method demonstrates its power to provide a quick result here.
\end{remark}

\vspace{10pt}
\textbf{One more step: Derive the tightness conditions.}

Based on the approximation analysis above, we identify a tight instance $(R,C)$ with the following matrices:

\begin{align}
R=\begin{pmatrix}
0 & 0 & 0 \\
0 & 0 & 1 \\
1 / 3 & 2 / 3 & 2 / 3
\end{pmatrix}, \quad C=\begin{pmatrix}
0 & 1 / 3 & 1 / 3 \\
0 & 0 & 1 / 3 \\
0 & 1 & 2 / 3
\end{pmatrix},\label{eq:tight-DFM}
\end{align}

By considering the degraded triangular prism $\mathcal{A}$, constructed with $x_s=y_s=(1,0,0)^\T$, $z=w=\hat{w}=(0,0,1)^\T$, and $\rho=1/2$, we find that $f^*_{\mathcal{A}}=1/3$. We note that this tight instance is stronger than that provided in \cite{chen_tightness_2023}, which is only tight for the original DFM algorithm but not for the modified DFM algorithm in \Cref{example:modified-DFM}. Thus, our method for tightness analysis is much stronger and generally applicable.

The proof of tightness is provided in \Cref{app:1/3-tight}.
\end{example}

%% file: figures/linear-pieces.tex
\begin{tikzpicture}
\begin{axis}[
        xlabel={$\beta_1$},
        ylabel={$\beta_2$},
        zlabel={max term},
        view={20}{45}
        ]
    
    \addplot3 [surf, opacity=0.8, colormap/summer] table {
        0 0 0.8
        0.3 0 0.56
        0.6 0 0.5
        1 0 0.6
        
        0 0.5 0.6
        0.3 0.75 0.15
        0.6 0.5 0.21
        1 0.5 0.3
        
        0 0.75 0.6
        0.3 0.75 0.15
        0.6 0.75 0.25
        1 0.75 0.4
        
        0 1 0.7
        0.3 1 0.5
        0.6 1 0.58
        1 1 0.9    
        };
\end{axis}
\end{tikzpicture}

%% file: figures/linearization_comp.tex
\begin{tikzpicture}[declare function={f(\x,\y)=0.5^2+\y*(1.5^2-0.5^2)+(2.5^2-0.5^2)*\x;g(\x,\y)=(2*\x-2*\y+0.5)^2;}]
\begin{axis}[
        xlabel={$x$},
        ylabel={$y$},
        zlabel={function value},
        ticks = none,
        view={-30}{30},
        ]

\begin{scope}
    \clip plot[variable=\x,domain=0:1.0] (\x,{1-\x},{g(\x,1-\x)}) -- plot[variable=\x,domain=0:1.0] ({1-\x},0,{g(1-\x,0.0)}) -- (0,0,{g(0.0,0.0)}) -- plot[variable=\x,domain=0:1.0] (0,\x,{g(0,\x)});
    \addplot3 [no marks, surf, fill opacity=0.6, 
    shader=faceted, faceted color = none, draw opacity=0., line width=0.,draw=none,
    colormap/autumn, 
    domain={0:1},
    ] {(2*x-2*y+0.5)^2};
\end{scope}

\addplot3 [domain=0:1.0, BrickRed, thick, samples=30, samples y=0, opacity = 0.9] (x,1-x,{g(x,1-x)});
\addplot3 [domain=0:1.0, BrickRed, thick, samples=30, samples y=0, opacity = 0.9] (x,0,{g(x,0)});
\addplot3 [domain=0:1.0, BrickRed, thick, samples=30, samples y=0, opacity = 0.9] (0,x,{g(0,x)});

\begin{scope}
    \clip plot[variable=\x,domain=0:1.0] (\x,{1-\x},{f(\x,1-\x)}) -- plot[variable=\x,domain=0:1.0] (\x,0,{f(\x,0.0)}) -- (0,0,{f(0.0,0.0)}) -- plot[variable=\x,domain=0:1.0] (0,\x,{f(0,\x)});
    \addplot3 [no marks, surf, fill opacity=0.8, 
    shader=faceted, faceted color = none, draw opacity=0., line width=0.,draw=none,
    colormap/viridis,
    domain={0:1}
    ] {0.5^2+y*(1.5^2-0.5^2)+(2.5^2-0.5^2)*x};
\end{scope}

\addplot3 [domain=0:1.0, MidnightBlue, thick, samples=30, samples y=0, opacity = 0.9] (x,1-x,{f(x,1-x)});
\addplot3 [domain=0:1.0, MidnightBlue, thick, samples=30, samples y=0, opacity = 0.9] (x,0,{f(x,0)});
\addplot3 [domain=0:1.0, MidnightBlue, thick, samples=30, samples y=0, opacity = 0.9] (0,x,{f(0,x)});

\end{axis}
\end{tikzpicture}

%% file: figures/1-2-upper-bound.tex
\begin{tikzpicture}[scale=4]
    \draw[-Stealth, thick] (0,0) -- (1.1*1.5,0) node[above right] {$\alpha_1$};
    \draw[-Stealth, thick] (0,0) -- (0,1.1);
    \draw[very thick] (0,0.2) -- (1*1.5,0.6) node[right] {$f_C$};
    \draw[very thick] (0,0.6) -- (0.2*1.5,0.4)  -- (0.6*1.5,0.3) node [below] {$f_R$}-- (1*1.5,0.4);
    \draw[very thick, dashed] (0,0.6) -- (1*1.5,0.4);
    \draw[ultra thin] (0.3*1.5,0.6) -- (0.35*1.5,0.8) node[above] {linearized $f_R$};
    
    \draw (0,0) node [above left] {$0$};
    \draw (0,0) node[below] {$(x_1,y_1)$};
    \draw[thick, dashed] (1*1.5,0) node[below] {$(x_1,y_2)$} -- (1*1.5,0.6);
    
    \draw[fill,Periwinkle] (0.384615*1.5,0.353846) circle [radius=0.025];
    \draw[fill,Periwinkle] (0.7*1.5,0.95) circle [radius=0.025] node[right] {\quad global minimum};
    
    \draw[fill,Salmon] (2/3*1.5,-0.2*2/3+0.6) circle [radius=0.025];
    \draw[fill,Salmon] (0.7*1.5,0.8) circle [radius=0.025] node[right] {\quad upper bound};
\end{tikzpicture}

%% file: figures/upper-bound-constructor.tex
\begin{tikzpicture}[scale=4,elegant/.style={smooth, very thick, samples=200}]
    \draw[-Stealth, thick] (0,0) -- (1.1*1.5,0) node[above right] {$\alpha_1=\beta_1$};
    \draw[-Stealth, thick] (0,0) -- (0,1.1);

    \draw[elegant, domain=0:1] plot({1.5*\x},{0.5*max(\x,1-\x)-0.5*\x^2+0.2});
    \draw (0.2*1.5,0.5) node {$f_R$};
    \draw[elegant, dashed, domain=0:1] plot({1.5*\x},{0.5-0.5*\x^2+0.2});
    \draw[ultra thin] (0.3*1.5,0.7) -- (0.35*1.5,0.8) node[above] {linearized $f_R$};

    \draw[elegant, domain=0:1] plot({1.5*\x},{0.5*max(\x,1-0.5*\x)+0.2*\x^2-0.2});
    \draw (0.25*1.5,0.15) node {$f_C$};
    \draw[elegant, dashed, domain=0:1] plot({1.5*\x},{0.5+0.2*\x^2-0.2});
    \draw[ultra thin] (-0.02*1.5,0.4)  node[left] {linearized $f_C$} -- (0.18*1.5,0.33);
    
    \draw (0,0) node [above left] {$0$};
    \draw (0,0) node[below] {$(x_1,y_1)$};
    \draw[thick, dashed] (1*1.5,0) node[below] {$(x_2,y_2)$} -- (1*1.5,0.5);
    
    \draw[fill,Periwinkle] ({sqrt(4/7)*1.5},{4/7*0.2-0.2+0.5*sqrt(4/7)}) circle [radius=0.025];
    \draw[fill,Periwinkle] (0.8*1.5,0.95) circle [radius=0.025] node[right] {\quad global minimum};
    
    \draw[fill,Salmon] (1*1.5,0.5) circle [radius=0.025];
    \draw[fill,Salmon] (0.8*1.5,0.8) circle [radius=0.025] node[right] {\quad simple upper bound};

    \draw[fill,SpringGreen] ({sqrt(4/7)*1.5},{4/7*0.2+0.3}) circle [radius=0.025];
    \draw[fill,SpringGreen] (0.8*1.5,0.65) circle [radius=0.025] node[right] {\quad strong upper bound};

\end{tikzpicture}

%% file: section/append/detail_proof.tex
\section{Missing proofs and algorithms}\label{app:miss-proof}

\subsection{The \texorpdfstring{$(2,m)$}{(2,m)}-separation algorithm}
\label{app:prec-2-m-separation}
This problem can be restated as a famous problem in computational geometry called \emph{envelope problem}, which is a special case of \emph{half-plane intersection problem}. The half-plane intersection problem can be solved with the plane sweep method in time $O(n\log n)$, see, e.g., Section 4.2 of~\cite{MOMM08}. For completeness, we restate the full algorithm here.

Specifically, suppose we are given two series $\{a_i\}_{i=1}^k,\{b_i\}_{i=1}^k$. We want to compute the breakpoints of function $h(x)=\max_{i\in[k]} \{a_ix+b_i\}(x\in[0,1])$ and the value of $h$ on these points. We present a method based on ideas from computational geometry.

First, we turn the case into $a_1< a_2<...< a_k$. To do so, reorder functions $\{a_ix+b_i\}_i$ so that $a_1\leq a_2\leq...\leq a_k$. Then we check all contiguous pairs $(a_i,a_{i+1})$. If $a_i=a_{i+1}$, then we delete the function with smaller $b_i$, since it is strictly smaller that the other one. By this procedure, we obtain $a_1< a_2<...< a_k$ in time $O(k\log k)$.

Let us use a list $w$ to memorize the breakpoints and a list $t$ to memorize the value of $h(x)$ on these points. Define $h_s(x)=\max_{i\in[s]}\{a_ix+b_i\}$. We use a recursion method to find the breakpoints by gradually updating the set of breakpoints from $h_1(x)$ to $h_k(x)=h(x)$. For the beginning, since $h_1(x)=a_1x+b_1$, we can initialize list $w$ with $w(0)=0,w(1)=1$ and list $t$ with $t(0)=b_1,t(1)=a_1+b_1$. Then we consider how to update from $h_s(x)$ to $h_{s+1}(x)$.

By assumption $a_1<a_2<...<a_s<a_{s+1}$, function $\Delta h_s(x)=h_s(x)-a_{s+1}x-b_{s+1}$ is continuous on $[0,1]$ and decreasing on every linear piece. So $\Delta h_s$ is decreasing on $[0,1]$ and has at most one zero point. Therefore, $h_s(x)$ has at most one intersection point with $a_{s+1}x+b_{s+1}$. If such point exists, say $x^*$, then we have $h_s(x)\leq a_{s+1}x+b_{s+1}$ if and only if $x\geq x^*$. So, we only need to add $x^*$ into list $w$ and delete all the points in list $w$ which belong to $[x^*,1)$. Similarly, we add $a_{s+1}x^*+b_{s+1}$ into the list $t$, update the value corresponding to $1$ with $a_{s+1}+b_{s+1}$ and delete all values between them. The geometric illustration of such a procedure is given in \Cref{fig:2-m-sep}. 
\begin{figure}[ht]
    \centering
    \includegraphics[width=0.8\textwidth]{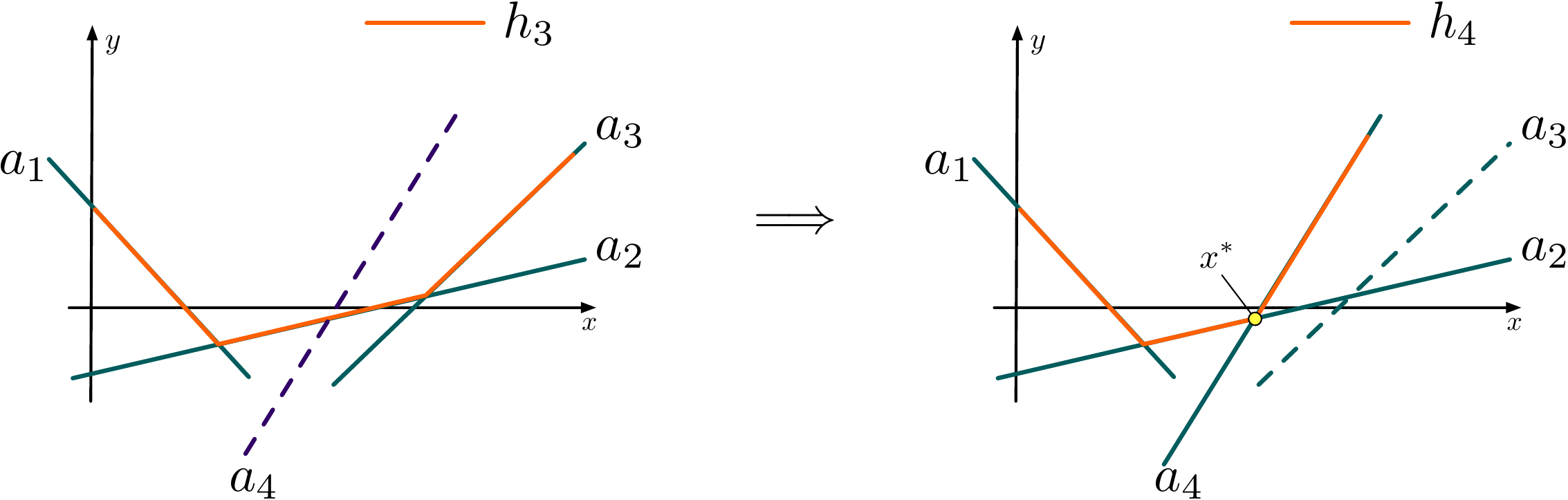}
    \caption{Illustration of the update procedure from $h_3$ (orange line on the left) to $h_4$ (orange line on the right). We try to add term $a_3x+b_3$ (dashed line on the left) into the max operator in $h_3$. $a_4x+b_4$ intersects $h_3$ at $a_2x+b_2$. Thus $x^*$ (yellow dot) is calculated and the breakpoints larger than $x^*$ are deleted. Equivalently, $a_3x+b_3$ is removed (dashed line on the right).}
    \label{fig:2-m-sep}
\end{figure}

To find such $x^*$, we use a binary search on index $t$ to locate the proper line $a_tx+b_t$ forming the intersection point $x^*$. Such a search costs only logarithm time of the number of lines.

Now we analyze the time complexity of this algorithm. There are in total $k$ rounds of binary searches, with the $i$th round using time $O(\log i)$. In total, the time complexity is $O\left(\sum_{i=1}^k \log i\right)=O(k\log k)$. We collect the above arguments into the following proposition.

\begin{proposition}\label{prop:2dbpoint}
There exists an algorithm that outputs all the breakpoints of $h(x)$ and their corresponding function values in time $O(k\log k)$. 
\end{proposition}

\subsection{The \texorpdfstring{$(3,m)$}{(3,m)}-separation algorithm}
\label{app:prec-3-m-separation}
Note that since $t=3$, $\beta$ can be represented by two free variables, i.e., $\beta=(x,y,1-x-y)$. Then the polytope $P_i$ in \Cref{def:tm-sep-algo} is actually a polygon on the plane. The $(3,m)$-separation algorithm needs to find a clockwise enumeration of vertices of $P_i$. This problem, again, can be stated by the half-plane intersection. For completeness, we present the algorithm here.

In fact, a proper application of the $(2,m)$-separation algorithm will give us the desired algorithm. A key observation is that the boundary of a polygon can be expressed as a union of four parts: left boundary, right boundary, upper semi-boundary and lower semi-boundary. If we write all constraints of $P_i$ in the form $l_j:=\tilde{a}_jx+\tilde{b}_jy+\tilde{c}_j\geq 0, j\in[k]$, then each constraint belongs to exactly one part of the boundary:
\begin{enumerate}
    \item When $\tilde{b}_j=0$, $l_j=0$ is a candidate of the left (right) boundary if $\tilde{a}_j>0$ ($<0$).
    \item When $\tilde{b}_j\neq 0$, $l_j=0$ is a candidate of the upper (lower) semi-boundary if $\tilde{b}_j<0$ ($>0$).
\end{enumerate}
In the second case, we write the boundary into the form $y=\tilde{a}x+\tilde{b}$. Then we can apply \Cref{prop:2dbpoint} on the upper (lower) semi-boundary to obtain ordered vertices in time $O(k\log k)$. Next, we combine the two semi-boundaries to obtain the leftmost and rightmost vertices.

Now we determine the vertical boundaries. The left (right) boundary, if exists, has the maximum (minimum) $-\tilde{c}_j/\tilde{a}_j$, which can be found in $O(k)$ time. If the left (right) boundary does not rule out the leftmost (rightmost) vertex, then there is no left (right) boundary. Otherwise, by a binary search on vertices of the two semi-boundaries, we can find two segments adjacent to the left (right) boundary in $O(\log k)$ time. 
An illustration of this procedure is presented in \Cref{fig:cal-vertices-polygon}.
\begin{figure}[ht]
    \centering
    \includegraphics[width=0.8\textwidth]{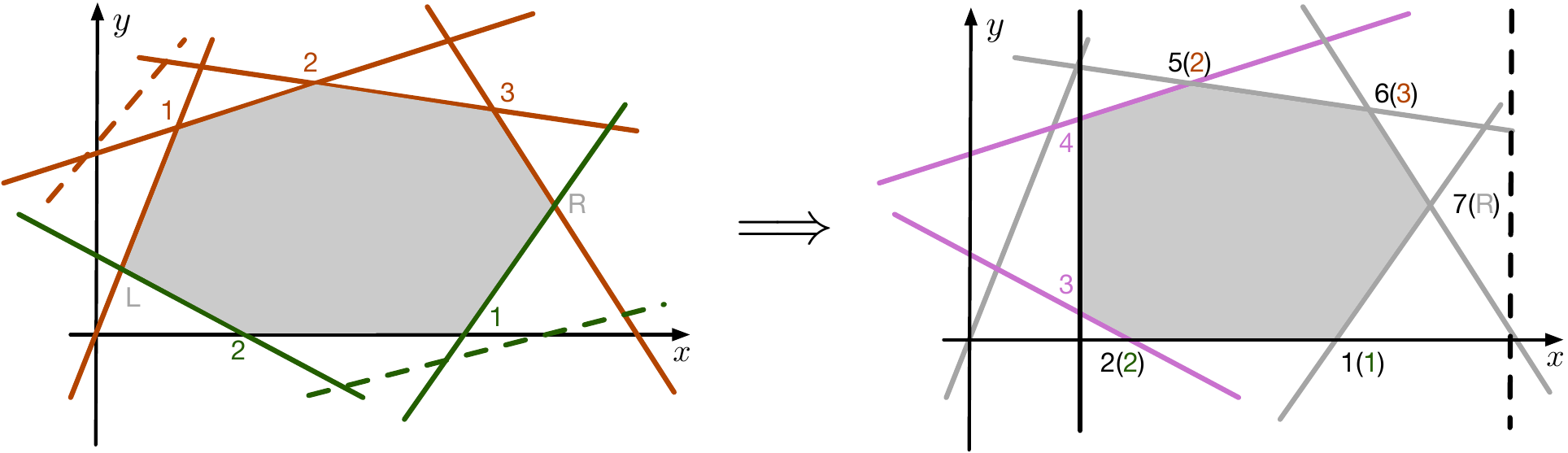}
    \caption{Illustration of the procedure computing ordered vertices (vertical line cases). On the left, the lower semi-boundary and the upper one is colored with green and orange, respectively. The vertices are labeled clockwise on each semi-boundary. On the right, we try to add vertical lines. The black dashed line does not change the structure at all, so it is omitted. The black solid line will change the structure, and the new labels of vertices are computed (and in the bracket is the original labels).
    }
    \label{fig:cal-vertices-polygon}
\end{figure}

We collect the above arguments into the following proposition.
\begin{proposition}
\label{prop:3dbpoint}
There exists an algorithm that outputs all vertices in a clockwise order of the polygon $P_i$ in time $O(k\log k)$.
\end{proposition}

\subsection{Proof of \Cref{details:lemma:argmax}}\label{app:argmax}
In case 1, take any $x\in S$. Then $g_2(x)\leq M_2\leq m_1\leq g_1(x)$. Therefore, $\min_S g(x)=\min_S g_1(x)=m_1$. The minimum is obtained exactly in $\{x\in S: g_1(x)=m_1\}$, namely $g_1^{-1}(m_1)$.

Case 2 is symmetric to case 1, so we omit it.

For case 3, we consider which set function $g$ reaches the minimum. Suppose that $x$ is a minimum, and $g_1(x)\neq g_2(x)$, then suppose without loss of generality that $g_1(x)>g_2(x)$. By the continuity of $g_1,g_2$, $g_1>g_2$ in some neighborhood of $x$. Therefore, $g=g_1$ in this neighborhood, so $x$ must be a local minimum for $g_1$. The problem becomes solving the first-order condition for optimization with linear constraints. Thus $x$ must also be a KKT point of $g_1$ and satisfy the KKT condition given in the statement.

\subsection{Proof of \Cref{prop:stand}}\label{details:stand}

We prove the statement in the order of \ref{prop:stand:state1}, \ref{prop:stand:state2}, \ref{prop:stand:S}, \ref{prop:stand:facetS}, \ref{prop:stand:partialS}, and \ref{prop:stand:faceS}.

\begin{proof}[Proof of Statement \ref{prop:stand:state1} and \ref{prop:stand:state2}]
For statement \ref{prop:stand:state1}, we note that notation $\dim (S)=m$ means that the smallest affine space containing it has dimension $m$. Therefore, the space can be expressed by the solution of $(n-m)$ linear equations, say $u_i^\T  x=v_i, i\in [n-m]$. Moreover, the remaining constraints cannot contain any equation like $u^\T x=v$; otherwise, $\dim(S)\leq m-1$, which violates the definition of the affine hull.
    
For statement \ref{prop:stand:state2}, by the definition of parallel, $d\parallel S$ if and only if there exists a line segment defined by $x_0,x_1\in S$ such that $x_1=x_0+\delta d$, where $\delta$ is a nonzero constant. Since $x_0,x_1\in S$, we have for every $i\in [n-m]$, $u_i^\T  x_0= v_i$ and  $u_i^\T (x_0+\delta d)=v_i$, so $u_i^\T  d=0$.
\end{proof}
To prove statements \ref{prop:stand:S}, \ref{prop:stand:facetS} and \ref{prop:stand:partialS}, we need the representation theorem for polytopes. We say a halfspace (inequality) $a^\T x\leq b$ is \emph{facet-defining} (for polytope $P$) if $P\cap\{x\in\R^n:a^\T x\geq b\}$ defines a facet of $P$.
\begin{theorem}[Representation theorem for polytopes, Theorem 2.15 in \cite{Z12_GTMpolytope}]\label{thm:represent-polytope}
A subset $P\subseteq\R^n$ is a polytope if and only if it can be described as a bounded intersection of facet-defining halfspaces, one for each facet, and of the affine hull of $P$. Moreover, the facet-defining inequalities are uniquely determined (if we write them as $a_i^\T x\leq 1)$, and none of them can be deleted.
\end{theorem}
\begin{proof}[Proof of Statement \ref{prop:stand:S} and \ref{prop:stand:facetS}]
By \Cref{thm:represent-polytope}, we can write $S$ in the form
\[\aff(S)\cap\bigcap_{i\in \tilde{W}}\left\{x\in\R^n:\tilde{a}_i^\T x\leq 1\right\}.\]
Here, inequalities $\tilde{a}_i^\T x\leq 1$ are facet-defining. Since they are unique and cannot be deleted, each inequality $\tilde{a}_i^\T x\leq 1$ corresponds to a constraint $a_j^\T x\leq b_j$ with $a_j=b_j\tilde{a}_i$. For each $i\in\tilde{W}$, pick such a $j$. Then we have already selected an index subset $W$ of $[k]$. By statement \ref{prop:stand:state1}, $\aff(S)$ can be written in the form of  $\left\{x\in\R^n:u_i^\T  x=v_i,\forall i\in [n-m]\right\}$. Now we have proved statement \ref{prop:stand:S}.

For statement 5, since constraint $a_j^\T x\leq b_j$, $j\in W$ is fact-defining, again by \Cref{thm:represent-polytope}, $S_j'=\aff(S)\cap\left\{x\in\R^n:a_j^\T x\leq b_j\right\}$ exactly represents a facet and \emph{vice versa}.
\end{proof}

Now we prove statement \ref{prop:stand:partialS}.

\begin{proof}[Proof of statement \ref{prop:stand:partialS}]
We first prove that
\[S_1:=\left\{x\in S:\forall k\in W, a_k^\T  x<b_k\right\} \subseteq S^\circ.\] 
Pick $x\in S_1$. By the continuity of $a_k^\T x$, there exists a small neighborhood $U$ of $x$ such that $U\cap S\subseteq S_1$. Therefore, $x$ is an interior point of $S$. Since $x$ is arbitrary, $S_1\subseteq S^\circ$ holds.
    
Second, we show that 
\[S_2:=\left\{x\in S:\exists k\in W, a_k^\T  x=b_k\right\}\subseteq \partial S .\] 
Note that by the construction of statement \ref{prop:stand:S}, all constraints indexed in $W$ in $S$ can not be deleted. It then guarantees that for every $k\in W$, there exists $x_0\in \R^n\setminus S$ such that $a_k^\T  x_0 > b_k$ and $a_j^\T  x_0\leq b_j$ holds for each $j\in W\setminus\{k\}$. Suppose $x_1\in S_2$. Then $a_k^\T  x_1=b_k$ for some $k\in W$. Consider the direction $d=x_0-x_1$. Notice that the set $\left\{x\in\R^n: a_k^\T  x \geq b_k\text{ and }\forall j\in W\setminus\{k\}, a_j^\T  x\leq b_j\right\}$ is convex. Since $x_1,x_0$ are both in it, the line segment defined by $x_1,x_0$ also lies in it. Therefore, for arbitrarily small $\epsilon>0$,
\begin{align*}
    a_k^\T \left(x_1+\epsilon(x_0-x_1)\right)=(1-\epsilon)\underbrace{a_k^\T x_1}_{=b_k}+\epsilon \underbrace{a_k^\T x_0}_{>b_k}>b_k.
\end{align*}
Hence $x_1+\epsilon(x_0-x_1)\notin S$ and thus $x_1\in\partial S$. Since $x_1$ is arbitrary, $S_2\subset \partial S$ holds.

Now we combine these two results. Clearly $S_1\cap S_2=\varnothing$ and $S_1\cup S_2=S=\partial S\cup S^{\circ}$. So we must have $S_1=S^\circ$ and $S_2=\partial S$.
\end{proof}
Finally, statement \ref{prop:stand:faceS} is the direct corollary of a result on face lattice. 

\begin{definition}
A \emph{graded lattice} is a finite partially ordered set $(S,\leq)$ if it shares \emph{all} the following properties.
\begin{itemize}
    \item It has a unique minimal element $\hat{0}$ and a unique maximal element $\hat{1}$.
    
    \item Every maximal chain has the same length. 
    \item Every two elements $x, y\in S$ have a unique minimal upper bound in $S$, called the \emph{join} $x \vee y$, and every two elements $x, y\in S$ have a unique maximal lower bound in $S$, called the \emph{meet} $x\wedge y$.
\end{itemize}

For a graded lattice, the minimal elements of $S\setminus\hat{0}$ are called \emph{atoms}, and the maximal elements of $S\setminus\hat{1}$ are called \emph{coatoms}.

A lattice is \emph{atomic} if every element is a join $x = a_1 \vee\dots\vee a_k$ of $k\geq 0$ of atoms. Similarly, a lattice is \emph{coatomic} if every element is a meet of coatoms.
\end{definition}

\begin{theorem}[Proposition 2.3 and Theorem 2.7 in \cite{Z12_GTMpolytope}]\label{thm:face-lattice}
Let $P$ be a convex polytope. Consider the set of all faces $L(P)$, partially ordered by inclusion.
\begin{enumerate}
    \item Set $L(P)$ is a graded lattice of length $\dim(P) + 1$. The meet operation is exactly the intersection of sets.
    \item The face lattice $L(P)$ is both atomic and coatomic.
    \item The faces of $F$ are exactly the faces of $P$ that are contained in $F$.
\end{enumerate}
\end{theorem}
\begin{proof}[Proof of statement \ref{prop:stand:faceS}]
By \Cref{thm:face-lattice}, $L(S)$ is a graded lattice. Suppose $F$ is a face in $L(S)$. Then since $L(S)$ is coatomic, $F$ is the meet (i.e., intersection) of coatoms, i.e., facets.
\end{proof}

\subsection{Proof of \Cref{details:cor:mincol}}\label{app:mincol}
We prove the corollary by discussing all possible cases that achieve the minimum. Since $g(x)=\max\{g_1(x),g_2(x)\}$, we partition the domain into three parts according to 
whether $g_1(x)$ is greater than, smaller than or equal to $g_2(x)$. 

\begin{proof}[Proof of Statement \ref{cor:mincol:state1}]
By symmetry, we only need to consider the case of 
\[S_1:=\left\{x\in S : g_1(x)>g_2(x),\exists \lambda\geq 0, \nabla g_1(x)+ \lambda^\T  U=0\text{ and }\forall i\in [m],  \lambda_i (U_i x-V_i)=0  \right\}.\]
It suffices to show that for every $x\in S_1$, either $\nabla g_1(x)=0$ or $x\in \partial S$.

For a given $x\in S_1$, if  $\nabla g_1(x)=0$, then $\lambda=0$ is a solution for the KKT conditions given in Theorem 12.1 in \cite{N99_NumOpt}. Otherwise, since $\nabla g_1(x)=-\lambda^\T  U\neq 0$, there must exist $i$ such that $\lambda_i\neq 0$. Therefore, we must have $U_i x=V_i$. By definition, every $x\in S$ satisfies $U_i x\leq V_i$. By our assumption on $U$, $U_i\neq 0$. So there is a vector $d\in\R^n$ such that $U_i d>0$. For any $\epsilon>0$, $U_i(x+\epsilon d)>V_i$. Hence $x+\epsilon d\notin S$, i.e., $x\in \partial S$. 
\end{proof}

To prove the rest statements, we need the following claim.

\begin{claim}\label{claim:face-parallel}
For any face $T$ of $S$, if $T$ is parallel to $e_i$, we have: for any $x\in T\cap S_1$, if the minimum of $f$ can be obtained at $x$, then either  $x$ is contained in a facet not parallel to $e_i$ or $\partial g_1(x)/ \partial x_i=0$.
\end{claim}

\begin{proof}
Since $x\in S_1$, by the KKT condition given in $S_1$, there exists $\lambda$ such that $\nabla g_1(x)=-\lambda^\T  U$. Thus $\partial g_1(x)/\partial x_i=\nabla g_1(x)^\T  e_i=-\lambda ^\T  U e_i$.  Note that by the definition of parallel, there exists a line $l$ such that $l\subseteq \aff(T)$ and $l\parallel e_i$. Therefore, for any $x\in T$, there exists a line $l_x\subseteq\aff(T)$ such that $x\in l_x$ and $l_x\parallel e_i$. Define $\tilde{l}_x:=l_x\cap T$. 

If $x$ is not an endpoint of the line segment $\tilde{l}_x$, then $\pm e_i$ are both feasible directions for $x$, namely for any sufficiently small $\epsilon>0$, $x\pm \epsilon e_i\in T$. We show that in this case, $x$ must satisfy $\partial g_1(x)/\partial x_i=0$. The KKT condition implies that for any $j\in [n]$, we either have $\lambda_j=0$ or $U_j x=V_j$. If $U_j x=V_j$, since $x\pm \epsilon e_i\in S$, $U_j (x\pm \epsilon e_i )\leq V_j= U_j x$, which means $U_j e_i=0$. Therefore, either $\lambda_j =0$ or $ U_j e_i=0$, we have $\lambda^\T  U e_i=\sum_j \lambda_j U_j e_i=0$. Thus $\partial g_1(x)/\partial x_i= -\lambda^\T  U e_i=0$ as desired.

To finish our proof, it suffices to show that if $x\in S_1\cap T$ is an endpoint of $\tilde{l}_x$, then either $x$ is contained in a facet $N$ of $S$ not parallel to $e_i$ or $\partial g_1(x)/\partial x_i=0$. We prove it by induction on $\dim(T)$. 

Suppose $\dim(T)=1$. By the definition of polytopes, $T$ must be a one-dimensional bounded and closed convex set, i.e., a line segment. In this case, $e_i\parallel T$, so for any $x_0\in T$, $\tilde{l}_{x_0}$ is exactly $T$. Thus, $x_0$ is the endpoint of $\tilde{l}_{x_0}$ if and only if $\{x_0\}$ is a face of $T$. By \Cref{thm:face-lattice}, $\{x_0\}$ is a face of $S$. Then by statement \ref{prop:stand:faceS} in \Cref{prop:stand}, $\{x_0\}$ is the intersection of several facets of $S$.  Since $S=\left\{x\in\R^n:U_i^\T  x\leq V_i, i\in [m]\right\}$, by statement \ref{prop:stand:facetS} in \Cref{prop:stand}, the facets of $S$ can be expressed as $S_i'=\left\{x\in S: U_i^\T  x=V_i\right\}$, $i\in W$ for some index subset $W$.  Thus there exists a nonempty subset of $W$, denoted by $I$, such that $\{x_0\}=\bigcap_{i\in I} S_i'$ and $x_0\notin S_j'$ for every $j\in W\setminus I$. Note that by assumption $\dim(S)=n$, we have $\aff(S)=\R^n$. By statement \ref{prop:stand:S} in \Cref{prop:stand}, $S=\left\{x\in \R^n: U_i^\T  x\leq V_i, \forall i\in W\right\}$. Then we have
\begin{align*}
    \{x_0\}=&\bigcap_{i\in I} S_i'\\
    =&\bigcap_{i\in I}\left(S\cap \left\{x\in \R^n: U_i^\T  x=V_i\right\}\right)\\
    =&S\cap \bigcap_{i \in I}\left(\left\{x\in \R^n: U_i^\T  x=V_i\right\}\right)\\
    =&\left\{x\in\R^n: U_i^\T  x=V_i, i\in I, U_j^\T  x<V_j, j\in W\setminus I\right. \}.
\end{align*}
Now we show that there exists a facet $S_i'$, $i\in I$ not parallel to $e_i$. Suppose on the contrary that for any $i\in I$, $S_i'$ is parallel to $e_i$, then we have $U_i^\T  e_i=0$ by the definition of parallel. Thus for any $k\in \R$, we have $U_i^\T (x_0+k e_i)=V_i$ for every $i\in I$. Note that by continuity there exists a sufficiently small $\epsilon>0$ such that for every $j\in W\setminus I$, $U_j^\T  (x+\epsilon e_i)< V_j$. Thus $x_0+\epsilon e_i$ is also contained in the set $\left\{x\in\R^n: U_i^\T  x=V_i, i\in I, U_j^\T  x< V_j, j\in W\setminus I \right\}=\{x_0\}$, a contradiction. So we finish the proof of the case $\dim(T)=1$.

Now we suppose that the result holds on every $h$-dimensional face with $h=m-1\leq n-1$, and let $\dim(T)=m$. Note that for any $x\in T^\circ$, there exists $\epsilon>0$ such that every $d\parallel  \aff(T)$ satisfies $x+\epsilon d\in T$. So $x$ must be an interior point of $\tilde{l}_x$ and thus not be an endpoint of $\tilde{l}_x$. We have assumed that $x$ is the endpoint of $\tilde{l}_x$, so this is not the case. We must have $x\in \partial T$. By statement \ref{prop:stand:partialS} and \ref{prop:stand:facetS} in \Cref{prop:stand}, $x$ must be contained in a face $T'\subseteq \partial T$ of $T$ with $\dim(T')=m-1$. By \Cref{thm:face-lattice}, $T'$ is also a face of $S$. If $T'\parallel e_i$, then line $l_x\subseteq\aff(T')$ with $x\in l_x$. Let $\tilde{l}'_x=l_x\cap T'$. By the same argument, $\tilde{l}'_x$ is a line segment. If $x$ is not an endpoint of this segment, then we have proved that $\partial g_1(x)/\partial x_i=0$ as desired. If $x$ is an endpoint, then by the induction hypothesis either $x$ is contained in a facet $N$ of $S$ not parallel to $e_i$ or $\partial g_1(x)/\partial x_i=0$. Thus the induction holds. If face $T'$ is not parallel to $e_i$, then we show that $T'$ is contained in some facet $N$ of $S$ not parallel to $e_i$. Note that since face $T'$ is not parallel to $e_i$, $l_x\cap T'=\{x\}$. If all facets $S'$ of $S$ containing $T'$ are parallel to $e_i$, by the same argument on the case of $\dim(T)=1$, for sufficiently small $\epsilon>0$, $x+\epsilon e_i\in T'$. However, clearly $x+\epsilon e_i\in l_x$ and thus $x+\epsilon e_i\in l_x\cap T'$, which leads to a contradiction.
\end{proof}

Now we can continue the main proof.
\begin{proof}[Proof of Statement \ref{cor:mincol:state2}]
For statement \ref{cor:mincol:state2}, it suffices to show that for every $x\in \partial S\cap S_1$, either $x\in \partial S_N$ or both $x\in \partial S_P$ and $\partial g_1(x)/\partial x_i=0$ hold. Equivalently, we show that for every $x \in P\cap S_1$, where $P$ is a facet parallel to $e_i$, either there exists a facet $N$ which is not parallel to $e_i$ such that $x \in N$, or $\partial g_1(x)/\partial x_i=0$. This immediately follows by taking $m=n-1$ in \Cref{claim:face-parallel}.
\end{proof}

\begin{proof}[Proof of Statement \ref{cor:mincol:state3}]
We first show that any face $T$ of $S$ must have the form 
\[
T_{I_1,I_2}:=\prod_{i\in [n]} S_i.
\]
Here, $S_i=\{m_i\}$ for every $i\in I_1$, $S_i=\{M_i\}$ for every $i\in I_2$ and $S_i=[m_i,M_i]$ for every $i\in[n]\setminus(I_1\cup I_2)$, where $I_1,I_2\subseteq[n]$ and $I_1\cap I_2=\varnothing$.

By applying statement \ref{prop:stand:S} in \Cref{prop:stand}, $S$ can be written into $\{x\in \R^n: x_i\leq M_i, -x_i\leq -m_i, \forall i\in [n]\}$. Therefore, by statement \ref{prop:stand:facetS} in \Cref{prop:stand}, the facets of $S$ is given by $S\cap \{x\in\R^n: x_i=m_i\}$ or $S\cap \{x\in\R^n: x_i=M_i\}$. By statement \ref{prop:stand:faceS} in \Cref{prop:stand}, any face $T$ of $S$ can be expressed as the intersection of several facets. Thus there exist index subsets $I_1, I_2$ such that 
\[T=S\cap \bigcap_{i\in I_1} \{x\in\R^n: x_i=m_i\}\cap \bigcap_{i\in I_2} \{x\in\R^n: x_i=M_i\}.\]
If $I_1\cap I_2=\varnothing$, then exactly $T=T_{I_1,I_2}$; otherwise, if $i\in I_1\cap I_2$, by $m_i\neq M_i$, $T=\varnothing$. Now, it suffices to show that any $T_{I_1,I_2}$ is a face of $S$. For any $T_{I_1, I_2}$, consider $a=-\sum_{i\in I_1} e_i +\sum_{i\in I_2} e_i, b=-\sum_{i\in I_1} m_i +\sum_{i\in I_2} M_i$. On the one hand, for any $x\in S$, we have $m_i\leq x^\T  e_i=x_i\leq M_i$, so $a^\T  x=\sum_{i\in I_1} (-x_i)+\sum_{i\in I_2} x_i\leq -\sum_{i\in I_1}m_i+\sum_{i\in I_2} M_i=b$. On the other hand, we can see that the equality holds if and only if $x_i=m_i$ for every  $i\in I_1$ and $x_i=M_i$ for every $i\in I_2$. This set is exactly given by $T_{I_1,I_2}$, so $T_{I_1,I_2}$ is the face of $S$ determined by $a,b$, and we finish our proof.

With the clear description of all the faces of $S$ by $T_{I_1,I_2}$, we can easily see that a face $T_{I_1,I_2}$ is a single point if and only if $I_1\cup I_2=[n]$. Also, for any face $T$ that is not a single point, there exists $i\in[n]$ such that $i\notin I_1\cup I_2$. So, for any $y=(y_1,...,y_n)$ in $T$, $\prod_{j\neq i} \{y_j\}\times [m_i,M_i]\subseteq T$, which defines a line from $(y_1,\dots, y_{i-1}, m_i, y_{i+1},\dots, y_n)$ to $(y_1,\dots, y_{i-1}, M_i, y_{i+1},\dots, y_n)$ parallel to $e_i$. Therefore, by the definition of parallel, $e_i\parallel T$. So every face $T$ of $S$ is parallel to some $e_i$.

Suppose $x\in\partial S$. We show that either $x$ is a single point face, or $x$ belongs to the interior of some face $T$ parallel to some $e_i$. Note that $x$ must belong to some face $T$ of $S$. We prove it by induction on the dimension of $T$. When $\dim(T)=0$, $T=\{x\}$ is a single point face. For $\dim(T)=n$, we only need to consider the case that $x\notin T^\circ$. Immediately, $x\in \partial T$, so by \Cref{prop:stand} and \Cref{thm:face-lattice}, $x$ belongs to some face of $S$ with lower dimension. The result then follows by the induction hypothesis. 

Thus, either $x$ is a single point face given by $\left\{x\in\R^n: \forall i, x_i \in\left\{m_i, M_i\right\}\right\}$, or $x$ belongs to the interior of a face $T$ parallel to some $e_i$. We now apply the result in \Cref{claim:face-parallel}. Suppose $x$ is not a single point face that attains the minimum of $f$. If $g_1(x)>g_2(x)$, then $\partial g_1(x)/\partial x_i=0$; if $g_1(x)<g_2(x)$, then $\partial g_2(x)/\partial x_i=0$. So, $x$ must be contained in the set $S^+$ given in this statement.
\end{proof}

\subsection{The \texorpdfstring{$(2,2)$}{(2,2)}-mixing algorithm}
\label{app:prec-2-2-mixing}
We first give some notations Let $F_R(\alpha,\beta)=f_R(\alpha x_1+(1-\alpha) x_2, \beta y_1+(1-\beta)y_2)$. Define $F_C(\alpha,\beta)$ similarly. Then let $F(\alpha,\beta)=\max\{F_R(\alpha,\beta),F_C(\alpha,\beta)\}$. The goal of the $(2,2)$-mixing algorithm is to calculate the minimum of $F$ on square $\mathcal{A}=[0,1]\times[0,1]$. Now we state the algorithm.

Applying the $(2,m)$-separation algorithm in \Cref{app:prec-2-m-separation}, we can construct a mesh grid of $(\alpha,\beta)$ so that on each grid, both $F_R$ and $F_C$ are linear in $\alpha$ and $\beta$ respectively. Then both $F_R$ and $F_C$ have the form $x_1+x_2 \alpha+x_3 \beta+x_4 \alpha \beta$, where $x_i$'s are constants determined by $F_R$ or $F_C$ values on four vertices of the grid. Our next step is then to give a method computing the minimum point of $F(\alpha,\beta)$ on each grid.

On each grid, by statement 3 in \Cref{details:cor:mincol}, it suffices to minimize $F=\max\{F_R,F_C\}$ over:
\begin{enumerate}
    \item points with $\partial F_k(\alpha,\beta)/\partial\alpha=0$ or $\partial F_k(\alpha,\beta)/\partial\beta=0$, $k\in\{R,C\}$,
    \item the four vertices of the grid, and
    \item points with $F_R=F_C$.
\end{enumerate}

\begin{itemize}\label{state:append}
    \item (Case 1) Equations $\partial F_k(\alpha,\beta)/\partial\alpha=0$ and $\partial F_k(\alpha,\beta)/\partial\beta=0$ have the form that $\alpha$ or $\beta$ takes a fixed value.\footnote{When the coefficient of $\alpha$ (or $\beta$) is zero, all or none of $\alpha$ (or $\beta$) solve the equation.} Hence the problem becomes computing the minimum of two univariate linear functions, which can be solved by statement 3 in \Cref{prop:minmax}.
    \item (Case 2) We just need to enumerate the value of $F$ on the four vertices.
    \item (Case 3) By solving the equation $F_R(\alpha,\beta)=F_C(\alpha,\beta)$, we obtain an expression of $\beta$ given by a linear fraction of $\alpha$. If the denominator linear function of $\alpha$ is zero, then we can solve it just like in case 1. Otherwise, by substituting the expression of $\beta$ into the expression of $F=F_R$, we convert this problem into finding the minimum of a function $g(\alpha)$ with the form $(a_2\alpha^2+a_1\alpha +a_0)/(b_1\alpha+b_0)$. This can be done by calculating its values at two boundary points and points with zero derivatives. Note that $g'(\alpha)=0$ is equivalent to a quadratic equation in $\alpha$, which has at most two solutions. So in this case we can test at most four points to find the minimum. 
\end{itemize}

We collect the above arguments into the following proposition:

\begin{proposition}
\label{prop:2dmin}
There exists an algorithm finding the minimum point of $F(\alpha,\beta)$ on any grid in $O(1)$ time.
\end{proposition}

With these results above, we can efficiently calculate the minimum point of $f$ on each grid where both $F_R$ and $F_C$ are linear in $\alpha$ and $\beta$ respectively. Note that the numbers of breakpoints of $\alpha$ and $\beta$ are at most $m$ and $n$ respectively, so there are at most $mn$ grids. On each grid the minimization procedure takes $O(1)$ time, implying a total $O(mn)$ time on $\mathcal{A}$. Thus time complexity of calculating the minimum of $f$ on $\mathcal{A}$ is $O(\max\{m,n\}\log\max\{m,n\})+O(m n)=O(m n)$. We summarize it as the following theorem. 

\begin{theorem}\label{thm:mn-min}
Give any strategies $x_1,x_2\in\Delta_m$ and $y_1,y_2\in\Delta_n$, let
\[F(\alpha,\beta)=f(\alpha x_1+(1-\alpha) x_2,\beta y_1+(1-\beta)y_2),\quad\alpha,\beta\in[0,1].\]
Then there exists an algorithm finding the minimum point of $F(\alpha,\beta)$ in time $O(mn)$.
\end{theorem}

\subsection{The \texorpdfstring{$(2,3)$}{(2,3)}-mixing algorithm}
\label{app:prec-2-3-mixing}
We begin by some notations. Let
\begin{align*}
F_R(\alpha,\beta,\gamma)=&\max\{R(\gamma y_1+(1-\gamma)y_2)\}-(\alpha x_1+\beta x_2+(1-\alpha-\beta)x_3)^\T R(\gamma y_1+(1-\gamma)y_2),\\
F_C(\alpha,\beta,\gamma)=&\max\left\{C^\T (\alpha x_1+\beta x_2+(1-\alpha-\beta)x_3)\right\}-\\
&\qquad (\alpha x_1+\beta x_2+(1-\alpha-\beta)x_3)^\T C(\gamma y_1+(1-\gamma)y_2).
\end{align*}
Define $F(\alpha,\beta,\gamma)=\max\{F_R(\alpha,\beta,\gamma),F_C(\alpha,\beta,\gamma)\}$. Then the algorithm in this part minimizes $F$ on the prism $\mathcal{A}=\{(\alpha,\beta,\gamma)\in[0,1]^3:\alpha+\beta\leq 1\}$.

Using $(2,m)$-separation algorithm in \Cref{app:prec-2-m-separation} and $(2,m)$-separation algorithm in \Cref{app:prec-2-3-mixing}, we can obtain the linear region\footnote{To shorten statements, we say region $X$ is a \emph{linear region} of function $F$ if $F$ is linear in every variable on $X$.} of function $F$. Our next step is then to minimize $F$ on each linear region, in which both $F_R$ and $F_C$ have form $a_0 \alpha \gamma+a_1 \beta \gamma+a_2 \gamma+a_3 \alpha+a_4 \beta+a_5$. Note that every linear region $S$ is given by the Cartesian product of a polygon $P$ and an interval $I$. Thus by statement 2 in \Cref{details:cor:mincol}, the minimum of $F$ must be obtained when:
\begin{enumerate}
    \item $(\alpha,\beta)$ belongs to side surfaces of $S$ and 
    \begin{enumerate}
        \item either there exists $k\in\{R,C\}$ such that $\partial F_k/\partial \gamma=0$, or
        \item $(\alpha,\beta)$ is in the intersection of side surfaces and top/bottom surfaces.
    \end{enumerate}
    \item $(\alpha,\beta)$ belongs to top/bottom surfaces of $S$ and  
    \begin{enumerate}
        \item there exists $k\in\{R,C\}$ such that either $\partial F_k/\partial \alpha=0$ or $\partial F_k/\partial \beta=0$, or
        \item $(\alpha,\beta)$ is in the intersection of side surfaces and top/bottom surfaces.
    \end{enumerate}
    \item $F_R(\alpha,\beta)=F_C(\alpha,\beta)$.
    \item $\nabla F_R(\alpha,\beta)=0$ or $\nabla F_C(\alpha,\beta)=0$.
\end{enumerate}

\label{state:append2}
For case 1b and 2b, note that the boundary is formed by $O(m)$ line segments. Furthermore, $F_R$ and $F_C$ are linear on each segment. Thus it suffices to apply \Cref{prop:minmax} on each segment.

For case 2a and case 4, the equation of zero derivative gives a linear equation on $\gamma$. Then $\gamma$ takes a fixed value. Now we have to minimize $F$ over the polygon $P$ of $(\alpha,\beta)$. We can use statement 2 in \Cref{details:cor:mincol} again, and minimize $F$ at points on the boundary of $P$, with $F_R=F_C$, and with zero partial derivative in $\alpha$ or $\beta$. The case of the boundary is similar to case 1b and 2b, consuming time $O(m)$. The rest cases are similar to discussions in \Cref{prop:2dmin}: We can turn the problem into minimizing a univariate function $g$. The only difference here is the domain $J$ of $g$. By the same calculation in the proof of \Cref{prop:2dmin}, it can be shown that $J$ is a segment (or line) $\tilde{J}$ parallel or perpendicular to $\alpha=0$ when we ignore the restriction of $P$. Domain $J$ then can be determined by searching the intersection points of $\tilde{J}$ and the boundary of $P$ in $O(m)$ time. 

For case 1a, the equation of zero derivative gives a linear equation on $(\alpha,\beta)$. Then the equation produces a line $l$ on $(\alpha,\beta)$. Now, the feasible set of $(\alpha,\beta)$ is a segment $L$ determined by the intersection of $l$ and $P$. Similar to $J$, we can compute two endpoints of $L$ in $O(m)$ time. Note that by a suitable linear transformation from $(\alpha,\beta)$ to $(\alpha',\beta')$, the equation of $l$ becomes $\alpha'=0$. Then on $L\times I$, $F$ becomes a function of $(\beta',\gamma)$ being linear in $\beta'$ and $\gamma$, respectively. Now we can apply \Cref{prop:2dmin} to minimize $F$ on $L\times I$ in $O(1)$ time.

For case 3, by $F_R=F_C$, we obtain an expression of $\gamma$ given by the fraction of linear functions in $\alpha$ and $\beta$. A special case is that the denominator equals zero. We can deal with this case in the same way as case 1a. Otherwise, by substituting this expression into $F_R$, it suffices to minimize a function $h(\alpha,\beta)$ with the form $(c_0+c_1 \alpha+c_2 \beta+c_3 \alpha^2+c_4 \alpha \beta+c_5 \beta^2)/(d_0+d_1 \alpha+d_2 \beta)$ on a given linear region. 

When $d_1=d_2=0$, this is to solve quadratic programming on a polygon with $O(m)$ sides. The minimum is taken either on the sides or at interior points with zero derivatives. Since it has only two variables, we can cancel one of the variables via the linear equation of a side. Then the minimization on the side is equivalent to minimizing a univariate quadratic function on a segment. On the other hand, the zero-derivative condition is exactly two linear equations with two variables. In both situations, the calculation can be completed within $O(m)$ time. 

Otherwise, we substitute the denominator with $\theta$, and the expression is transformed into
\[
G_1(\alpha, \theta):=\frac{e_0+e_1 \alpha+e_2 \alpha^2}{\theta}+e_3+e_4 \alpha+e_5 \theta.
\]

First, we consider the minimum about $\theta$. By the property of hyperbolic function, the minimum can only be obtained at the boundary points or at $\theta=\pm \sqrt{\left(e_0+e_1 \alpha+e_2 \alpha^2\right) / e_5}$ (if exists). Since $\theta$ is linear in $(\alpha,\beta)$ and the domain of $(\alpha,\beta)$ is a polygon $P$, the domain of $\theta$ is a interval given by $[M_{\min}(\alpha),M_{\max}(\alpha)]$, where $M_{\min},M_{\max}$ are piecewise linear functions with $O(m)$ pieces. By considering vertices of $P$ in order, we can calculate linear pieces of $M_{\min}$ and $M_{\max}$ in $O(m)$ time, denoted by $I_i^{\min}$ and $I_j^{\max}$, respectively. So we only need to consider $O(m)$ cases that $\theta$ takes $M_{\min}(\alpha)$ on $\alpha\in I_i^{\min}$, $M_{\max}(\alpha)$ on $\alpha \in I_j^{\max}$, or $\pm \sqrt{\left(e_0+e_1 \alpha+e_2 \alpha^2\right) / e_5}$ when $\left(e_0+e_1 \alpha+e_2 \alpha^2\right) e_5\geq 0$. In each case, it suffices to find the minimum of either
\begin{gather*}
\frac{e_0+e_1 \alpha+e_2 \alpha^2}{t(\alpha)}+e_3+e_4 \alpha+e_5 t(\alpha), t \in\{M_{\min}, M_{\max}\}\text{ or }
e_3+e_4 \alpha \pm 2 \sqrt{e_5\left(e_0+e_1 \alpha+e_2 \alpha^2\right)},
\end{gather*}
where $\alpha$ belongs to a certain interval. Each case can be solved by calculating points on the boundary and points with zero derivatives in $O(m)$ time.

We conclude the discussion above with the following proposition.

\begin{proposition}\label{prop:3dmin}
There exists an algorithm finding the minimum point of $F(\alpha,\beta,\gamma)$ on linear region $S=P\times I$ in $O(m)$ time, where $P$ is a polygon of $(\alpha,\beta)$ formed by $O(m)$ linear constraints and $I$ is a closed interval of $\gamma$.

\end{proposition}

Now we come back to function $F(\alpha, \beta, \gamma)$. Using \Cref{prop:2dbpoint} and \Cref{prop:3dbpoint}, we can split the domain of $F$ into $O(mn)$ linear regions in time $O(n\log n+m^2\log m)$. Then on each region, we can use \Cref{prop:3dmin} to compute the minimum value of $F$ in time $O(m)$. The total time complexity is then $O\left(m^2(n+\log m)+n\log n\right)$. We summarize it into the following theorem.
\begin{theorem}
Give any $x_1, x_2, x_3 \in \Delta_m$ and $y_1, y_2 \in \Delta_n$, let
$$
F(\alpha, \beta, \gamma)=f\left(\alpha x_1+\beta x_2+(1-\alpha-\beta) x_3, \gamma y_1+(1-\gamma) y_2\right),
$$
where $\alpha, \beta, \gamma, \alpha+\beta \in[0,1]$. Then there exists an algorithm finding the minimum point of $F(\alpha, \beta, \gamma)$ in time $O\left(m^2(n+\log m)+n\log n\right)$.

\end{theorem}

\subsection{Proof of \Cref{prop:minmax-mixing}}
\label{app:minmax-mixing}

Here we prove a slightly generalized proposition:

\begin{proposition}
Given any two points $A=(x_A,y_A),B=(x_B,y_B) \in\Delta_m\times \Delta_n$, we can define a line segment $AB:=\{t A+(1-t) B:t\in[0,1]\}$. It has the following properties:

\begin{enumerate}

\item $AB\subseteq \Delta_m\times \Delta_n$. Therefore the value of $f$ is well-defined on $AB$.

\item Denote the value of $f_I$ on the vertices $A$ and $B$ by $a_I$ and $b_I$ respectively, where $I\in\{R,C\}$. Let $f^{AB}_{*}$ be the minimum of $f$ on $AB$. Suppose 
\[(x_A-x_B)^\T  R(y_A-y_B)\leq 0 \text{ and }(x_A-x_B)^\T C(y_A-y_B)\leq 0.\] 
Specifically, this condition holds when $x_A=x_B$ or $y_A=y_B$. We have the following properties.
\begin{enumerate}
\item If $a_C\leq a_R$ and $b_C\leq b_R$, then $f^{AB}_{*}\leq \min\{f_R(A),f_R(B)\}= \min\{a_R,b_R\}$.
\item If $a_R\leq a_C$ and $b_R\leq b_C$, then $f^{AB}_{*}\leq \min\{f_C(A),f_C(B)\}= \min\{a_C,b_C\}$.
\item If $a_R+b_C-a_C-b_R=0$, then 
\[f_*^{AB}\leq\min\{f(A),f(B))=\min\{\max\{a_C,b_C\},\max\{a_R,b_R\}\}.\]
\item Otherwise, $f^{AB}_{*}\leq f\left(\frac{b_C-a_C}{a_R+b_C-a_C-b_R} A+\frac{a_R-b_R}{a_R+b_C-a_C-b_R}B\right)\leq \frac{a_Rb_C-a_Cb_R}{a_R+b_C-a_C-b_R}$.

\end{enumerate}
Conversely, when $f_R$ and $f_C$ are both linear in parameter $t$, the equalities of $f_*^{AB}$ in all three cases hold.
\end{enumerate}
\end{proposition}

We prove these properties one by one.

\begin{enumerate}[fullwidth, listparindent=\parindent]
\item Since any standard simplex is convex and the Cartesian product of two convex sets is again convex, the property follows directly from the convexity of $\Delta_m\times \Delta_n$.

\item Define $F^{AB}_I(t)=f_I(tA+(1-t)B),t\in[0,1], I\in\{R,C\}$. Then 
\begin{align*}
    F^{AB}_R(t)&=\max\left\{R(ty_A+(1-t)y_B)\right\}-(tx_A+(1-t)x_B)^\T  R(ty_A+(1-t)y_B),\text{ and}\\ F^{AB}_C(t)&=\max\left\{C^\T (tx_A+(1-t)x_B)\right\}-(tx_A+(1-t)x_B)^\T  C(ty_A+(1-t)y_B).
\end{align*}

First, we show the convexity of $F_R^{AB}$ and $F_C^{AB}$ in $t$. Note that for any vector $u, v$ with identical length, $\max\{tu+v\}$ is convex in $t$, so the first parts of both functions are convex. Besides, the second parts of both functions are quadratic. So, it is clear that both functions are convex when $(x_A-x_B)^\T  R(y_A-y_B)\leq 0$ and $(x_A-x_B)^\T C(y_A-y_B)\leq 0$.

Recall that we define $a_I=f_I(x_A,y_A), b_I=f_I(x_B,y_B), I\in \{R,C\}$. By the convexity of $F_R^{AB}(t)$ and $F_C^{AB}(t)$, we can bound the value of them on $t\in [0,1]$ by linear functions $L_R^{AB}(t)=ta_R+(1-t)b_R$ and $L_C^{AB}=ta_C+(1-t)b_C$. Thus, we have 
\[
f_*^{AB}=\min_{t\in [0,1]} \left\{F_R^{AB}(t),  F_C^{AB}(t)\right\}\leq \min_{t\in [0,1]} \left\{L_R^{AB}(t),  L_C^{AB}(t)\right\}.
\]

Now, we consider the exact expression of $\min_{t\in [0,1]} \left\{L_R^{AB}(t),  L_C^{AB}(t)\right\}$.

\begin{enumerate}
    \item If $a_C\leq a_R$ and $b_C\leq b_R$, then $L_R^{AB}$ is totally above $L_C^{AB}$ on $t\in [0,1]$. Thus,
\begin{align*}
    \min_{t\in [0,1]} \left\{L_R^{AB}(t),  L_C^{AB}(t)\right\}=\min\left\{L_R^{AB}(0),L_R^{AB}(1)\right\}=\min\{a_R,b_R\}.
\end{align*}

The equality holds at $t=0$.

\item If $a_R\leq a_C$ and $b_R\leq b_C$, analogously, 
\[\min_{t\in [0,1]} \left\{L_R^{AB}(t),  L_C^{AB}(t)\right\}= \min\{a_C,b_C\}.\] 
The equality holds at $t=1$.

\item If $a_R+b_C-a_C-b_R=0$, then lines defined by $L_R^{AB}$ and $L_C^{AB}$ on $t\in [0,1]$ are parallel. Thus we have
\[\min_{t\in [0,1]} \left\{L_R^{AB}(t),  L_C^{AB}(t)\right\}=\min\{\max\{a_C,b_C\},\max\{a_R,b_R\}\}.\]
The equality holds at $t=0$ if $f^{AB}(0)\leq f^{AB}(1)$ and at $t=1$ if $f^{AB}(0)>f^{AB}(1)$.
\item Otherwise, two lines intersect. We solve the equation $L_R^{AB}(t)=L_C^{AB}(t)$ and get $t_0=\frac{b_C-b_R}{a_R+b_C-a_C-b_R}$. Thus,
\begin{align*}
   \min_{t\in [0,1]} \left\{L_R^{AB}(t),  L_C^{AB}(t)\right\}=\max\left\{L_R^{AB}(t_0),L_C^{AB}(t_0)\right\}=\frac{a_Rb_C-a_Cb_R}{a_R+b_C-a_C-b_R}.
\end{align*}
\end{enumerate}

When $f_R$ and $f_C$ are both linear, the equality holds at $t_0$ since the minimum of the maximum of two linear functions is obtained exactly at the intersection point.
\end{enumerate}

\subsection{Proof of \Cref{lemma:DFM-relation}}\label{app:DFM-relation}
Relation 1 follows from \Cref{lemma:frfc}. Relation 2 follows from the fact that $f_R,f_C\in[0,1]$. Relation 3 follows from  \Cref{lemma:increasing}. Relation 5 follows from the definition of $\hat{w}$. Thus we only prove relations 4 and 6. 

The three inequalities in relation 4 can be obtained by substituting $(w,y_s),(x_s,z),(\hat{w},y_s)$ into \Cref{lemma:strongineq}, respectively.

It is convenient to draw a prism to show the positions of little letter aliases, as in \Cref{fig:app-prism}.
\begin{figure}[ht]
    \centering
    \includegraphics[width=0.9\textwidth]{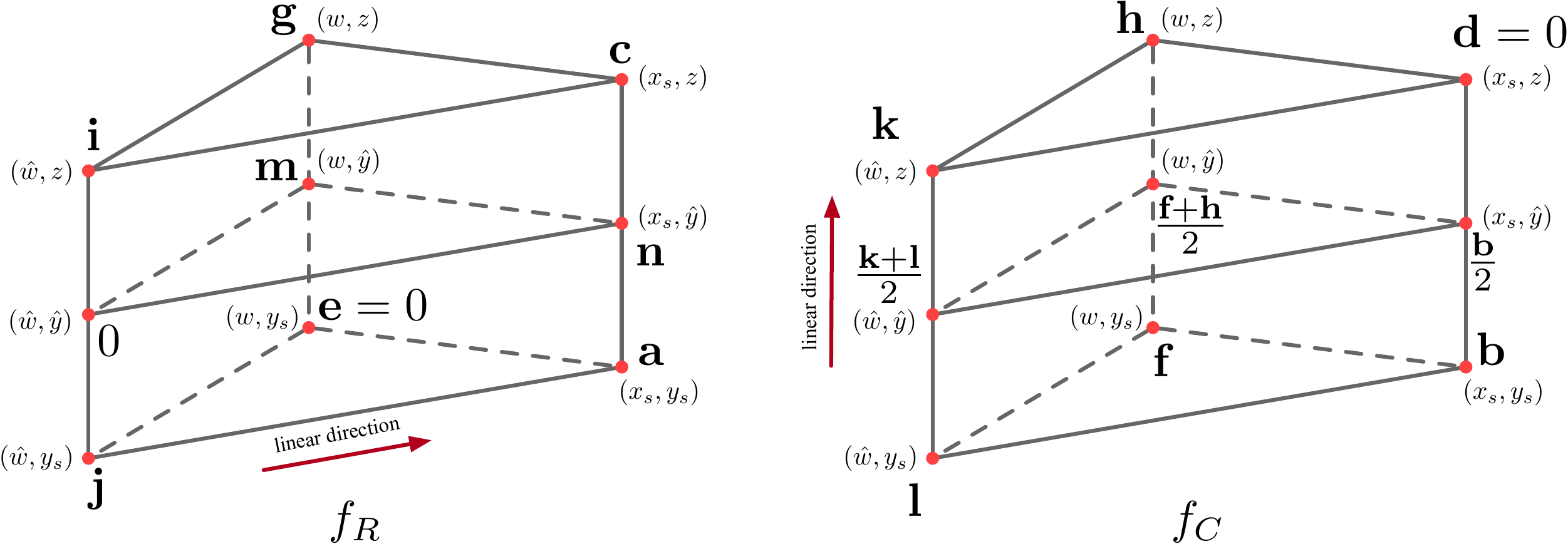}
    \caption{Positions of little letter aliases.}
    \label{fig:app-prism}
\end{figure}
Relation 6 relies on the linearity of $x^\T Ry$. Define $\zm:=f_R(w,\hat{y})$, $\zu:=\max\{Rz\}$, $\zv:=\max\{R\hat{y}\}$, and $\zw:=\max\{Ry_s\}$. By relation 5, $\hat{w}^\T Ry_s+\hat{w}^\T Rz=2\hat{w}^\T R\hat{y}$. Thus we have $2\zv=(\zu-\zi)+(\zw-\zj)$. By the definition of $\hat{y}$, $w^\T Ry_s+w^\T Rz=2w^\T R\hat{y}$ and $x^\T Ry_s+x^\T Rz=2x^\T R\hat{y}$. Hence $2(\zv-\zm)= (\zu-\zg)+(\zw-0)$. By subtracting the second equation from the first one, we obtain the equation $\zm=(\zg-\zi-\zj)/2$. 
Relation 6 then follows from that $\zm$ is nonnegative.

Relation 7 relies on the linearity of $x^\T C y$ and $\hat{y}=\frac{1}{2}(y_s+z)$. Note that 
\begin{align*}
    f_C(\hat{w},\hat{y})=&\max\{C^\T \hat{w}\}-\hat{w}^\T C \hat{y}\\
    =&\frac{1}{2}(\max\{C^\T \hat{w}\}-\hat{w}^\T C y_s)+\frac{1}{2}(\max\{C^\T \hat{w}\}-\hat{w}^\T C z)\\
    =&\frac{1}{2}(f_C(\hat{w},y_s)+f_C(\hat{w},z))=\frac{1}{2}(\zl+\zk).
\end{align*}

\subsection{Proof of 
\Cref{details:thm:1/3-approx}}\label{app:approx-DFM}

Now, we show the $1/3+\delta$ approximation of the modified DFM algorithm. Again, to simplify the notations, we consider stationary points and omit $\delta$ in the proof.

The key observation is that we can drop several cases in the discussion. We choose cases to be dropped according to the tight case computed by the numerical optimization solver.

We want to give an upper bound of $f^*_{\mathcal{A}}$, the global minimum of $f$ on the triangular prism $\mathcal{A}=\{(\alpha_1 x_s+\alpha_2 w +\alpha_3 \hat{w},\beta_1 y_s+ \beta_2 z+\beta_3 \hat{y}), \alpha\in \Delta_3, \beta\in \Delta_3\}$. To do so, we use the upper bound constructor.

\begin{lemma}
\label{lemma:1stpart}
Suppose that $\zg\geq\zh$. Define
\[D=\begin{cases}
    \frac{\zi(\zk+\zl)}{2\zi+\zl-\zk},&\text{if }\zi>\zk,\\[5pt]  \frac{\zg\zk-\zi\zh}{\zg+\zk-\zi-\zh},&\text{if }\zi\leq \zk.
\end{cases}\]
Then
\[f^*_{\mathcal{A}}\leq \min\left\{\rho(\zc-\zg),(1-\rho)(\zf-\zh),\rho \zj+(1-\rho)(\zl-\zk), \frac{\zf \zg}{\zf+\zg-\zh},D\right\}.\] 
\end{lemma}
\begin{proof}
Note that the constructor gives an upper bound of $f^*_{\mathcal{A}}$ by applying the $(1,2)$-upper bound constructor on each edge of the prism. In this way, we can only pick certain edges, and give a loose upper bound.

First, we pick the edge $\{(x_s, \beta y_s+ (1-\beta)\hat{y}): \beta\in [0,1]\}$. The values on the vertices are: $f_R(x_s,y_s)=\za, f_C(x_s,y_s)=\zb, f_R(x_s,\hat{y})=\zn,
f_C(x_s,\hat{y})=\frac{\zb}{2}$. From \Cref{lemma:DFM-relation},we have $\zn\geq \za\geq \zb\geq \frac{\zb}{2}$. Thus, by relation 3 of \Cref{prop:minmax-mixing}, the upper bound is given by $\za$. Besides, from relation 4 of \Cref{lemma:DFM-relation}, we have $\za\leq \rho(\zc-\zg), \zb\leq (1-\rho)(\zf-\zh)$, $\za\leq \rho \zj+(1-\rho)(\zl-\zk)$. Therefore, the upper bound given by this edge is $\{\rho(\zc-\zg), (1-\rho)(\zf-\zh), \rho \zj+(1-\rho)(\zl-\zk)\}$. 

Second, consider the edge $\{(w, \beta y_s+ (1-\beta)z): \beta\in [0,1]\}$. The values on the vertices are: $f_R(w,y_s)=\ze, f_R(w,z)=\zg, f_C(w,y_s)=\zf, f_C(w,z)=\zh$. From the condition $\zg\geq \zh$ and relation 1 of \Cref{lemma:DFM-relation} that $\ze=0$, the upper bound on this edge is $\frac{\zf \zg}{\zf+\zg-\zh}$. 

Finally, we discuss the relationship between $\zi$ and $\zk$. 
\begin{enumerate}
    \item If $\zi>\zk$, consider the edge $\{(\hat{w},\beta\hat{y}+(1-\beta)z):\beta\in[0,1]\}$. From relation 5 and 7 of \Cref{lemma:DFM-relation}, we have $f_R(\hat{w},\hat{y})=0, f_R(\hat{w},z)=\zi, f_C(\hat{w},\hat{y})=\frac{1}{2}(\zk+\zl), f_C(\hat{w},z)=\zk$. The upper bound given by \Cref{prop:minmax-mixing} is $\frac{\zi(\zk+\zl)}{2\zi+\zl-\zk}$.

    \item Otherwise, consider the edge $\{(\alpha \hat{w}+(1-\alpha)w,z):\alpha\in[0,1]\}$, similarly, the upper bound is given by $\frac{\zg\zk-\zi\zh}{\zg+\zk-\zi-\zh}$.\qedhere
\end{enumerate}

\end{proof}

Then, we prove the following proposition to assist the proof.

\begin{proposition}[Min-max Properties]\label{prop:minmax}
Consider continuous functions $f:\R^n\rightarrow\R$, $g:\R^m\rightarrow\R$, $h_1,h_2:\R^n\times\R^m\rightarrow\R$, $f_1,f_2: \R\rightarrow \R$ and set $C\subseteq\R^n\times\R^m$. The following properties hold.

\begin{enumerate}
\item $\displaystyle\sup_{(a,b)\in C}\min\left\{h_1(a,b),h_2(a,b)\right\}\leq \sup_{\left\{a:\exists b \,(a,b)\in C\right\}}\min\left\{\sup_{\left\{b:(a,b)\in C\right\}} h_1(a,b),\sup_{\left\{b:(a,b)\in C\right\}} h_2(a,b)\right\}$.

\item  If $C=A\times B$ for some $A\in\R^n$ and $B\in\R^m$, then 
\[\sup_{(a,b)\in C}\min\left\{f(a),g(b)\right\}=\min\left\{\sup_{a\in A} f(a),\sup_{b\in B} g(b)\right\}.\]

\item If $C=A\times B$ for some $A\in\R^n$ and $B\in\R^m$, then 
\[\sup_{(a,b)\in C}\min\{f(a),g(b),h_1(a,b)\}=\sup_{a\in A}\min\left\{f(a),\sup_{b\in B}\min\left\{g(b),h_1(a,b)\right\}\right\}.\]

\item Suppose $f_1$ is non-decreasing on $[u,v]$ and $f_2$ is non-increasing on $[u,v]$. Then 
\[\max_{a\in [u,v]}\min\{f_1(a),f_2(a)\}=
\begin{cases} 
f_1(v), &f_1(v)<f_2(v),\\
f_2(u), &f_1(u)>f_2(u), \\ 
f_1(a^*),&\text{otherwise}. \end{cases}\]
Here, $a^*$ is the solution of $f_1(a^*)=f_2(a^*)$.

\item Suppose $f_1$ is non-decreasing on $[u,v]$ and $f_2$ is non-increasing on $[u,v]$. Suppose further that the solution of equation $f_1(a^*)=f_2(a^*)$, denoted by $a^*$, is in $[u,v]$. Consider interval $[u_0,v_0]\subseteq [u,v]$. Then $\max_{a\in [u_0,v_0]}\min\{f_1(a),f_2(a)\}\leq f_1(a^*)$.
\end{enumerate}
\end{proposition}
\begin{proof}
    For convenience, we write abbreviations LHS and RHS for the words ``left-hand side'' and ``right-hand side'', respectively. They are used to represent two sides of an equation in the proposition.
\begin{enumerate}[fullwidth, listparindent=\parindent]
\item Since $\LHS=\sup_{\left\{a:\exists b \,(a,b)\in C\right\}}\sup_{\left\{b:(a,b)\in C\right\}}\min\left\{h_1(a,b),h_2(a,b)\right\}$, it suffices to prove
the exchange between the second $\sup$ and $\min$ makes a greater value. Equivalently, we prove that for any functions $F_1,F_2$,
\[\sup_x\min\{F_1(x), F_2(x)\}\leq\min\left\{\sup_x F_1(x),\sup_x F_2(x)\right\}.\]
Clearly, for $i=1,2$, $\min\{F_1(x), F_2(x)\}\leq F_i(x)$. Thus $\sup_x\min\{F_1(x), F_2(x)\}\leq\sup_x F_i(x)$. By taking the minimum, the desired inequality then follows.

\item By the same argument of statement 1, $\LHS\leq\RHS$. To show the reverse inequality, pick any $\epsilon>0$. Let $a^*\in A$ such that $f(a^*)>\sup_{a\in A} f(a)-\epsilon$ and $b^*\in B$ such that $g(b^*)>\sup_{b\in B} g(b)-\epsilon$. Then $ \RHS<\min\{f(a^*),g(b^*)\}+\epsilon\leq\LHS+\epsilon$. Since $\epsilon$ is arbitrary, $\RHS\leq\LHS$.

\item By statement 1, $\LHS\leq\RHS$. The reverse inequality follows by the same argument in the proof of statement 2.

\item The proof is the same as that of \Cref{prop:minmax-mixing}. We here omit it.

\item It is a  direct corollary of statement 4.\qedhere
\end{enumerate}
\end{proof}

Finally, we prove that under the relations given by \Cref{lemma:DFM-relation}, 
\[
\min\left\{\rho(\zc-\zg),(1-\rho)(\zf-\zh),\rho \zj+(1-\rho)(\zl-\zk), \frac{\zf \zg}{\zf+\zg-\zh},D\right\}\leq \frac{1}{3},
\]
which finishes the proof. 

We first give a simple lemma which helps to deal with some details in our proof.
\begin{lemma}\label{lemma:sugarineq}
For any real numbers $a,b,u,v$ satisfying $u\geq 0$, $v\geq u$, and $u+b\geq 0$, we have $\frac{u+a}{u+b}\leq \frac{v+a}{v+b}$ holds if only if $a\leq b$.
\end{lemma}

\begin{proof}
Since $v+b\geq u+b\geq 0$, $\frac{u+a}{u+b}\leq \frac{v+a}{v+b}\iff (u-v)(a-b)\geq 0\iff a-b\leq 0$.
\end{proof}

Now we give the proof of the main theorem.

Due to symmetry, we can assume without loss of generality that $\zg\geq\zh$.

We divide the proof into two cases. The main idea of upper bound estimation is to gradually eliminate parameters by properly relaxing inequalities. The main structure of our discussion is as follows.
\[
\begin{cases}
\ref{sec:proof:minmax:case1}~\zi>\zk\quad
    \begin{cases}
    \ref{sec:proof:minmax:case1(a)}\quad
          \zi\leq 2\zg-1,\\
           \ref{sec:proof:minmax:case1(b)}\quad \zi>2\zg-1.
    \end{cases}\\[15pt]
    \ref{sec:proof:minmax:case2}~\zi\leq \zk\quad
    \begin{cases}
        \ref{sec:proof:minmax:case2(a)}\quad \Term 3=\Term 4,\\
        \ref{sec:proof:minmax:case2(b)}\quad \Term 2=\Term 4,\\
        \ref{sec:proof:minmax:case2(c)}\quad \Term 1=\Term 2.
    \end{cases}
\end{cases}
\]
\begin{enumerate}[label=\textbf{Case \arabic*}., fullwidth, listparindent=\parindent]
    \item $\zi> \zk$.\label{sec:proof:minmax:case1}
    By \Cref{lemma:1stpart}, we have 
\[f^*_{\mathcal{A}}\leq \min
\left\{\rho(\zc-\zg),(1-\rho)(\zf-\zh),\rho \zj+(1-\rho)(\zl-\zk), \frac{\zf \zg}{\zf+\zg-\zh},\frac{\zi(\zk+\zl)}{2\zi+\zl-\zk}\right\}.\]

We first eliminate $\zc$ and $\zf$. By relation 2 in \Cref{lemma:DFM-relation},  $\zc\leq 1$ and $\zf\leq 1$. Thus 
\[(1-\rho)(\zf-\zh)\leq (1-\rho)(1-\zh)\quad\text{and}\quad\rho(\zc-\zg)\leq \rho(1-\zg).\]
Since $\zg-\zh\geq 0$ and $\zf\leq 1$, by \Cref{lemma:sugarineq},  
\[\frac{\zf\cdot\zg}{\zf+\zg-\zh}\leq \frac{1\cdot\zg}{1+\zg-\zh}.\]
Now we eliminate $\zj$. By relation 6 in \Cref{lemma:DFM-relation}, $\zj\leq \zg-\zi$. Hence 
\[\rho \zj+(1-\rho)(\zl-\zk)\leq \rho (\zg-\zi)+(1-\rho)(\zl-\zk).\] 

After these eliminations, our goal is then to prove
\[\max_{\rho,\zg,\zh,\zi,\zl,\zk} \min\left\{\rho(1-\zg),(1-\rho)(1-\zh),\rho (\zg-\zi)+(1-\rho)(\zl-\zk),\frac{\zg}{1+\zg-\zh},\frac{\zi(\zk+\zl)}{2\zi+\zl-\zk}\right\}\leq 1/3.\]
We further divide the discussion into two cases.

\begin{enumerate}[fullwidth, listparindent=\parindent]
\item \label{sec:proof:minmax:case1(a)}
$\zi\leq 2\zg-1$.

We continue to eliminate parameters. First, we eliminate $\zi$, $\zk$ and $\zl$. For term $\rho(\zg-\zi)+(1-\rho)(\zl-\zk)$, we simply throw it. Since $\zi>\zk$, we have $\frac{\zk+\zl}{2\zi+\zl-\zk}\leq 1$. Therefore, 
\[\frac{\zi(\zk+\zl)}{2\zi+\zl-\zk}\leq \zi\leq2\zg-1.\] 

Then we eliminate $\rho$. By statement 4 in \Cref{prop:minmax}, 
\[\max_{\rho}\min\{\rho(1-\zg),(1-\rho)(1-\zh)\}=\frac{(1-\zg)(1-\zh)}{2-\zg-\zh}.\] 

Next, we eliminate $\zh$. Since $1-\zg\geq0$, we have $\frac{(1-\zh)}{2-\zg-\zh}\leq 1$, so 
\[\frac{(1-\zg)(1-\zh)}{2-\zg-\zh}\leq 1-\zg.\]

Now, by statement 3 and 4 in \Cref{prop:minmax}, $f^*_{\mathcal{A}}$ is bounded by
\begin{align*}
    &\max_{\rho,\zg,\zh}\min\{\rho(1-\zg),(1-\rho)(1-\zh),2\zg-1\} \\
    =&\max_{\zg,\zh}\min\{\max_\rho\{\rho(1-\zg),(1-\rho)(1-\zh)\},2\zg-1\}\\
    =&\max_{\zg,\zh}\min\left\{\frac{(1-\zg)(1-\zh)}{2-\zg-\zh},2\zg-1\right\}\\
    \leq&\max_{\zg}\min\{1-\zg,2\zg-1\}\\
    \leq&\frac{2}{3}(1-\zg)+\frac{1}{3}(2\zg-1)=1/3.
\end{align*}
The last equality holds when $\zg=2/3$.

\item \label{sec:proof:minmax:case1(b)}
$\zi>2\zg-1$.

We first eliminate $\zl$. By relation 2 in \Cref{lemma:DFM-relation}, $\zl\leq 1$. By \Cref{lemma:sugarineq}, we have 
\[\frac{\zi(\zk+\zl)}{2\zi+\zl-\zk}\leq \frac{\zi(\zk+1)}{2\zi+1-\zk}\quad\text{and}\quad\rho (\zg-\zi)+(1-\rho)(\zl-\zk)\leq \rho (\zg-\zi)+(1-\rho)(1-\zk).\]
Then we eliminate $\zh$. For term $(1-\rho)(1-\zh)$, we simply throw it. Since $\zg\geq\zh$, we have 
\[\frac{\zg}{1+\zg-\zh}\leq \zg.\]
    
Now our upper bound becomes
\[\max_{\rho,\zg,\zi,\zk}\min\left\{\rho(1-\zg),\rho (\zg-\zi)+(1-\rho)(1-\zk),\zg,\frac{\zi(\zk+1)}{2\zi+1-\zk}\right\}.\]

Then we eliminate $\rho$, which only occurs in the first two terms.  By assumption, $\zi>2\zg-1$, so $\zg-\zi< 1-\zg$. Since $\zi>\zk$ by assumption and $\zg\leq 1$ by Relation 2 in \Cref{lemma:DFM-relation}, $\zg-\zi<1-\zk$. When other parameters are fixed, the range of the first and the second terms are $[0,1-\zg]$ and $[\zg-\zi,1-\zk]$, which intersect. So by statement 3 and 4 in \Cref{prop:minmax}, the bound of $f^*_{\mathcal{A}}$ becomes
\begin{align*}
    &\max_{\zg,\zi,\zk}\min\left\{\max_{\rho\in[0,1]}\{\rho(1-\zg),\rho (\zg-\zi)+(1-\rho)(1-\zk)\},\zg,\frac{\zi(\zk+1)}{2\zi+1-\zk}\right\}\\
    =&\max_{\zg,\zi,\zk}\min\left\{\frac{(1-\zk)(1-\zg)}{2-\zk+\zi-2\zg},\zg,\frac{\zi(\zk+1)}{2\zi+1-\zk}\right\}.
\end{align*}

Note that none of the cases $\zk=1$, $\zg=1$, $\zg=0$ and $\zi=0$ can attain the maximum. So below we suppose without loss of generality that $\zk\in[0,1)$, $\zg\in(0,1)$ and $\zi\in(0,1]$. In the following discussion, this will avoid the discussion of corner cases.

Next, we eliminate $\zg$. By statement 1 in \Cref{prop:minmax}, the bound of $f^*_{\mathcal{A}}$ becomes
\[\max_{\zi,\zk}\min\left\{\max_\zg\min\left\{\frac{(1-\zk)(1-\zg)}{2-\zk+\zi-2\zg},\zg\right\},\frac{\zi(\zk+1)}{2\zi+1-\zk}\right\}.\]
We complete the elimination by calculating the maximum over $\zg$ on
\[G(\zg):=\min\left\{\frac{(1-\zk)(1-\zg)}{2-\zk+\zi-2\zg},\zg\right\}.\]

Note that $\zg$ is increasing on $[0,1]$ and $\frac{(1-\zk)(1-\zg)}{2-\zk+\zi-2\zg}$ is decreasing in $\zg$ on $[0,1]$. We then apply statement 4 in \Cref{prop:minmax}. To do so, we solve the following equation of $\zg^*$ under constraint $\zi>2\zg^*-1$ and $\zg^*\in[0,1]$, or equivalently, $\zg^*\in(0,(1+\zi)/2)$.
\[\frac{(1-\zk)(1-\zg^*)}{2-\zk+\zi-2\zg^*}=\zg^*.\]
The equation is equivalent to a quadratic equation
\[2(\zg^*)^2-(3-2\zk+\zi)\zg^*+1-\zk=0.\]
The solution candidates are
\[\zg^*=\frac{3-2\zk+\zi \pm \sqrt{(3-2\zk+\zi)^2-8(1-\zk)}}{4}.\] 
To verify the candidates are real, we have the following lower bound on the discriminant.
\begin{align*}
    (3-2\zk+\zi)^2-8(1-\zk)&\geq(3-2\zk)^2-8(1-\zk)\\
    &=9-12\zk+4\zk^2-8+8\zk\\
    &=4\zk^2-4\zk+1=(2\zk-1)^2\geq 0.
\end{align*}
Now we pick the proper solution of $\zg^*$. Note that
\begin{align*}
    &\frac{3-2\zk+\zi+\sqrt{(3-2\zk+\zi)^2-8(1-\zk)}}{4}<\frac{1+\zi}{2}\\ 
    \implies & (3-2\zk+\zi)^2-8(1-\zk)< (\zi+2\zk-1)^2\\
    \iff & 0<8(\zk-1)\zi.
\end{align*}
However, $\zk-1<0$ and $\zi>0$, so that is not a solution. For the other candidate, we have
\begin{align*}
    &\frac{3-2\zk+\zi-\sqrt{(3-2\zk+\zi)^2-8(1-\zk)}}{4}<\frac{1+\zi}{2}\\ 
    \iff & \sqrt{(3-2\zk+\zi)^2-8(1-\zk)}> 1-\zi-2\zk.
\end{align*}
If $\zi+2\zk-1\geq 0$, the inequality is trivially true. Suppose $\zi+2\zk-1<0$, then 
\begin{align*}
    &\sqrt{(3-2\zk+\zi)^2-8(1-\zk)}> 1-\zi-2\zk\\
    \iff & (3-2\zk+\zi)^2-8(1-\zk)> (1-2\zk-\zi)^2\\
    \iff & 4(1-\zk)\zi>0.
\end{align*}
The last inequality holds since $\zk<1$ and $\zi>0$. Clearly, this solution candidate is greater than zero. Thus the valid solution is
\[\zg^*=\frac{3-2\zk+\zi-\sqrt{(3-2\zk+\zi)^2-8(1-\zk)}}{4}.\] 

So we now obtain a new upper bound \[\max_{\zk,\zi}\min\left\{\frac{3-2\zk+\zi-\sqrt{(3-2\zk+\zi)^2-8(1-\zk)}}{4},\frac{\zi(\zk+1)}{2\zi+1-\zk}\right\}.\]

If $\zk+\zi\geq 2/3$, then we have 
\begin{align*}
    &\zk+\zi\geq 2/3\\
    \iff& \left(\frac{5}{3}-2\zk+\zi\right)^2\leq (3-2\zk+\zi)^2-8(1-\zk)\\
    \implies
    &\frac{3-2\zk+\zi-\sqrt{(3-2\zk+\zi)^2-8(1-\zk)}}{4}\leq \frac{1}{3}.
\end{align*}
Otherwise, $ \zk+\zi<2/3$. We also have
\begin{align*}
    &\zk+\zi<2/3\\
    \implies& 3\zi\zk+\zi+\zk< 1\\ 
    \iff&\frac{\zi(\zk+1)}{2\zi+1-\zk}< \frac{1}{3}.
\end{align*}
 Therefore, we have finished the proof of this case.
\end{enumerate}

\item \label{sec:proof:minmax:case2}
$\zk\geq \zi$.
In this case, by \Cref{lemma:1stpart}, we have 
\[f^*_{\mathcal{A}}\leq \min\left\{\rho(1-\zg),(1-\rho)(1-\zh),\rho \zj+(1-\rho)(\zl-\zk), \frac{\zf \zg}{\zf+\zg-\zh}, \frac{\zk\zg-\zi\zh}{\zk+\zg-\zi-\zh}\right\}.\]

First, we eliminate $\zl$ and $\zj$. By relation 2 and 6 in \Cref{lemma:DFM-relation}, $\zl\leq 1$ and $\zj\leq \zg-\zi$. Thus
\[
\rho \zj+(1-\rho)(\zl-\zk)\leq \rho (\zg-\zi)+(1-\rho)(1-\zk).
\]
Similarly, by relation 2 in \Cref{lemma:DFM-relation}, $\zf\leq 1$ and $\zc\leq 1$. So we can eliminate $\zc$ and $\zf$ as follows.
\begin{align*}
    (1-\rho)(\zf-\zh)\leq (1-\rho)(1-\zh)\quad\text{and}\quad \rho(\zc-\zg)\leq \rho(1-\zg).
\end{align*}

For term $\frac{\zf\zg}{\zf+\zg-\zh}$, we simply throw it.
It suffices to prove 
\[\max_{\rho,\zg,\zh,\zi,\zj,\zk} \min\left\{\rho(1-\zg),(1-\rho)(1-\zh),\rho (\zg-\zi)+(1-\rho)(1-\zk),\frac{\zk\zg-\zi\zh}{\zk+\zg-\zi-\zh}\right\}\leq 1/3.\]

In this case, we prove that, to bound $f^*_{\mathcal{A}}$ by $1/3$, it is feasible to assume that the max-min value is attained when the four terms in the minimum are equal. We discuss them case by case.

\begin{enumerate}[fullwidth, listparindent=\parindent]
\item \label{sec:proof:minmax:case2(a)}
$\Term 3 = \Term 4$.

Since the only nontrivial constraints on $\zk,\zi$ are  $\zk\geq\zi$ and $\zi\leq \zg$, we can present our eliminations in a two-stage form. In the first stage, we gradually increase $\zi$ and $\zk$ by the same incremental until $\zi=\zg$ or $\zk=1$. Then, at least one of two parameters $\zi$ and $\zk$ attain its maximum. In the second stage, we solely increase the other parameter which has not reached its maximum until it also reaches $\zi=\zg$ or $\zk=1$. 

During such a two-stage operation, $\Term 3$ is decreasing since its coefficients on $\zi$ and $\zk$ are negative. On the other hand, $\Term 4$ is increasing. To verify this fact, note that in the first stage, the denominator does not change, while the numerator increases. In the second stage, when $\zi=\zg$, by \Cref{lemma:sugarineq}, $\Term 4$ is increasing in $\zk$; when $\zk=1$, by the assumption that $\zg\geq\zh$, the derivative of $\zi$ is $\frac{(\zh-\zk)(\zh-\zg)}{(\zk+\zg-\zi-\zh)^2}\geq 0$, so $\Term 4$ is increasing in $\zi$. 

Now we use this operation to show that in the final elimination form, we must have $\Term 3 = \Term 4$. Suppose $\Term 3>\Term 4$, then $\Term 3>0$. However, when the two-stage operation is no longer possible, we must have $\zi=\zg$ and $\zk=1$, which makes $\Term 3=0$, a contradiction.

We then show that $\Term 3<\Term 4$ is also not possible. Similarly, we can do the procedure of decreasing $\zi$ and $\zk$ with the same value until $\zi=0$ or $\zk=0$, and then making possible decreasing on the left non-zero parameter to $0$. This will make $\Term 3$ increase and $\Term 4$ decrease as the discussion above. So if $\Term 3<\Term 4$, we can again adjust the value of $\zi$ and $\zk$ to ensure that they are equal and the result will not get worse.

\item  \label{sec:proof:minmax:case2(b)}
$\Term 2=\Term 4$.

$\Term 2$ is decreasing in $\zh$ since $\zh$ has a negative coefficient. On the other hand, $\Term 4$ is increasing in $\zh$. To see this, note that by assumption, $\zi\leq \zk$ and by Relation 6 in \Cref{lemma:DFM-relation}, $\zi\leq \zg-\zj\leq \zg$. Then derivative in $\zh$ is
$\frac{(\zi-\zk)(\zi-\zg)}{(\zk+\zg-\zi-\zh)^2}\geq 0$, so $\Term 4$ is increasing in $\zh$.

The valid value of $\zh$ ranges in $[0,\zg]$. Thus the range of $\Term 2$ is $[(1-\rho)(1-\zg),1-\rho]$ and the range of $\Term 4$ is $\left[\frac{\zk\zg}{\zk+\zg-\zi},\zg\right]$. By statement 4 in \Cref{prop:minmax}, $\Term 2=\Term 4$ if and only if the ranges intersect, or equivalently, $\zg\geq (1-\rho)(1-\zg)$ and $1-\rho\geq \frac{\zk\zg}{\zk+\zg-\zi}$.

Let us suppose otherwise. If $\zg< (1-\rho)(1-\zg)$, then by \Cref{prop:minmax},
\begin{align*}
    f^*_{\mathcal{A}}&\leq\max_{\rho,\zg} \min\{\rho(1-\zg),\zg\}\\
    &=\max_{\rho,\zg}\min\{\rho(1-\zg),(1-\rho)(1-\zg),\zg\}\\
    &=\max_\zg\min\left\{\max_\rho\min\{\rho(1-\zg),(1-\rho)(1-\zg)\},\zg \right\}\\
    &=\max_{\zg}\min\left\{\frac{1-\zg}{2},\zg\right\}=1/3.
\end{align*}

If $1-\rho< \frac{\zk\zg}{\zk+\zg-\zi}$, then since $\Term 2=(1-\rho)(1-\zh)\leq 1-\rho$, we may assume $\rho< 2/3$. Since $\Term 1=\rho(1-\zg)< 2/3(1-\zg)$, we may assume $\zg<1/2$. Therefore, 
\begin{align*}
    f^*_{\mathcal{A}}&\leq \max_{\rho,\zg,\zi,\zk}\min\{1-\rho, \rho(\zg-\zi)+(1-\rho)(1-\zk)\}\\
    &\leq \max_{\rho,\zi,\zk}\min\{1-\rho, \rho(\zk-\zi-1/2)+1-\zk\}.
\end{align*}
Also, by \Cref{lemma:sugarineq}, $1-\rho< \frac{\zk\zg}{\zk+\zg-\zi}\leq \frac{\zk\cdot(1/2)}{\zk+(1/2)-\zi}=\frac{\zk}{2\zk-2\zi+1}$. We divide the discussion into two cases: 

\begin{enumerate}
\item If $\zk\geq \zi+1/2$, then since the term $1-\rho$ is decreasing with $\rho$ while the term $\rho(\zk-\zi-1/2)+1-\zk$ is non-decreasing with $\rho$. By statement 5 in \Cref{prop:minmax}, we can ignore the range of $\rho$ constrained by other parameters, and the max-min value is no more than the case the two terms are equal. Solving the equation, we have $\rho=\frac{2\zk}{2\zk-2\zi+1}\geq 2(1-\rho)$, so $\rho\geq 2/3$, a contradiction! So this case cannot happen.

\item If $\zk<\zi+1/2$, then both terms in the max-min form is decreasing in $\rho$, so it suffices to consider the case $\rho$ obtains its minimum, which is defined by $1-\rho\leq \frac{\zk}{2\zk-2\zi+1}$. Substituting $\zk-\zi$ with $\zt$, we have $0<\zt<1/2$, and
\[f^*_{\mathcal{A}}\leq \max_{\zk,\zt}\min\left\{\frac{\zk}{2\zt+1}, \frac{2\zt (\zt-1/2)+(1-\zk)(3\zt+1/2)}{2\zt+1}\right\}.\]

Since the first term is increasing in $\zk$ and the second term is decreasing in $\zk$, the max-min value is not larger than the case the two terms are equal. Solving the equation about $\zk$, we have $\zk=\frac{1+2\zt}{3}$, or equivalently, $\frac{\zk}{2\zt+1}=1/3$, which is the desired result.
\end{enumerate}
With the discussion above, we can suppose $\Term 2=\Term 4$ as well.

\item \label{sec:proof:minmax:case2(c)}
$\Term 1= \Term 2$.

For $\Term 1$, since we have verified the feasibility of the assumption that the three other terms are equal, if $\Term 1$ is not equal to these terms, then it must be greater or less than other terms. However, $\Term 1$ is decreasing with $\rho$ and $\Term 2$ with increasing in $\rho$. The fact that their range must intersect at $0$ suggests that $\Term 1$ cannot be smaller or larger than $\Term 2$; otherwise, we can increase or decrease the value of $\rho$ to make $\Term 1=\Term 2$ but do not make the bound larger.
\end{enumerate}

After all the above discussions, now we finally proved that we can suppose the four terms in the minimum to be equal when the max-min value is attained. 

Now we solve the equation when the four terms are equal.

Solving $\rho(1-\zg)=(1-\rho)(1-\zh)$, we have $\rho=\frac{1-\zh}{2-\zg-\zh}$. Substituting this expression into the equation $\Term 1=\Term 3$, we have
\[\frac{(1-\zg)(1-\zh)}{2-\zg-\zh}=\frac{(1-\zg)(1-\zk)+(1-\zh)(\zg-\zi)}{2-\zg-\zh},\] 
namely 
\begin{equation}\label{eq:eq-on-k-1}
    (\zk - \zh) (1 - \zg)=(1-\zh)(\zg-\zi).
\end{equation}
Substituting this into $\Term 4,$ we have
\[\Term 4=\frac{\zk\zg-\zi\zh}{\zk+\zg-\zi-\zh}=\frac{(\zk\zg-\zi\zh)(1-\zg)}{(2-\zg-\zh)(\zg-\zi)}.\]
And by $\Term 1=\Term 4$, we have
\begin{equation}\label{eq:eq-on-k-2}
    (1-\zh)(\zg-\zi)=\zk\zg-\zi\zh.
\end{equation}
Combining \eqref{eq:eq-on-k-1} and \eqref{eq:eq-on-k-2}, we have 
\[\zk=\frac{\zg - \zg \zh - \zi + 2 \zh \zi}{\zg}=\frac{-\zg - \zh + 2 \zg \zh + \zi - \zh \zi}{-1 + \zg},\] 
which can be simplified to $(\zg-\zi)(3\zg\zh+1-2\zg-2\zh)=0$.

If $\zg=\zi$, we have 
\[f^*_{\mathcal{A}}\leq\max_{\rho,\zg,\zh,\zk}\min\{\rho(1-\zg),(1-\rho)(1-\zh),(1-\rho)(1-\zk),\zg\}.\] 
Since $\zg=\zi\leq \zk$, the right-hand side is bounded by
\[\max_{\rho,\zg}\min\{\rho(1-\zg),(1-\rho)(1-\zg),\zg\}=\max_\zg\min\left\{\frac{1-\zg}{2},\zg\right\}=1/3.\]

If $3\zg\zh+1-2\zg-2\zh=0$, we have $\zg =(-1 + 2 \zh)/(-2 + 3 \zh)$. Plugging this into the expression of Term 1, 
\[(x^*,y^*)\leq \max_{\rho,\zg}\Term 1=\frac{(1-\zg)(1-\zh)}{2-\zg-\zh}=1/3.\]
\end{enumerate}

Now we finish our discussion and show that $f^*_{\mathcal{A}}\leq 1/3$ in all cases.

\subsection{Proof of the tightness result in \Cref{ex-modifiedDFM}}\label{app:1/3-tight}
Now we show that \eqref{eq:tight-DFM} is a tight instance for \Cref{ex-modifiedDFM}. We need the following proposition to verify the validity of a stationary point and its dual solution.
\begin{proposition}[Proposition 3 in \cite{CDH+21_0.3393tight}]\label{prop:verify-sp}
strategy pair $(x_s,y_s)$ is a stationary point with dual solution $(\rho,w,z)$ if and only if
\begin{align*}
\supp (x_s) &\subseteq \suppmin \left\{-\rho R y_s+(1-\rho) C(z-y_s)\right\}\text{ and} \\
\supp(y_s) &\subseteq \suppmin\left\{\rho R^\T(w-x_s)-(1-\rho) C^\T x_s\right\}.
\end{align*}
\end{proposition}
Now we begin to prove the tightness. We verify the following three facts.
\begin{enumerate}[fullwidth, listparindent=\parindent]
\item $(x_s,y_s)$ is a stationary point and its dual solution is $(\rho,w,z)$.

We verify this by \Cref{prop:verify-sp}. One can check that 
\begin{gather*} 
\supp(x_s)=\supp(y_s)=\{1\},\\ 
\suppmin \left\{-\rho R y_s+(1-\rho) C(z-y_s)\right\}=\suppmin(1/6,1/6,1/6)=\{1,2,3\},\\
\suppmin\left\{\rho R^\T(w-x_s)-(1-\rho) C^\T x_s\right\}=\suppmin(1/6,1/6,1/6)=\{1,2,3\}.
\end{gather*}

\item $\hat{w}\in \br_{R}(\hat{y})$.

Since $\suppmax\{R\hat{y}\}= \suppmax (0,1/2,1/2)=\{2,3\}$, $\supp\{w\}\subseteq \suppmax\{R\hat{y}\}$, namely $\hat{w}\in \br_{R}(\hat{y})$.

\item The minimum of $f$ on $\mathcal{A}$ is $1/3$.

Due to our construction, it suffices to minimize $f$ over
\[\mathcal{A}:=\{(\alpha x_s+(1-\alpha)w,\beta y_s+(1-\beta)z):\alpha,\beta\in[0,1]\}.\]
We have
\begin{align*}
&\min_{\alpha\in[0,1],\beta\in[0,1]}\{f(\alpha x_s+(1-\alpha)w,\beta y_s+(1-\beta)z)\}\\
=&\min\left\{ \max\left\{1-\beta,\frac{2-\beta}{3}\right\}-\frac{(2-\beta)(1-\alpha)}{3},1-\frac{2\alpha}{3}-\frac{(1-\beta)(2-\alpha)}{3}\right\}.
\end{align*}

To prove the lower bound of $1/3$, we suppose on the contrary that
\begin{align*}
    \max\left\{1-\beta,\frac{2-\beta}{3}\right\}-\frac{(2-\beta)(1-\alpha)}{3}< \frac{1}{3}\quad\text{and}\quad 1-\frac{2\alpha}{3}-\frac{(1-\beta)(2-\alpha)}{3}<\frac{1}{3}. 
\end{align*}
Therefore, we have $\alpha\beta> 2\alpha-2\beta$, $\alpha\beta >2\alpha-1$ and $\alpha\beta > 2\beta-\alpha$.
 Simplifying them on $\alpha$, we have
 \[\alpha\in\left(\frac{2\beta}{\beta+1},\min\left\{\frac{1}{2-\beta},\frac{2\beta}{2-\beta}\right\}\right).\]
 If $\beta<1/2$, then such $\alpha$ exists if and only if $2-\beta < \beta+1$, namely $\beta> 1/2$, a contradiction. If $\beta\geq 1/2$, then such $\alpha$ exists if and only if $\frac{2\beta}{\beta+1}<\frac{1}{2-\beta}\iff(\beta-1)(2\beta-1)>0\implies\beta<1/2$, also a contradiction.
\end{enumerate}

%% file: section/append/detail_example.tex
\section{Automatic approximation analysis for algorithms in \Cref{tab:approx-result}}\label{details:literature-examples}

In this section, we provide details regarding the automatic approximation analysis for algorithms in \Cref{tab:instances}, including experiment setup and formalization of approximation bounds for all algorithms in \Cref{tab:instances}, which we utilize to derive the results presented in \Cref{tab:approx-result}.

\subsection{Experiment setup}

We run the programs on MacBook Pro, 13-inch, 2020. The CPU configuration is Intel(R) Core(TM) i5-1038NG7 CPU @ 2.00GHz. The memory is 16 GB 3733 MHz LPDDR4X. The macOS version is 13.5.2 (22G91). The power supply is using battery.\footnote{Actually, the performance of Mac OS is not affected by which power supply is used.}

We implement codes in Wolfram Mathematica, version 13.3.0.0. The codes are presented in notebook and run in local kernel.

All optimizers are using built-in function \texttt{NMaximize}. In the second case of BBM-$0.36$ \cite{BBM07_0.36NE}, it uses parameters as follows: 
\begin{itemize}
    \item \texttt{method -> "RandomSearch"}, \texttt{WorkingPrecision -> 10}, \texttt{MaxIterations -> 1000}.
\end{itemize}
In the case of TS-$0.3393$ \cite{TS07_0.3393NE} and DFM-$1/3$ \cite{DFM22_0.3333NE}, they use parameters as follows:
\begin{itemize}
    \item \texttt{AccuracyGoal -> 10}, \texttt{WorkingPrecision -> 20}, \texttt{Method -> "DifferentialEvolution"}.
\end{itemize}
For all other cases, we use default parameters.

\subsection{The KPS algorithm in \titlecite{KPS06_0.75NE}}
\begin{itemize}[fullwidth]
    \item \textbf{Search phase.}

    \begin{itemize}
    \item Row strategy: $e_{i_1}, e_{i_2} \in \Delta_m$.
    \item Column strategy: $e_{j_1}, e_{j_2}\in \Delta_n$. 
    \end{itemize}
    The mixing region is denoted by $\mathcal{A}=\{(\alpha e_{i_1}+(1-\alpha) e_{i_2}, \beta e_{j_1}+(1-\beta)e_{j_2}): \alpha\in [0,1], \beta\in [0,1]\}$.

    \item \textbf{Variables.}

    Denote $u_1=f_R(e_{i_1}, e_{j_1})$, $v_1=f_R(e_{i_1}, e_{j_2})$, $h_1=f_R(e_{i_2}, e_{j_1})$,
    $g_1=f_R(e_{i_2}, e_{j_2})$,
    $u_2=f_C(e_{i_1}, e_{j_1})$, $v_2=f_C(e_{i_1}, e_{j_2})$, $h_2=f_C(e_{i_2}, e_{j_1})$,
    $g_2=f_C(e_{i_2}, e_{j_2})$.
    \item \textbf{Relations.}

    In the KPS algorithm, $R_{i_1 j_1}=\max_{i,j} R_{ij}=1$ and $C_{i_2 j_2}=\max_{i,j} C_{ij}=1$. Thus, $u_1=g_2=0$.

    \item  \textbf{Upper bound.}

By taking $v=w=2$ in \Cref{prop:v-w-strong-aux-mixing}, the global minimum of $f$ over $\mathcal{A}$ is upper bounded by:
\begin{align*}
    \min_{\alpha,\beta} & \quad h\\
    \text{s.t.} &\quad h \geq \alpha \beta u_1+\alpha (1-\beta) v_1+(1-\alpha) \beta h_1+(1-\alpha) (1-\beta) g_1,\\
    &\quad h \geq \alpha \beta u_2+\alpha (1-\beta) v_2+(1-\alpha) \beta h_2+(1-\alpha) (1-\beta) g_2,\\
    &\quad 0\leq \alpha, \beta \leq 1.
\end{align*}
Thus, the approximation upper bound of the algorithm is given by:
\begin{align*}
\max_{u_1,v_1, h_1, g_1, u_2, v_2, h_2, g_2} \min_{\alpha,\beta} & \quad h\\
    \text{s.t.} &\quad h \geq \alpha \beta u_1+\alpha (1-\beta) v_1+(1-\alpha) \beta h_1+(1-\alpha) (1-\beta) g_1,\\
    &\quad h \geq \alpha \beta u_2+\alpha (1-\beta) v_2+(1-\alpha) \beta h_2+(1-\alpha) (1-\beta) g_2,\\
    &\quad 0\leq \alpha,\beta \leq 1,\\
    &\quad u_1=0, g_2=0,\\
    &\quad 0\leq v_1, h_1, g_1, u_2,v_2, h_2 \leq 1.
\end{align*}
\end{itemize}
\begin{remark}
It uses the strong $(2,2)$-constructor. Thus, it actually follows the extension framework as stated in \Cref{remark:stronger-upper-bound}.
\end{remark}

\subsection{The DMP algorithm in \titlecite{DMP06_0.5NE}}
Similar to the notation in \Cref{ex:1/2NE},
denote $a=f_R(i,j), b=f_R(k,j), c=f_C(i,j), d=f_C(k,j)$.

From \Cref{ex:1/2NE}, the expression of the approximation bound is given by:
\begin{align*}
\max_{a,b,c,d} \min_{\alpha} & \quad h\\
    \text{s.t.} &\quad h \geq \alpha a + (1-\alpha)b,\\
    &\quad h \geq \alpha c + (1-\alpha) d,\\
    &\quad 0\leq \alpha\leq 1,\\
    &\quad b=0, c=0,\\
    &\quad 0\leq a,d \leq 1.
\end{align*}

\subsection{The DMP algorithm in \titlecite{DMP07_0.38NE}}

\begin{itemize}[fullwidth]
    \item \textbf{Search phase.}
    \begin{itemize}
    \item Row strategy: $\alpha, x \in \Delta_m$.
    \item Column strategy: $\beta, y\in \Delta_n$. 
    \end{itemize}
    The mixing region is denoted by $\mathcal{A}=\{(\delta_1 \alpha+(1-\delta_1) x, \delta_2 \beta+(1-\delta_2)y): \delta_1\in [0,1], \delta_2 \in [0,1]\}$.

    \item \textbf{Variables.}

    Denote $u_1=f_R(\alpha, \beta)$, 
    $v_1=f_R(\alpha, y)$, 
    $h_1=f_R(x, \beta)$,
    $g_1=f_R(x, y)$,
    $u_2=f_C(\alpha, \beta)$, 
    $v_2=f_C(\alpha, y)$, 
    $h_2=f_C(x, \beta)$,
    $g_2=f_C(x, y)$.
    \item \textbf{Relations.}

    According to the DMP algorithm, there exists $v_R$, $v_C$ such that $0\leq u_1\leq 1+3\epsilon/2-v_R$, 
    $0\leq v_1\leq 2\epsilon$, 
    $0\leq h_1\leq 1+3\epsilon/2-v_R$,
    $0\leq g_1\leq v_R+\epsilon/2$,
    $0\leq u_2\leq 1+3\epsilon/2-v_C$, 
    $0\leq v_2\leq 1+3\epsilon/2-v_C$, 
    $0\leq h_2\leq 2\epsilon$,
    $0\leq g_2\leq v_C+\epsilon/2$.

    \item  \textbf{Upper bound.}
  
By taking $v=w=2$ in \Cref{prop:v-w-strong-aux-mixing}, the approximation upper bound of the algorithm is given by:
\begin{align*}
\max_{v_R, v_C} \max_{u_1,v_1, h_1, g_1, u_2, v_2, h_2, g_2} \min_{\delta_1,\delta_2} & \quad h\\
    \text{s.t.} &\quad h \geq \delta_1 \delta_2 u_1+\delta_1 (1-\delta_2) v_1+(1-\delta_1) \delta_2 h_1+(1-\delta_1) (1-\delta_2) g_1,\\
    &\quad h \geq \delta_1 \delta_2 u_2+\delta_1 (1-\delta_2) v_2+(1-\delta_1) \delta_2 h_2+(1-\delta_1) (1-\delta_2) g_2,\\
    &\quad 0\leq \delta_1,\delta_2 \leq 1,\\
    &\quad  0\leq u_1\leq 1+3\epsilon/2-v_R , 
     0\leq v_1\leq 2\epsilon , 
     0\leq h_1\leq 1+3\epsilon/2-v_R ,\\
      &\quad 0\leq g_1\leq v_R+\epsilon/2, 
   0\leq u_2\leq 1+3\epsilon/2-v_C , 
     0\leq v_2\leq 1+3\epsilon/2-v_C ,\\
      &\quad 0\leq h_2\leq 2\epsilon ,
     0\leq g_2\leq v_C+\epsilon/2,\\
     &\quad  0\leq v_R, v_C\leq 1.
\end{align*}
\end{itemize}

\begin{remark}
We note the solution to the optimization problem above is related to $v_R$ and $v_C$. Thus viewing $v_R$ and $v_C$ as parameters in the search phase. Also, it uses the strong $(2,2)$-constructor. Thus, it actually follows the extension framework as stated in \Cref{remark:0.36-modify-BBM} and \Cref{remark:stronger-upper-bound}.
\end{remark}

\subsection{The BBM algorithm in \titlecite{BBM07_0.36NE} with 0.382 approximation}
\label{details:ex:BBM0.38}
Similar to the notations of \Cref{example:start}, denote $g_1=f_R(x^*, y^*)$, $g_2=f_C(x^*, y^*)$, $h_1=f_R(x^*, b_2)$, $h_2=f_C(x^*, b_2)$, $v_1=f_R(r_1, b_2)$, $v_2=f_C(r_1, b_2)$, $u_1=f_R(r_1, y^*)$, $u_2=f_C(r_1, y^*)$.

From \Cref{example:approx-ana}, the expression of the approximation upper bound is given by:
\begin{align*}
\max_{g_1,g_2,h_1,h_2,u_1,u_2,v_1,v_2} & \quad \min_{\alpha_1,\alpha_2,\alpha_3,\alpha_4} \min\{s_1,s_2,s_3,s_4\}\\
    \text{s.t.} &\quad s_1 \geq \alpha_1 g_1 + (1-\alpha_1)u_1, s_1 \geq \alpha_1 g_2 + (1-\alpha_1)u_2,\\
    &\quad s_2 \geq \alpha_2 g_1 + (1-\alpha_2)h_1, s_2 \geq \alpha_2 g_2 + (1-\alpha_2)h_2,\\
    &\quad s_3 \geq \alpha_3 v_1 + (1-\alpha_3)h_1, s_3 \geq \alpha_3 v_2 + (1-\alpha_3)h_2,\\
    &\quad s_4 \geq \alpha_4 v_1 + (1-\alpha_4)u_1, s_4 \geq \alpha_4 v_2 + (1-\alpha_4)u_2,\\
    &\quad 0\leq \alpha_1, \alpha_2, \alpha_3, \alpha_4 \leq 1,\\
    &\quad u_1=v_2=0,
    0\leq g_1,g_2,h_1,h_2,u_2,v_1\leq 1,
    g_1\geq g_2,
    u_2\leq 1-g_1.
\end{align*}

\subsection{The CDFFJS algorithm in \titlecite{CDF+16_0.382NE&0.653WSNE}}
\begin{itemize}[fullwidth]
    \item \textbf{Search phase.}
    \begin{itemize}
    \item Row strategy: $x^*,\hat{x},r \in \Delta_m$.
    \item Column strategy: $y^*, j\in \Delta_n$. 
    \end{itemize}
    The mixing region is denoted by $\mathcal{A}=\{(\alpha \hat{x}+(1-\alpha) x^*, \beta \hat{y}+(1-\beta) y^*: \alpha \in [0,1], \beta \in [0,1]\}$.

    \item \textbf{Variables.}

    Denote $u_1=f_R( \hat{x}, j)$, 
    $v_1=f_R( \hat{x}, y^*)$, 
    $h_1=f_R(x^*, j)$,
    $g_1=f_R(x^*, y^*)$,
    $s_1=f_R(r,  j)$,
    $t_1=f_R(r, y^*)$,
    $u_2=f_C( \hat{x},  j)$, 
    $v_2=f_C( \hat{x}, y^*)$, 
    $h_2=f_C(x^*,  j)$,
    $g_2=f_C(x^*, y^*)$,
    $s_2=f_C(r, j)$,
    $t_2=f_C(r, y^*)$.
    
    \item \textbf{Relations.}

    According to the CDFFJS algorithm, denote $v_C={\hat{x}}^\T C\hat y$ and $v_R=(x^*)^\T Ry^*$. Assume $v_R\geq v_C$. From \cite{CDF+16_0.382NE&0.653WSNE}, we have $0\leq v_1\leq v_R, 0\leq v_2\leq v_C$, $g_1=h_2=s_1=0$, $0\leq h_1\leq 1-v_R$. Besides, we have $0\leq u_1,g_1, t_1, u_2, g_2, s_2, t_2\leq 1$.

    \item  \textbf{Upper bound.}
  
By applying \Cref{prop:v-w-aux-mixing}, the approximation upper bound of the algorithm is given by:
\begin{align*}
\max_{v_R,v_C}\max_{ v_1,u_1,g_1,h_1,t_1,s_1,v_2,u_2,g_2,h_2,t_2,s_2} & \quad \min_{ \alpha_1,\alpha_2,\alpha_3,\alpha_4,\alpha_5,\alpha_6,\alpha_7,\alpha_8,\alpha_9 } \min\{ b_1,b_2,b_3,b_4,b_5,b_6,b_7,b_8,b_9 \}\\
\text{s.t.}&\quad b_1\geq \alpha_1 u_1 +(1-\alpha_1) v_1,b_1\geq \alpha_1 u_2 +(1-\alpha_1) v_2,\\
&\quad b_2\geq \alpha_2 h_1 +(1-\alpha_2) g_1,b_2\geq \alpha_2 h_2 +(1-\alpha_2) g_2,\\
&\quad b_3\geq \alpha_3 s_1 +(1-\alpha_3) t_1,b_3\geq \alpha_3 s_2 +(1-\alpha_3) t_2,\\
&\quad b_4\geq \alpha_4 g_1 +(1-\alpha_4) v_1,b_4\geq \alpha_4 g_2 +(1-\alpha_4) v_2,\\
&\quad b_5\geq \alpha_5 t_1 +(1-\alpha_5) v_1,b_5\geq \alpha_5 t_2 +(1-\alpha_5) v_2,\\
&\quad b_6\geq \alpha_6 t_1 +(1-\alpha_6) g_1,b_6\geq \alpha_6 t_2 +(1-\alpha_6) g_2,\\
&\quad b_7\geq \alpha_7 h_1 +(1-\alpha_7) u_1,b_7\geq \alpha_7 h_2 +(1-\alpha_7) u_2,\\
&\quad b_8\geq \alpha_8 s_1 +(1-\alpha_8) u_1,b_8\geq \alpha_8 s_2 +(1-\alpha_8) u_2,\\
&\quad b_9\geq \alpha_9 s_1 +(1-\alpha_9) h_1,b_9\geq \alpha_9 s_2 +(1-\alpha_9) h_2,\\
&\quad 0\leq \alpha_1,\alpha_2,\alpha_3,\alpha_4,\alpha_5,\alpha_6,\alpha_7,\alpha_8,\alpha_9 \leq 1,\\
&\quad 0\leq v_1\leq v_R, 0\leq v_2\leq v_C, g_1=h_2=s_1=0,\\
&\quad 0\leq h_1\leq 1-v_R, 0\leq u_1, t_1, u_2, g_2, s_2, t_2\leq 1,\\
&\quad 0\leq v_C\leq v_R\leq 1.
\end{align*}
\end{itemize}

\begin{remark}
We note the solution to the optimization problem above is related to $v_R, v_C$. Thus, it actually follows the extension framework as stated in \Cref{remark:0.36-modify-BBM}. 
\end{remark}

\subsection{The BBM algorithm in \titlecite{BBM07_0.36NE} with 0.36 approximation}
\begin{itemize}[fullwidth]
    \item \textbf{Search phase.}
    \begin{itemize}
    \item Row strategy:  $\hat{x}=(1-\delta_1) x^*+\delta_1 r_1 \in \Delta_m$.
    \item Column strategy: $y^*, b_2\in \Delta_n$. 
    \end{itemize}
    The mixing region is denoted by $\mathcal{A}=\{(\hat{x}, \delta_2 y^*+(1-\delta_2) b_2: \delta_2 \in [0,1]\}$.

    \item \textbf{Variables.}

    Denote 
    $a=f_R(\hat{x},y^*), b=f_R(\hat{x},b_2), c=f_C(\hat{x},y^*), d=f_C(\hat{x},b_2)$.
    
    \item \textbf{Relations.}

    According to the BBM algorithm, denote $g_1=f_R(x^*,y^*), g_2=f_C(x^*,y^*), h_2=f_C(x^*,b_2)$. For symmetry, we suppose $g_1\geq g_2$, and 
    \[\delta_1= \begin{cases}0, & \text { if } g_1 \in[0,1 / 3], \\ \left(1-g_1\right)\left(-1+\sqrt{1+\frac{1}{1-2 g_1}-\frac{1}{g_1}}\right), & \text { if } g_1 \in(1 / 3, \beta], \\ 1, & \text { otherwise.}\end{cases}\]
    where $\beta$ is the root of the polynomial $x^3-x^2-2x+1$ in $[1/3,1/2]$. From \cite{BBM07_0.36NE}, we have $0\leq a \leq (1-\delta_1)g_1$, $0\leq b\leq 1-(1-\delta_1)h_2$, $c\leq (1-\delta_1)h_2+\delta_1 (1-g_1)$, $d=0$. Moreover, we have
    \[\begin{cases}a=g_1, c=g_2, & \text { if } g_1 \in(1 / 3, \beta], \\ h_2=g_2, & \text { if }g_1\in[\beta,1].\end{cases}\]

    \item  \textbf{Upper bound.}
    
By taking $w=2$ in \Cref{prop:1-w-aux-mixing}, the approximation upper bound of the algorithm is given by:
\begin{align*}
\max_{g_1,h_2,\delta_1} \max_{a,b,c,d} \min_{\delta_2} & \quad h\\
    \text{s.t.} &\quad h \geq \delta_2 a + (1- \delta_2)b,\\
    &\quad h \geq  \delta_2 c + (1- \delta_2) d,\\
    &\quad 0\leq \delta_2 \leq 1,\\
    &\quad0\leq a \leq 1-(1-\delta_1)h_2, 0\leq b\leq (1-\delta_1)g_1,\\
    &\quad c=0, d\leq (1-\delta_1)h_2+\delta_1 (1-g_1),\\
    & \quad 0\leq g_1, h_2\leq 1, g_1\geq g_2, g_2\geq h_2,\\
    & \quad \begin{cases}\delta_1=0, & \text { if } g_1 \in[0,1 / 3], \\ \delta_1=\left(1-g_1\right)\left(-1+\sqrt{1+\frac{1}{1-2 g_1}-\frac{1}{g_1}}\right), a=g_1, c=g_2, & \text { if } g_1 \in(1 / 3, \beta], \\ \delta_1=1, h_2=g_2, & \text { otherwise.}\end{cases}
\end{align*}

\end{itemize}
\begin{remark}
We note the solution to the optimization problem above is related to $g_1, h_2$. Thus, it actually follows the extension framework as stated in \Cref{remark:0.36-modify-BBM}. 
\end{remark}

\subsection{The TS algorithm in \titlecite{TS07_0.3393NE}}\label{app:subsec:TS-formalization}

Similar to the notation in \Cref{ex-TSalgo},  denote $\za=f_R(x_s,y_s)$, $\zb=f_C(x_s,y_s)$, $\zc=f_R(x_s,z)$, $\zd=f_C(x_s,z)$, $\ze=f_R(w,y_s)$, $\zf=f_C(w,y_s)$, $\zg=f_R(w,z)$, $\zh=f_C(w,z)$.

From \Cref{ex-TSalgo}, the expression of the approximation upper bound is given by:
\begin{align*}
\max_{\za,\zb,\zc,\zd,\ze,\zf,\zg,\zh, \rho} & \quad \min_{\alpha_1,\alpha_2,\alpha_3,\alpha_4} \min\{b_1,b_2,b_3,b_4\}\\
    \text{s.t.} &\quad b_1 \geq \alpha_1 \za + (1-\alpha_1)\ze, b_1 \geq \alpha_1 \zb + (1-\alpha_1)\zf,\\
    &\quad b_2 \geq \alpha_2 \za + (1-\alpha_2)\zc, b_2 \geq \alpha_2 \zb + (1-\alpha_2)\zd,\\
    &\quad b_3 \geq \alpha_3 \zg + (1-\alpha_3)\zc, b_3 \geq \alpha_3 \zh + (1-\alpha_3)\zd,\\
    &\quad b_4 \geq \alpha_4 \zg + (1-\alpha_4)\ze, b_4 \geq \alpha_4 \zh + (1-\alpha_4)\zf,\\
    &\quad 0\leq \alpha_1, \alpha_2, \alpha_3, \alpha_4 \leq 1,\\
    &\quad \za=\zb, \zd=\ze=0, \zc\geq \zg, \zc\geq \za, \zf\geq \zh, \zf\geq \zb, \zg\geq \zh, \za\leq \rho(\zc-\zg), \za\leq (1-\rho)(\zf-\zh),\\
    &\quad 0\leq \za, \zb, \zc, \zd, \ze, \zf, \zg, \zh, \rho \leq 1.
\end{align*}

\subsection{The DFM algorithm in \titlecite{DFM22_0.3333NE}}\label{app:subsec:DFM-formalization}
Due to \Cref{ex-modifiedDFM}, the expression of the approximation upper bound is given by:
\begin{align*}
\max_{ \za,\zb,\zc,\zf,\zg,\zh,\zi,\zj,\zk,\zl, \zn, \zm } & \quad \min_{ \alpha_{1},\alpha_{2},\alpha_{3},\alpha_{4},\alpha_{5},\alpha_{6},\alpha_{7},\alpha_{8},\alpha_{9},\alpha_{10},\alpha_{11},\alpha_{12},\alpha_{13},\alpha_{14},\alpha_{15} } \min\{ b_i \}_{i=1}^{15}\\
\text{s.t.}&\quad b_{1}\geq \alpha_{1} \za +(1-\alpha_{1}) \zn,b_{1}\geq \alpha_{1} \zb +(1-\alpha_{1}) \zb/2,\\
&\quad b_{2}\geq \alpha_{2} \zn +(1-\alpha_{2}) \zc,b_{2}\geq \alpha_{2} \zb/2 +(1-\alpha_{2}) \zd,\\
&\quad b_{3}\geq \alpha_{3} \ze +(1-\alpha_{3}) \zm,b_{3}\geq \alpha_{3} \zf +(1-\alpha_{3}) (\zf+\zh)/2,\\
&\quad b_{4}\geq \alpha_{4} \zm +(1-\alpha_{4}) \zg,b_{4}\geq \alpha_{4} (\zf+\zh)/2 +(1-\alpha_{4}) \zh,\\
&\quad b_{5}\geq \alpha_{5} \zj ,b_{5}\geq \alpha_{5} \zl +(1-\alpha_{5}) (\zk+\zl)/2,\\
&\quad b_{6}\geq  (1-\alpha_{6}) \zi,b_{6}\geq \alpha_{6} (\zk+\zl)/2 +(1-\alpha_{6}) \zk,\\
&\quad b_{7}\geq \alpha_{7} \za +(1-\alpha_{7}) \ze,b_{7}\geq \alpha_{7} \zb +(1-\alpha_{7}) \zf,\\
&\quad b_{8}\geq \alpha_{8} \ze +(1-\alpha_{8}) \zj,b_{8}\geq \alpha_{8} \zf +(1-\alpha_{8}) \zl,\\
&\quad b_{9}\geq \alpha_{9} \za +(1-\alpha_{9}) \zj,b_{9}\geq \alpha_{9} \zb +(1-\alpha_{9}) \zl,\\
&\quad b_{10}\geq \alpha_{10} \zn +(1-\alpha_{10}) \zm,b_{10}\geq \alpha_{10} \zb/2 +(1-\alpha_{10}) (\zf+\zh)/2,\\
&\quad b_{11}\geq \alpha_{11} \zm ,b_{11}\geq \alpha_{11} (\zf+\zh)/2 +(1-\alpha_{11}) (\zk+\zl)/2,\\
&\quad b_{12}\geq \alpha_{12} \zn ,b_{12}\geq \alpha_{12} \zb/2 +(1-\alpha_{12}) (\zk+\zl)/2,\\
&\quad b_{13}\geq \alpha_{13} \zc +(1-\alpha_{13}) \zg,b_{13}\geq \alpha_{13} \zd +(1-\alpha_{13}) \zh,\\
&\quad b_{14}\geq \alpha_{14} \zg +(1-\alpha_{14}) \zi,b_{14}\geq \alpha_{14} \zh +(1-\alpha_{14}) \zk,\\
&\quad b_{15}\geq \alpha_{15} \zc +(1-\alpha_{15}) \zi,b_{15}\geq \alpha_{15} \zd +(1-\alpha_{15}) \zk,\\
&\quad 0\leq \alpha_{1},\alpha_{2},\alpha_{3},\alpha_{4},\alpha_{5},\alpha_{6},\alpha_{7},\alpha_{8},\alpha_{9},\alpha_{10},\alpha_{11},\alpha_{12},\alpha_{13},\alpha_{14},\alpha_{15} \leq 1,\\
&\quad \zg\geq \zh, \zd=\ze=0, \za=\zb,\\
&\quad 0\leq  \za,\zb,\zc,\zf,\zg,\zh,\zi,\zj,\zk,\zl, \zn, \zm \leq 1,\\
&\quad \zc\geq \zg, \zg\geq \zm, \zc\geq \zn\geq \za, \zf\geq \zh, \zf\geq \zb,\\
&\quad \za\leq \rho(\zc-\zg), \zb\leq (1-\rho)(\zf-\zh), \za\leq \rho \zj+(1-\rho)(\zl-\zk), \zg\geq \zi+\zj.
\end{align*}

%% file: section/append/def_geo.tex
\section{Definitions in discrete geometry}\label{app:def-geo}

Below are the definitions of several concepts in discrete geometry: 
\begin{definition}
The concepts below are either from \cite{M13_GTMdiscretegeo} or basic concepts in linear algebra. We append the location of the concepts from \cite{M13_GTMdiscretegeo} for further interests.
\begin{enumerate}[fullwidth]
    \item (Affine Space, P1, Section 1.1) An affine space is a displacement of a vector space. It has the form of $v+V=\{v+x:x\in V\}$, where $v$ is a vector and $V$ is a vector space. Equivalently, an affine space can be expressed as $\left\{x\in\R^n: u_i^\T  x=v_i, i\in [m]\right\}$ for some $u_i\in\R^n\setminus\{0\}$ and $v_i\in\R$, $i\in[m]$. The dimension of an affine space $V+v$ is defined to be the dimension of $V$.
    
    \item (Affine Hull, P1, Section 1.1) The affine hull of a set $S\in \R^n$, denoted by $\aff(S)$, is the minimal affine space containing it.
    
    \item (Dimension, P83, Section 5.2) The dimension of a set $S\in \R^n$, denoted by $\dim(S)$, is given by the dimension of its affine hull.
    
    \item (Hyperplane, P3, Section 1.1) In any linear space $H$, a hyperplane is an affine subspace whose dimension is one less than that of $H$.
    
    \item (Half-space, P3, Section 1.1) In $\R^n$, a half-space is the set $\{x\in\R^n:a^\T x\leq b\}$ or $\{x\in\R^n:a^\T x<b\}$, where $a$ is a nonzero vector in $\R^n$ and $b\in\R$.
    
    \item (Polytope, P82, Section 5.2) A convex polytope $S$ is defined as a bounded set that is the intersection of finitely many half-spaces, namely $S=\{x\in\R^n:Ax\leq b\}$, where $A\in\R^{k\times n}$ has no zero rows and $b\in\R^k$.
    
    \item (Face, P86, Section 5.3) A face of a polytope $S$ is defined by the set $\{x\in S:\forall y\in S ,a^\T  x\leq a^\T  y\}$ for a certain $a\in \R^n$. Note that by setting $a=0$, we have $S$ itself as a face. By definition, every face of a polytope is also a polytope.

    \item (Facet, P87, Section 5.3) A facet of polytope $S$ is a face of dimension exactly $\dim(S)-1$.

    \item (Boundary) The boundary of a face $S$ is defined as the set of points $x\in S$ such that for any $\epsilon>0$, there exists $y\in \aff(S)\setminus S$ satisfying $\|y-x\|<\epsilon$. We denote the boundary of $S$ as $\partial S$. The interior of $S$, denoted as $S^\circ$, is defined as $S\setminus\partial S$.

    \item (Line and segment) In $\R^n$, a line (segment) is a set of the form $\{y\in\R^n:y=td+b, t\in I\}$, where $d\in\R^n$ is a nonzero vector, $b\in\R^n$, and $I=\R$ ($I=[u,v]$). It represents a one-dimensional affine space. The vector $d$ is called the direction of the line.

    \item (Parallel lines) Two lines (segments) are said to be parallel if their directions $d_1$ and $d_2$ satisfy $d_1=kd_2$ for some nonzero real number $k$. The relation of being parallel is an equivalence class, and two lines are equivalent if and only if they share a proportional direction. Hence, we can also represent a line (segment) using its direction vector, which we will use directly in reference to a line below.

    \item (Parallel between lines and polytopes) For an affine space $A\subseteq \R^n$, we say that a vector $d$ is parallel to $A$ if $A$ contains a line parallel to $d$. Equivalently, if $A=\left\{x\in\R^n:u_i^\T x=v_i, i\in [m]\right\}$, then $d\parallel A$ if and only if $u_i^\T d=0$ holds for each $i\in [m]$. For a vector $d$ and a polytope $P\subseteq \R^n$, we say that $d$ is parallel to $P$ if $d\parallel \aff(P)$.
\end{enumerate}
\end{definition}